\documentclass[reqno,12pt]{amsart}

\usepackage{amsmath,amssymb,amsthm,amsfonts,amsxtra}
\usepackage{fullpage}
\usepackage[utf8]{inputenc}
\usepackage[english]{babel}
\usepackage[usenames, dvipsnames]{color}
\usepackage{verbatim}
\numberwithin{equation}{section} 

\newtheorem{thm}{Theorem}{\bf}{\em}
{\bf}{\em}
\newtheorem{prop}{Proposition}{\bf}{\em}
\newtheorem{lem}{Lemma}{\bf}{\em}
\newtheorem{rem}{Remark}{\bf}{\em}
{\bf}{\em}

\def\pe{\perp}
\def\pa{\parallel}

\def\p{\partial}

\def\mb{\mathbf}
\def\mr{\mathrm}
\def\mk{\mathfrak}
\def\bs{\boldsymbol}

\def\beq{\begin{equation}}
\def\eeq{\end{equation}}

\def\bpm{\begin{pmatrix}}
\def\epm{\end{pmatrix}}

\def\const{{\rm const.}}

\def\Rnum{\mathbb{R}}
\def\Cnum{\mathbb{C}}

\DeclareMathOperator{\re}{Re}
\DeclareMathOperator{\im}{Im}

\DeclareMathOperator{\tr}{tr}

\DeclareMathOperator{\ad}{ad}
\DeclareMathOperator{\Ad}{Ad}
\DeclareMathOperator{\End}{End}

\def\t{{\rm t}}

\def\inv{{}^{-1}}

\def\id{{\rm id}}

\def\diff{\mb{d}}
\def\Dx{D_x}
\def\Dxinv{D_x^{-1}}

\def\gsp{\mk{g}}
\def\msp{\mk{m}}
\def\hsp{\mk{h}}

\def\g{\mathrm{g}}

\def\aff{\mathrm{aff}}

\def\Kill#1{K(#1)}

\def\i{\mr{i}}

\def\e{\mr{e}}
\def\efac{\chi^2}
\def\enorm{\tfrac{1}{\chi}}
\def\enorminv{\chi}

\def\brack#1{\langle #1\rangle}
\def\wedgebar{{\,\bar\wedge\,}}

\def\odotbar{{\,\bar\odot\,}}

\def\Pop{{\mathcal P}}

\def\covder{\nabla^M}
\def\Gcovder{\nabla^G}
\def\Hcovder{\nabla^H}

\def\Mop{{\mathcal M}}

\def\curvflow{{\vec\gamma}}
\def\grad{\vec\nabla}
\def\P{P}
\def\Pvec{\vec{\P}}

\def\map{\Psi}

\def\cofra{e}

\def\conx{\omega}
\def\hook{\rfloor}

\def\comm#1{{\bs [}#1{\bs ]}}

\def\T{\mathbf{T}}
\def\N{\mathbf{N}}
\def\B{\mathbf{B}}

\def\q{q}
\def\Q{Q}

\def\qq{\mr{q}}
\def\qqbar{\mr{\bar q}}
\def\hpe{\mr{h}_\pe}
\def\hpebar{\mr{\bar h}_\pe}

\def\hhpa{{\mb h}_\pa}

\def\wpe{\mr{w}_\pe}
\def\wpebar{\mr{\bar w}_\pe}

\def\wwpa{{\mb w}_\pa}

\def\lrep{{\bs(}}
\def\rrep{{\bs)}}

\def\AA{\mb{A}}

\def\a{\mr{a}}

\def\abar{\bar\a}

\def\b{\mr{b}}
\def\bbar{\bar\b}

\def\d{\mr{d}}

\def\s{\mr{s}}
\def\sbar{\bar\s}
\def\S{\mb{S}}

\def\Rop{{\mathcal R}}

\def\Jop{{\mathcal J}}
\def\Hop{{\mathcal H}}
\def\Eop{{\mathcal E}}
\def\Dop{{\mathcal D}}

\def\X{{\rm X}}

\def\secref#1{Sec.~\ref{#1}}

\def\Ref#1{Ref.~\cite{#1}}
\def\Refs#1{Refs.~\cite{#1}}

\def\eg/{e.g.}
\def\ie/{i.e.}

\begin{document}
\allowdisplaybreaks[3]
\tolerance =99999

\title{Hasimoto variables, generalized vortex filament equations, Heisenberg models and Schr\"odinger maps arising from group-invariant NLS systems}

\begin{abstract}
The deep geometrical relationships holding among
the NLS equation, the vortex filament equation,
the Heisenberg spin model, and the Schr\"odinger map equation
are extended to the general setting of Hermitian symmetric spaces.
New results are obtained by utilizing a generalized Hasimoto variable
which arises from applying the general theory of parallel moving frames. 
The example of complex projective space $\Cnum P^N= SU(N+1)/U(N)$
is used to illustrate the method and results. 
\end{abstract}

\author{
Stephen C. Anco$^1$
\lowercase{\scshape{and}}
Esmaeel Asadi$^{2}$ 
\\\\\lowercase{\scshape{
${}^1$Department of Mathematics and Statistics\\
Brock University\\
St. Catharines, ON L2S3A1, Canada}} \\\\
\lowercase{\scshape{
${}^2$Department of Mathematics\\
Institute for Advance Studies in Basic Science (IASBS)\\
45137--66731, Zanjan, Iran}}
}

\maketitle

\begin{center}
email: 
${}^1$sanco@brocku.ca,
${}^2$easadi@iasbs.ac.ir, ${}^2$esmaeel.asadi@gmail.com
\end{center}

\section{Introduction }

The classical  nonlinear Schr\"{o}dinger (NLS) equation  
$i u_t=u_{xx}+\tfrac{1}{2}|u|^2u$ for a complex-valued function $u(x,t)$
is well-known to arise in numerous physical applications
and has a deep connection to differential geometry,
particularly symmetric spaces. 

On one hand,
from the viewpoint of isospectral theory in symmetric Lie algebras, 
the NLS equation is given by a linear isospectral flow 
in the Lie algebra $\mk{su}(2)\simeq\mk{so}(3)$.
This provides a natural Lax pair for the NLS equation,
which is directly related to the AKNS scheme \cite{AblKauNewSeg}
yielding a hierarchy of isospectral equations and integrable flows in $\mk{su}(2)\subset \mk{sl}(2,\Cnum)$. 
The NLS isospectral flow is also associated with a geometric curve flow 
$\gamma_t = [\gamma_s,\gamma_{ss}]$ in the Lie algebra $\mk{su}(2)$
which is given by the Sym-Pohlmeyer \cite{Sym,Poh1976} construction. 
Constant phase rotations on $u$, which are a symmetry of the NLS equation,
correspond to a $U(1)$ gauge group of transformations 
generated by a $\mk{u}(1)$ Cartan subalgebra in $\mk{su}(2)$. 
This is directly related to the structure of 
$\Cnum P^1= SU(2)/U(1) \simeq S^2$ as a Hermitian symmetric space. 

On the other hand,
from the viewpoint of geometric integrability theory in Riemannian symmetric spaces,
the NLS equation arises as the evolution equation for the curvature and torsion $(\kappa,\tau)$
of a non-stretching curve flow $\vec{\gamma}(x,t)$ in Euclidean space $\Rnum^3$.
The motion of the curve is given by the bi-normal equation 
$\vec{\gamma}_t = \kappa \T\times\N = \vec{\gamma}_x\times\vec{\gamma}_{xx}$
where $\T=\vec{\gamma}_x$ is the unit tangent vector,
$\N$ is the normal vector, and $\B=\T\times\N$ is the bi-normal vector
comprising a Frenet frame $(\T,\N,\B)$ along the curve.
Note that $\Rnum^3$ has the natural structure of a Riemannian symmetric space $(SO(3)\rtimes\Rnum^3)/SO(3)$
in which the $SO(3)$ is the structure group of orthonormal frames. 
In this formulation \cite{BefWanSan}, 
$(\kappa,\tau)=(|u|,\arg(u)_x)$ are invariants of the curve,
while $u=\kappa e^{i\int\tau\;dx}$ has the geometrical meaning of a $U(1)$-covariant of the curve, 
with $(\re u,\im u)$ being the components of the Cartan matrix of a parallel frame \cite{Bis}
given by the tangent vector $\T$
and the pair of normal vectors $\re(e^{i\int\tau\;dx}(\N+i\B))$, $\im(e^{i\int\tau\;dx}(\N+i\B))$. 
The phase-rotation symmetry on $u$ corresponds to a $U(1)\simeq SO(2)$
gauge group of transformations in $SO(3)$ 
preserving the structure of a parallel frame by rigid rotations in the normal plane along the curve.

As shown by Hasimoto \cite{Has},
the bi-normal equation also describes the physical motion of a vortex filament in fluid mechanics.
Furthermore, the bi-normal equation can be interpreted as a Heisenberg spin model in $\Rnum^3$, 
or equivalently as a Schr\"{o}dinger map equation on the sphere $S^2$.
The geometrical interrelationships among these physical and mathematical formulations of the NLS equation 
have been explored in Ref.\cite{AncMyr2010}. 
Abstract formulations involving differential invariants, Hamiltonian structures, and loop groups 
can be found in \Refs{TerUhl1999,Bef1998,Bef2001}. 

A generalization of linear isospectral flows
yielding group invariant multi-component integrable NLS systems
is known \cite{ForKul1983}
for all Lie algebras associated with Hermitian symmetric spaces
(namely, Riemannian symmetric spaces with a compatible complex structure). 
These isospectral flows,
called Fordy-Kulish systems,
each have an associated Sym-Pohlmeyer curve flow
which can be viewed as a Lie-algebra version of a bi-normal equation \cite{LanPer2000}. 

A generalization of parallel frames and Hasimoto variables
yielding group-invariant integrable systems 
is known \cite{Anc2008} for all Riemannian symmetric spaces.
This includes \cite{Anc2007} semi-simple Lie groups,
viewed in a natural way as diagonal Riemannian symmetric spaces. 
In particular,
the results in Ref.\cite{Anc2007}
show that the NLS equation arises directly from
a parallel frame formulation of a non-stretching geometric curve flow given by
a chiral Schr\"{o}dinger map equation in the Lie group $SU(2)$,
which is a variant of the Sym-Pohlmeyer curve flow in the Lie algebra $\mk{su}(2)$. 
This type of geometrical derivation of group-invariant integrable systems
and corresponding non-stretching map equations
has been pursued  \cite{SanWan,Anc2006,AncAsa2009,AncAsa2012,AncAsaDog2015,AhmAncAsa2018}
for all of the basic classical (non-exceptional) Riemannian symmetric spaces, 
yielding many types of multi-component mKdV and sine-Gordon systems,
as well as multi-component NLS systems \cite{Anc2007},
formulated in terms of Hasimoto variables. 

In the present paper,
we adapt the general results on parallel frames and Hasimoto variables
from \Ref{Anc2008}
to extend the geometrical relationships among
the NLS equation, the vortex filament equation,
the Heisenberg spin model, and the Schr\"odinger map equation
to the general setting of Hermitian symmetric spaces.
Several main results are obtained.

We give a new geometrical derivation of
the Fordy-Kulish isospectral NLS flows \cite{ForKul1983}, 
together with the associated Sym-Pohlmeyer curve flows \cite{LanPer2000}, 
in all Lie algebras that arise from Hermitian symmetric spaces. 
The derivation provides an explicit bi-Hamiltonian formulation
for the NLS flows and the Sym-Pohlmeyer curve flows. 
We also show that the underlying geometrical curve flows can be formulated as 
bi-normal equations which can be viewed as vortex filament equations
in affine symmetric spaces defined from the Lie algebras.

Most interestingly,
we derive generalized Heisenberg spin vector models and Schr\"odinger map equations
from the vortex filament equations in affine symmetric spaces, 
and we obtain their explicit bi-Hamiltonian structure. 
The Schr\"odinger map equations describe a geometric version of the NLS equation
for a map into a Hermitian symmetric space, 
which is generalization of the Grassmannian spaces studied in \Ref{TerUhl1999}. 
We also show that the equations corresponding to axial motion of a vortex filament 
in affine symmetric spaces 
correspond to a higher-order Heisenberg spin vector model and an mKdV analog of a Schr\"odinger map equation,
which have a simple bi-Hamiltonian structure.  
The explicit form of the Hamiltonian structures of these geometric map equations
is shown to involve only the intrinsic Hermitian and Riemannian structure of the Hermitian symmetric space.
Furthermore, 
we obtain an explicit recursion operator arising from Magri's theorem \cite{Mag}
applied to this bi-Hamiltonian structure, 
which yields a hierarchy of higher order geometric map equations. 
We show how the resulting hierarchy of Hamiltonians leads, 
through the inverse recursion operator, 
to a hierarchy of Hamiltonian forms for the Schr\"odinger map equations
and their mKdV analogs.
We also derive a new non-standard Hamiltonian structure for these equations. 

In addition,
we show how to obtain an explicit matrix representation for 
the Schr\"odinger map variable,
the spin vector variable,
and the Hasimoto variable.
The matrix representation provides a Backlund transformation relating
the isospectral flows, vortex filament equations, Heisenberg spin vector models, and Schr\"odinger map equations.
This allows the integrability properties of any one of these nonlinear PDE systems to be transferred to any other. 

All of these new results,
which make essential use of the generalized Hasimoto variable
in Lie algebras associated with Hermitian symmetric spaces, 
will be of interest for applications and further developments
in both mathematics and physics.

The basic geometric setting for this work will be to formulate 
a Lie algebra $\gsp$ that is associated with a general Hermitian symmetric space 
as being an affine semidirect-product symmetric space
$\gsp\simeq (G\rtimes\gsp)/G$, 
analogously to the treatment of semi-simple Lie groups in \Ref{Anc2007}. 
We develop a general theory and apply it to the Lie algebra
$\mk{su}(N+1)\simeq \Cnum^{N}\oplus \mk{u}(N)$ 
corresponding to the complex projective space $\Cnum P^N= SU(N+1)/U(N)$.
This example will provide a higher-dimensional version of
the vortex filament equation, the Heisenberg spin model,
and the Schr\"odinger map equation,
since in the case $N=1$,
we have
$\gsp=\mk{su}(2) \simeq \mk{so}(3) \simeq (SO(3)\rtimes \Rnum^3)/SO(3)$, 
with $\Rnum^3\simeq \mk{so}(3)$.
The NLS and mKdV isospectral flows for this example reproduce 
the vector NLS system and vector mKdV system derived in \Ref{Anc2007}
from the Lie group $SU(N+1)$, 
with the Hasimoto variable being a complex vector in $\Cnum^N$. 

The rest of the present paper is organized as follows.

\secref{sec:FKmethod} reviews the theory of linear isospectral flows
and the resulting integrable Fordy-Kulish NLS systems 
in Lie algebras associated with Hermitian symmetric spaces. 
For later comparison,
the recursion operator and its factorization into Hamiltonian operators
generating higher isospectral flows, in particular integrable mKdV systems,
are presented in an explicit form.
A short summary of the Sym-Pohlmeyer construction of associated curve flows is also given.

\secref{sec:geomderivation} presents the new geometrical derivation of this isospectral hierarchy. 
As a first step, 
the geometrical and algebraic structure defining an affine Hermitian symmetric space is stated.
Next, the formulation of parallel frames and Hasimoto variables
in Riemannian symmetric spaces is adapted to affine Hermitian symmetric spaces.
This formulation is then applied to non-stretching curve flows,
where the flow is shown to induce a bi-Hamiltonian isospectral equation on the Hasimoto variable. 
By considering geometric flows generated by
a translation symmetry and a Hermitian (phase) symmetry of this bi-Hamiltonian structure,
the resulting flows on the Hasimoto variable are shown to yield exactly the NLS hierarchy of Fordy-Kulish isospectral flows.
The relationship between the parallel frame formulation and the Sym-Pohlmeyer construction is explained.
It is emphasized that bi-Hamiltonian structure and its associated recursion operator
arise in a natural geometric fashion in the parallel frame formulation,
whereas the Hamiltonian structures are hidden in the Sym-Pohlmeyer construction.

\secref{sec:VFE} first shows how to obtain the explicit form of
the geometric non-stretching curve flows. 
For the Fordy-Kulish NLS systems and associated mKdV systems,
these curve flows are shown to be given by a generalized version of, respectively, 
the vortex filament equation
and the equation describing axial motion of a vortex filament. 
These equations are formulated geometrically in terms of
the tangent vector, normal vector, and a generalized bi-normal vector
of the underlying curves in affine Hermitian symmetric spaces.
Next, the vortex filament equation is shown to yield a general Heisenberg spin model,
while the equation for axial motion of a vortex filament is shown to give rise to a related 
higher-order spin vector model,
where the dynamical variable in these models is given by the tangent vector of the curve.
In addition,
a geometric recursion operator is derived which generates a hierarchy of
non-stretching curve flows and associated spin vector models.
This result establishes that each geometric curve flow and associated spin vector model
possesses an infinite hierarchy of higher symmetries. 

\secref{sec:biHam} extends these main results
by showing that the geometric curve flows and associated spin vector models 
have an explicit bi-Hamiltonian structure which arises from a factorization of
the geometric recursion operator.
A new non-standard Hamiltonian structure for these equations is also obtained. 
As a consequence, all of the geometric curve flows and spin vector models
are tri-Hamiltonian integrable systems. 

\secref{sec:geommap} derives, for all Hermitian symmetric spaces, 
the general Schr\"odinger map equation and its mKdV analog 
from, respectively, the general Heisenberg spin model and its higher-order counterpart 
in affine Hermitian symmetric spaces. 
These geometric map equations are shown to represent curve flows 
that are locally stretching but whose total arclength is a constant of motion. 
The main integrability features of these geometric map equations are presented, 
including their recursion operator and multi-Hamiltonian structure 
in an explicit geometrical form.
In addition,
the explicit matrix representations and Backlund transformations
for the general Schr\"odinger map equation
and the associated general Heisenberg spin vector model and vortex filament equation
are stated. 

\secref{sec:example} applies the preceding results to the example of
the Hermitian symmetric space $\Cnum P^N= SU(N+1)/S(U(N)\times U(1))=SU(N+1)/U(N)$ 
and the associated Lie algebra $\mk{su}(N+1)\simeq \Cnum^{N}\oplus \mk{u}(N)$. 
The correspondence between the isospectral NLS flows
and their associated Heisenberg spin models and Schr\:odinger map equations
is shown in an explicit form. 
The necessary algebraic structure is summarized in an appendix. 

Concluding remarks are given in \secref{sec:conclude}.

\section{Fordy-Kulish NLS hierarchy}\label{sec:FKmethod}

We will begin by summarizing
the theory of isospectral flows in Lie algebras associated with Hermitian symmetric spaces,
as developed in \Ref{ForKul1983}. 

Recall \cite{Hel},
a compact semi-simple Lie group $G$ with subgroup $H$ 
such that the quotient of Lie algebras $\msp = \gsp/\hsp$ has a complex structure 
defines a Hermitian symmetric space $G/H$. 
This is equivalent to stating that $\gsp=\msp\oplus\hsp$ 
is a Hermitian symmetric Lie algebra. 

There is a well-known classification of Riemannian and Hermitian symmetric spaces
\cite{Hel}.
These spaces divide up into classical and exceptional types. 
The corresponding irreducible Hermitian symmetric Lie algebras of classical type
are shown in Table~\ref{table:hermsymm}. 

\begin{table}[htb]
\centering
\caption{Non-exceptional Hermitian symmetric Lie algebras}
\label{table:hermsymm}
\begin{tabular}{l|c|c|c}
\hline
Type & $\gsp$ & $\hsp$ & $\msp$
\\
\hline\hline
A III
& $\mk{su}(n+m)$
& $\mk{s}(\mk{u}(n)\oplus\mk{u}(m))$
& $\Rnum^{2nm}$
\\
D III
& $\mk{so}(2n)$
& $\mk{u}(n)$
& $\Rnum^{n(n-1)}$
\\
BD I
& $\mk{so}(n+2)$
& $\mk{so}(n)\oplus\mk{so}(2)$
& $\Rnum^{2n}$
\\
C I
& $\mk{sp}(n)$
& $\mk{u}(n)$
& $\Rnum^{n(n+1)}$
\\
\hline
\end{tabular}
\end{table}

The Lie bracket structure of a Hermitian symmetric Lie algebra $\gsp=\msp\oplus\hsp$
is given by 
\begin{equation}\label{symmsp.brackets}
[\hsp,\hsp]\subseteq\hsp,
\quad
[\hsp,\msp]\subseteq\msp,
\quad
[\msp,\msp]\subseteq\hsp,
\end{equation}  
while the Hermitian structure consists of an element $A\in\hsp$ such that 
its centralizer in $\gsp$ is $\hsp$,
and it acts as an imaginary-unit on $\msp$:
\begin{align}
& C_{\gsp}(A)=\hsp , 
\label{centrA}
\\
&  \ad_\msp(A)^2 = -\id_\msp . 
\label{adA}
\end{align}
Thus, $\ad(A)\msp = \msp$ and $\ad(A)\hsp =0$. 
The element $A$ has the squared norm
\begin{equation}
\Kill{A,A}=\tr(\ad_\msp(A)^2) = -\dim(\msp)
\end{equation}  
in the Killing inner product on $\gsp$. 

In the example of 
$\gsp=\mk{su}(2)=\begin{pmatrix} ia & b+ic \\ -b +ic & -ia \end{pmatrix}$,
$a,b,c\in\Rnum$, 
the Hermitian structure is given by 
$A=\begin{pmatrix} \tfrac{1}{2}i & 0 \\ 0& -\tfrac{1}{2}i \end{pmatrix}\in 
\hsp=\begin{pmatrix} ia & 0 \\ 0& -ia \end{pmatrix}$,
where $\ad_\msp(A)=J$ acts as a complex structure on $\msp=\begin{pmatrix} 0 & b+ic \\ -b +ic & 0 \end{pmatrix}$ by 
$J\msp = \begin{pmatrix} 0 & i(b+ic) \\ -i(-b +ic) & 0 \end{pmatrix}
= \begin{pmatrix} 0 & -c+ib \\ c+ib & 0 \end{pmatrix}$. 
Here $\hsp\simeq \mk{u}(1)$ and $\msp\simeq \Rnum^2\simeq \Cnum$.

\subsection{Isospectral flows}

The theory starts with a linear overdetermined system 
\begin{align}\label{spectraleq}
\phi_s=U\phi,
\quad
\phi_t=V\phi,
\end{align}
in which $\phi=\phi(s,t,\lambda)$ takes values in the Lie group $G$,
while $U(s,t,\lambda)$ and $V(s,t,\lambda)$ take values in the Lie algebra $\gsp$. 
Compatibility of equations \eqref{spectraleq} yields a Lax pair
formulated as a zero-curvature equation
\begin{equation}\label{ZCeq}
U_t-V_s+[U,V]=0 . 
\end{equation}

Choose $U=U(s,t,\lambda)$ to be a linear polynomial in $\lambda$, 
\begin{equation}\label{U.akns}
U=\lambda A+\Q, 
\end{equation}
where $\Q=\Q(s,t)$, called the potential, sits in $\msp\subset\gsp$. 
Choose $V=V(s,t,\lambda)$ to be a Laurent polynomial in $\lambda$, 
\begin{equation}\label{V.akns}
V=\sum_k \lambda^k V^{(k)} . 
\end{equation}
Then for the zero-curvature equation \eqref{ZCeq} to hold with $\lambda$ being arbitrary,
the coefficient of each separate power of $\lambda$ must vanish.
This describes a generalization of the AKNS scheme \cite{AblKauNewSeg}
from $\mk{sl}(2,\Cnum)$ to $\gsp$.

To derive the NLS isospectral flow,
take $V$ to be a quadratic polynomial, 
\begin{equation}\label{V.quadr}
V=V^{(2)}\lambda^2+V^{(1)}\lambda+V^{(0)} . 
\end{equation}
The coefficients of $\lambda$ in the zero-curvature equation \eqref{ZCeq} 
projected into $\hsp$ and $\msp$ yield the system
\begin{align}
& D_s V^{(0)}_{\hsp}=[\Q,V^{(0)}_{\msp}],
\label{ZCh0}\\
& D_s V^{(1)}_{\hsp}=[\Q,V^{(1)}_{\msp}],
\label{ZCh1}\\
& D_s V^{(2)}_{\hsp}=[\Q,V^{(2)}_{\msp}],
\label{ZCh2}
\end{align}
and
\begin{align} 
& \Q_t=\Dx V^{(0)}_{\msp}-[\Q,V^{(0)}_{\hsp}],
\label{ZCm0}\\
& D_s V^{(1)}_{\msp}-[\Q,V^{(1)}_{\hsp}]-[A,V^{(0)}_{\msp}]=0,
\label{ZCm1}\\
& D_sV^{(2)}_{\msp}-[\Q,V^{(2)}_{\hsp}]-[A,V^{(1)}_{\msp}]=0,
\label{ZCm2}\\
& [A,V^{(2)}_{\msp}]=0 ,
\label{ZCm3}
\end{align}
for $V^{(k)}_{\msp}$, $V^{(k)}_{\hsp}$, $k=0,1,2$. 
Equation \eqref{ZCm3} leads to $V^{(2)}_{\msp}=0$ by property \eqref{centrA}.
As a consequence,
equation \eqref{ZCh2} can be integrated,
which gives 
\begin{equation}
V^{(2)}_{\hsp}=c_2 A,
\quad
c_2=\const . 
\end{equation}
Next, equation \eqref{ZCm2} becomes $[A,c_2\Q-V^{(1)}_{\msp}]=0$,
whence
\begin{equation}
V^{(1)}_{\msp}=c_2\Q
\end{equation}
by property \eqref{centrA}. 
Equation \eqref{ZCh1} then can be integrated,
yielding 
\begin{equation}
V^{(1)}_{\hsp}=c_1 A,
\quad
c_1=\const  . 
\end{equation}
Substitution of these expressions into equation \eqref{ZCm1} gives
$c_2\Q_s = [A,V^{(0)}_{\msp}-c_1\Q]$.
Hence
\begin{equation}
V^{(0)}_{\msp}=c_1\Q -c_2\ad(A)\Q_s
\end{equation}
by properties \eqref{centrA} and \eqref{adA}. 
Next, equation \eqref{ZCh0} becomes
$D_s V^{(0)}_{\hsp}=-c_2[\Q,\ad(A)\Q_s]$
which simplifies to $D_s V^{(0)}_{\hsp}=-\tfrac{1}{2}c_2[\Q,\ad(A)\Q]_s$
through the Jacobi identity $[\Q,\ad(A)\Q]_s = 2[\Q,\ad(A)\Q_s] + \ad(A)[\Q_s,\Q]$
combined with $\ad(A)[\Q_s,\Q]=0$ due to $[\Q_s,\Q]\in\hsp$. 
Thus, 
\begin{equation}
V^{(0)}_{\hsp}=c_0A -\tfrac{1}{2}c_2[\Q,\ad(A)\Q]_s
\end{equation}
with the integration constant chosen in the same way as before. 

Finally, the evolution equation \eqref{ZCm0} reduces to 
\begin{equation}\label{nls.flow}
\Q_t= c_0 \ad(A)\Q + c_1\Q_s -c_2( \ad(A)\Q_{ss} -\tfrac{1}{2}[\Q,[\Q,\ad(A)\Q]] ) . 
\end{equation}
This equation can be cleaned up as follows. 
Write
\begin{equation}
\ad(A) = J
\end{equation}
where $J^2=-\id$ on $\msp$. 
Now apply a Galilean transformation $t\to t$, $s\to s+c_1 t$
combined with a phase transformation $\Q\to \exp(c_0 \ad(A))\Q$,
and use a scaling $t\to \tfrac{1}{c_2} t$.
As a result, equation \eqref{nls.flow} becomes
\begin{equation}\label{FK.nls} 
\Q_t = -J\Q_{ss} +\tfrac{1}{2}[\Q,[\Q,J\Q]] , 
\end{equation}
which is called the Fordy-Kulish NLS system \cite{ForKul1983}.
Note that it comes from
$V^{(0)}=-\tfrac{1}{2}[\Q,J\Q]_s - J\Q_s$, $V^{(1)}=\Q$, $V^{(2)} =A$. 

This NLS system is integrable.
It has the Lax pair
\begin{equation}
U=\lambda A + \Q,
\quad
V=-\tfrac{1}{2}[\Q,J\Q]_s -J\Q_s +\lambda\Q +\lambda^2 A , 
\end{equation}
and corresponding isospectral equation
\begin{equation}
\phi_s = (\lambda A + \Q)\phi,
\quad
\phi_t = (-\tfrac{1}{2}[\Q,J\Q]_s - J\Q_s + \lambda\Q +\lambda^2 A)\phi . 
\end{equation}

\subsection{An isospectral hierarchy}

For any polynomial $V=\sum_{k=0}^{n} \lambda^k V^{(k)}$,
the zero-curvature equation \eqref{ZCeq} splits into a coupled system
that contains equations \eqref{ZCh0}, \eqref{ZCh1}, \eqref{ZCm0}, \eqref{ZCm1}.
These four equations turn out to encode a recursion operator.

First,
eliminate $V^{(0)}_{\hsp}$ from equation \eqref{ZCm0} via equation \eqref{ZCh0}
to get 
\begin{equation}
\Q_t =\Hop V^{(0)}_{\msp}
\end{equation}
where
\begin{equation}
\Hop=D_s -\ad(\Q)D_s^{-1}\ad(\Q) . 
\end{equation}
Next,
eliminate $V^{(1)}_{\hsp}$ from equation \eqref{ZCm1} via equation \eqref{ZCh1}
to get 
\begin{equation}
V^{(0)}_{\msp} =\Rop^* V^{(1)}_{\msp} 
\end{equation}
where
\begin{equation}
\Rop^*=J\inv(D_s -\ad(\Q)D_s^{-1}\ad(\Q))= J\inv\Hop
\end{equation}
is the adjoint of
\begin{equation}\label{iso.Rop}
\Rop=(D_s -\ad(\Q)D_s^{-1}\ad(\Q))J\inv = \Hop J\inv . 
\end{equation}

\begin{thm}\label{thm:iso.flow.hierarchy}
If $V$ is polynomial of degree $n\geq 2$ in $\lambda$,
then the corresponding isospectral flow is given by 
\begin{equation}\label{hierarchy.isoflow}
\Q_t=\Hop{\Rop^*}^{n-1}(\Q) =\Rop^{n}(J\Q) 
\end{equation}
called the $+n$ flow,
where $\Rop= \Hop J\inv$ is a hereditary recursion operator.
Moreover, $\Hop$ and $J$ are compatible Hamiltonian operators. 
Each of these operators, along with all of the isospectral flows, 
are invariant under the unitary group $Ad(H)$. 
\end{thm}

A geometrical derivation of $\Hop$ and $\Rop$ will be given in \secref{sec:geomderivation} 
where the proof of their properties will be discussed.
Moreover, the potential $\Q$ will be seen to be a Hasimoto variable,
which is related to the principal curvature of non-stretching geometric curve flows 
associated with the isospectral flows by the Sym-Pohlmeyer construction,
as shown in \secref{sec:VFE}. 
This will also provide a geometric origin for the unitary invariance of
the isospectral flows and the Hamiltonian operators. 

In the hierarchy \eqref{hierarchy.isoflow},
the $0$ flow is $\Q_t = J\Q$,
and the $+1$ flow is $\Q_t = \Q_s$,
while the $+2$ flow is the Fordy-Kulish NLS system \eqref{FK.nls}.
The $+3$ flow can be shown to be given by the mKdV system
\begin{equation}\label{FK.mkdv}
\Q_t=-\Q_{sss} +[\Q,[\Q,\Q_s]] +\tfrac{1}{2}[J\Q,[J\Q,\Q]]_s
\end{equation}
which is derived in \Ref{AthFor1987}.

As is well known,
negative flows arise by taking $V$ to be a Laurent polynomial in $\lambda$.
The first negative flow is given by 
\begin{equation}
U=\lambda A+\Q,
\quad
V=\lambda^{-1} V^{(-1)},
\end{equation}
which leads to the system 
\begin{equation}
\Q_t=-[A,V^{(-1)}_{\msp}],
\quad
D_s V^{(-1)}_{\msp}=[\Q,V^{(-1)}_{\hsp}],
\quad
D_s V^{(-1)}_{\hsp}=[\Q,V^{(-1)}_{\msp}] . 
\end{equation}
This system is equivalent to the flow equation
\begin{equation}
\Rop(\Q_t)=0
\end{equation}  
called the $-1$ flow.

\subsection{Sym-Pohlmeyer curves}

Each isospectral flow in the hierarchy \eqref{hierarchy.isoflow}
gives rise to an associated curve flow via the Sym-Pohlmeyer construction
starting from the linear isospectral system \eqref{spectraleq}
for $\phi(s,t,\lambda)$.
The curve is defined by
\begin{equation}
\gamma(s,t) = \phi\inv \phi_\lambda\big|_{\lambda=0}
\end{equation}
which lies in the Lie algebra $\gsp$.
It is straightforward to see that the tangent vector and the flow vector of $\gamma$
are given by
\begin{equation}\label{SP.tangent.flow.vectors}
\gamma_s = \phi\inv U_\lambda \phi\big|_{\lambda=0} = \phi_0\inv A\phi_0, 
\quad
\gamma_t = \phi\inv V_\lambda \phi\big|_{\lambda=0} = \phi_0\inv V^{(1)}\phi_0,
\end{equation}
where $\phi_0=\phi|_{\lambda=0}$.
Moreover,
the tangent vector $\gamma_s$ satisfies
$\Kill{\gamma_s,\gamma_s} = \Kill{\phi_0\inv A\phi_0,\phi_0\inv A\phi_0}
= \Kill{A,A} = -\dim(\msp)$.
Hence the arclength parameter of the curve $\gamma$
is related to $s$ by
\begin{equation}\label{s}
x=\sqrt{\dim(\msp)}\, s
\end{equation}
with $\Kill{\gamma_x,\gamma_x} = -1$. 
Furthermore,
$\gamma_{ts} = (\phi_0\inv A\phi_0)_t = \phi_0\inv[A,V^{(0)}]\phi_0 = \phi_0\inv J V^{(0)}\phi_0$,
and thus
$\Kill{\gamma_s,\gamma_{ts}} = \Kill{\phi_0\inv A\phi_0,\phi_0\inv J V^{(0)}\phi_0}
= \Kill{A,J V^{(0)}} = -\Kill{JA,V^{(0)}} = 0$.
This shows that the arclength $x$ is locally preserved in the curve flow,
\begin{equation}
\p_t \Kill{\gamma_x,\gamma_x} = 0 . 
\end{equation}

For a given isospectral flow,
$V^{(1)}$ will be a function of $\Q$ and $s$-derivatives of $\Q$,
which can be expressed in terms of $s$-derivatives of $\gamma_s$
through the following relations:
\begin{align}
\gamma_{ss} & = (\phi_0\inv A\phi_0)_s = \phi_0\inv[A,\Q]\phi_0 = \phi_0\inv J\Q\phi_0,
\label{SP.rel1}
\\
[\gamma_s,\gamma_{ss}] & = [\phi_0\inv A\phi_0,\phi_0\inv J\Q\phi_0]= \phi_0\inv[A,J\Q]\phi_0 = \phi_0\inv J^2\Q\phi_0 = -\phi_0\inv\Q\phi_0,
\label{SP.rel2}
\end{align}
using $\phi_{0\,s} = q \phi_0$ from the linear isospectral system \eqref{spectraleq}.
Then $\gamma_t = \phi_0\inv V^{(1)}\phi_0$ will yield an explicit equation of motion
for the curve flow. 

In particular,
the NLS isospectral flow \eqref{FK.nls} has $V^{(1)}=\Q$,
which yields 
$\gamma_t = \phi_0\inv \Q\phi_0 = - [\gamma_s,\gamma_{ss}]$.
Hence,
the Sym-Pohlmeyer curve flow associated with the Fordy-Kulish NLS system \eqref{FK.nls}
is given by \cite{LanPer2000}
\begin{equation}\label{SP.nls.curveflow}
\gamma_t = -[\gamma_s,\gamma_{ss}]
\end{equation}  
in the Hermitian symmetric Lie algebra $\gsp$.

The mKdV isospectral flow \eqref{FK.mkdv} can be seen to have
$V^{(1)}=-J\Q_s  -\tfrac{1}{2} [\Q,J\Q]$.
This yields
$\gamma_t = -\phi_0\inv(J\Q_s +\tfrac{1}{2} [\Q,J\Q])\phi_0$. 
From relations \eqref{SP.rel1}--\eqref{SP.rel1},
note
$[\gamma_{ss},[\gamma_{ss},\gamma_s]]   = \phi_0\inv [JQ,[JQ,A]]\phi_0 = \phi_0\inv [JQ,Q]\phi_0$
and
$\gamma_{sss} = \phi_0\inv( J\Q_s + [J\Q,Q])\phi_0$. 
Hence,
$\gamma_{sss} -\tfrac{3}{2}[\gamma_{ss},[\gamma_{ss},\gamma_s]]   
=\phi_0\inv(J\Q_s +\tfrac{1}{2} [\Q,J\Q])\phi_0$,
and so the Sym-Pohlmeyer curve flow associated with the mKdV system \eqref{FK.mkdv}
is given by \cite{LanPer2000}
\begin{equation}\label{SP.mkdv.curveflow}
\gamma_t = -\gamma_{sss} +\tfrac{3}{2}[\gamma_{ss},[\gamma_{ss},\gamma_s]]
\end{equation}  
in the Hermitian symmetric Lie algebra $\gsp$.

\section{Geometrical derivation of the isospectral hierarchy}\label{sec:geomderivation}

We will now present an entirely geometrical new derivation of
the hierarchy of isospectral flows \eqref{hierarchy.isoflow},
including the Fordy-Kulish NLS system \eqref{FK.nls}
and the associated mKdV system \eqref{FK.mkdv}.

The setting for the derivation is given by non-stretching curve flows
in an affine symmetric space constructed from
any Hermitian symmetric Lie algebra $(\gsp,A)$,
with a corresponding compact semisimple Lie group $G$.
Key ingredients will be to introduce a $G$-parallel frame
and corresponding Hasimoto variable,
by applying the general results in \Ref{Anc2008} for Riemannian symmetric spaces. 

\subsection{Affine Hermitian symmetric spaces}

A Hermitian symmetric Lie algebra $(\gsp,A)$
consists of \cite{KobNom,Hel} a vector space $\gsp=\msp\oplus\hsp$
equipped with the following structure:
a Lie bracket with the properties \eqref{symmsp.brackets};
a complex structure $J=\ad(A)$ on $\msp$,
satisfying properties \eqref{centrA}--\eqref{adA}, 
where $\ad(A)\hsp =0$;
an ad-invariant inner product given by the Killing form $\Kill{\;\cdot\;,\cdot\;}$. 

Typically,
a symmetric Lie algebra $\gsp=\msp\oplus\hsp$
is identified with the infinitesimal isometry group $G$ of a corresponding symmetric space defined by $M=G/H$,
with $H\subset G$ being the stabilizer group of the origin $o$ of $M$.
There is a left-invariant $\gsp$-valued $1$-form on $G$,
called the Maurer-Cartan form \cite{KobNom,Hel},
which provides a canonical isomorphism between the vector space $\msp$
and the tangent spaces of $M$.
As a manifold, the symmetric space $M=G/H$ will have a curvature $2$-form 
given in terms of Maurer-Cartan.
When $\gsp$ has a Hermitian structure,
the resulting curved Hermitian manifold $M=G/H$ is a Hermitian symmetric space
(namely, a Riemannian symmetric space with a compatible complex structure)
\cite{KobNom,Hel},

Rather than consider this geometrical setting,
we instead construct an affine symmetric space
\begin{equation}
\aff(G,H) := (G\rtimes\gsp)/G , 
\end{equation}
where $G$ acts by a semidirect product on $\gsp$ via the adjoint representation.
This space has a very different geometrical structure in comparison to $M=G/H$. 

Elements of $G\rtimes\gsp$ will be represented by $(\g,a)$,
with $\g\in G$, $a\in\gsp$,
where $G\subset G\rtimes\gsp$ is the inclusion $(\g,0)$. 
A group product for $G\rtimes\gsp$ will be defined by 
$(\g_1,a_1)(\g_2,a_2) = (\g_1\g_2,\Ad(\g_1)a_2+a_1)$,
which has a semidirect product structure. 

We will now explain the main features of the space $\aff(G,H)$. 

First, we consider the algebraic structure of $\aff(G,H)$. 
The Lie algebra of $G\rtimes\gsp$ is $\gsp_\aff := \gsp\rtimes\gsp$.
As a vector space,
\begin{equation}\label{aff.vs}
\gsp_\aff = \hsp_\aff\oplus\msp_\aff
\end{equation}
whose elements will be represented by 
\begin{equation}\label{aff.rep}
(a,a')\in \gsp_\aff, 
\quad
(a,0)\in \hsp_\aff \simeq \gsp, 
\quad
(0,a')\in\msp_\aff \simeq \gsp \simeq \Rnum^{\dim(\gsp)} ,
\end{equation}
with $a,a'\in\gsp$. 
The Lie bracket on $\gsp_\aff$ is given by 
\begin{equation}
[(a_1,a'_1),(a_2,a'_2)] = ([a_1,a_2],[a_1,a'_2]+[a'_1,a_2]) . 
\end{equation}
For later, it will be useful to note the semidirect algebra structure is given by 
\begin{equation}\label{aff-Liebrackets}
[(a_1,0),(a_2,0)] = ([a_1,a_2],0),
\quad
[(a_1,0),(0,a'_2)] = (0,[a_1,a'_2]),
\quad
[(0,a'_1),(0,a'_2)] = (0,0) . 
\end{equation}

The Lie bracket gives $\aff(G,H)$ the structure of a symmetric space 
\begin{equation}\label{aff.symmsp}
[\hsp_\aff,\hsp_\aff] \subseteq \hsp_\aff,
\quad
[\hsp_\aff,\msp_\aff] \subseteq \msp_\aff,
\quad
[\msp_\aff,\msp_\aff] =0 . 
\end{equation}
There are further decompositions 
$\msp_\aff = \Rnum^{\dim(\hsp)}\oplus\Rnum^{\dim(\msp)}$,
$\hsp_\aff = \hsp \oplus\msp$,
$\gsp_\aff = (\hsp\oplus\msp)\rtimes(\hsp\oplus\msp)$, 
which induce corresponding decompositions of the Lie bracket, used later. 

The element $A\in\hsp$ defines a complex structure $J_\aff =\ad((A,0))$ 
on both $(\msp,0)\subset\hsp_\aff$ and $(0,\msp)\subset\msp_\aff$. 
Moreover, $J_\aff^2 = -\id_\msp$ acts as a projection operator on these subspaces
and annihilates the complementary subspaces $(\hsp,0)\subset\hsp_\aff$ and $(0,\hsp)\subset\msp_\aff$.

The Killing form on $\gsp$ yields a positive-definite inner product on $\gsp_\aff$:
\begin{equation}\label{aff-innerprod}
\brack{(a_1,a'_1),(a_2,a'_2)}_\aff := -\Kill{a_1,a_2} -\Kill{a'_1,a'_2} . 
\end{equation}
This inner product is ad-invariant.
It is also Hermitian invariant when restricted to the subspaces $(\msp,0)\subset\hsp_\aff$ and $(0,\msp)\subset\msp_\aff$.

Next, we consider structure of $\aff(G,H)$ as a Klein geometry \cite{Sha,Anc2008}. 

The space $\aff(G,H)$ possesses a Maurer-Cartan form $\Omega$
which is a left-invariant $\gsp_\aff$-valued $1$-form on $G\rtimes\gsp$
satisfying
\begin{equation}\label{MCeqn}
\diff\Omega + \tfrac{1}{2}[\Omega,\Omega] =0 . 
\end{equation}
Moreover, 
$\aff(G,H)$ is a principal $G$-bundle over $\gsp$
on which $\Omega$ defines a zero-curvature connection. 
Let $\Omega^\aff$ be the pullback of $\Omega$ from $G$ to $\aff(G,H)$
in a local trivialization of the bundle $\aff(G,H)$. 
A smooth change of the local trivialization is described by a $G$-valued function $\g(x)$ on $\aff(G,H)$.
This induces a gauge transformation of $\Omega^\aff$:
\begin{equation}\label{MCfreedom}
\Omega^\aff \rightarrow \tilde\Omega^\aff = \Ad(\g(x)\inv)\Omega^\aff  +\g(x)\inv \diff\g(x) . 
\end{equation}

The $1$-form $\Omega^\aff$ can be decomposed via the vector space structure \eqref{aff.vs}.
This yields 
\begin{equation}\label{coframe-conx}
\cofra :=(\Omega^\aff)_{\msp_\aff},
\quad
\conx :=(\Omega^\aff)_{\hsp_\aff} ,
\end{equation}
where 
\begin{align}
& \diff\cofra + \comm{\conx,\cofra} =0,
\label{aff-tors}
\\
& \diff\conx + \tfrac{1}{2}\comm{\conx,\conx} =0
\label{aff-curv}
\end{align}
arise from the zero-curvature property \eqref{MCeqn} of the Maurer-Cartan form.
These equations \eqref{aff-tors}--\eqref{aff-curv} imply that
$\cofra$ is a $\gsp$-valued linear coframe on $\aff(G,H)$,
and $\conx$ is a $\gsp$-valued linear connection $1$-form $\conx$ on $\aff(G,H)$. 
In particular, $\cofra$ provides a linear map from the tangent space of $\aff(G,H)$ at any point $x$ into the vector space $\msp_\aff$:
\begin{equation}\label{solder}
\cofra: T_x\aff(G,H) \mapsto \msp_\aff \simeq \gsp . 
\end{equation}

The structure group of frame bundle of $\aff(G,H)$ is $G$. 
Under a local gauge transformation \eqref{MCfreedom},
the linear coframe and linear connection transform by 
\begin{equation}\label{gaugetransf}
\tilde\cofra=\Ad(\g(x)\inv)\cofra,
\quad
\tilde\conx=\Ad(\g(x)\inv)\conx + \g(x)\inv \diff\g(x) . 
\end{equation}

Last, we consider the Riemannian structure of $\aff(G,H)$. 

As a manifold, $\aff(G,H)$ is a vector space in which
vectors $X\in T_x\aff(G,H)\simeq \aff(G,H)$
can be identified with elements $\cofra\hook X \in \msp_\aff$
through the soldering relation \eqref{solder}. 
The inner product on $\msp_\aff$ thereby yields a positive-definite metric on $\aff(G,H)$:
\begin{equation}\label{metric}
g(X,Y) := \brack{\cofra\hook X,\cofra\hook Y}_\aff . 
\end{equation}

The manifold $\aff(G,H)$ carries a group action by $G$,
where group elements $\g\in G$ act on vectors $X\in\aff(G,H)$ via
\begin{equation}
\cofra\hook (\g X)= \Ad(\g)(\cofra\hook X) = \g (\cofra\hook X)\g\inv . 
\end{equation}  
The isometry group of $\aff(G,H)$ is $G\rtimes\gsp$,
with $\gsp\subset G\rtimes\gsp$ acting as translations. 

The metric $g$ on $\aff(G,H)$ determines a torsion-free Riemannian connection,
which can be viewed as a contravariant derivative $\grad$ satisfying $\grad g=0$.
This contravariant derivative provides a covariant derivative, $\nabla_X= g(X,\grad)$,
which is given in terms of the linear connection $\conx$
through the soldering form $\cofra$:
\begin{equation}\label{conx}
\cofra\hook\nabla_XY := \partial_X(\cofra\hook Y) +[\conx\hook X,\cofra\hook Y] . 
\end{equation}

The Maurer-Cartan equations \eqref{aff-tors}--\eqref{aff-curv} 
combined with the symmetric space structure \eqref{aff.symmsp}
show that $\grad$ has 
vanishing torsion $T(X,Y) = \nabla_X Y - \nabla_Y X -\comm{X,Y}=0$
and vanishing curvature $R(X,Y) = \comm{\nabla_X,\nabla_Y} -\nabla_{\comm{X,Y}}=0$;
here the bracket $\comm{\;\cdot,\cdot\;}$ denotes the commutator of vector fields.

Thus, $\aff(G,H)$ is a flat manifold,
whose Riemannian structure is isomorphic to a Euclidean space $\Rnum^{\dim(\gsp)}$.
It has a degenerate Hermitian structure $J$ given by
\begin{equation}\label{J}
\cofra\hook J(Y) := \ad(A)(\cofra\hook Y), 
\end{equation}
where $J^2(Y)=-Y$ if $\cofra\hook Y$ belongs to either $(\msp,0)\subset\hsp_\aff$ or $(0,\msp)\subset\msp_\aff$,
and otherwise $J^2(Y)=0$. 

Finally, the group action of $G$ provides an explicit representation for the soldering  \eqref{solder}. 

\begin{prop}\label{prop:soldering}
The soldering \eqref{solder} of the tangent space of $\aff(G,H)$ onto $\msp_\aff$ 
is explicitly given by 
\begin{equation}\label{soldering}
\Ad(\psi_\cofra(x))\msp_\aff = T_x\aff(G,H)\simeq \aff(G,H), 
\quad
\psi_\cofra(x)\in G . 
\end{equation}
In particular, the linear coframe and linear connection thereby have the corresponding 
explicit representation 
\begin{equation}\label{soldering.e.w}
\cofra=(0,\Ad(\psi_\cofra\inv)),
\quad
\conx = (\psi_\cofra\inv\d\psi_\cofra,0) . 
\end{equation}
\end{prop}

Thus, there is a one-to-one correspondence between functions $\psi_\cofra(x): \aff(G,H)\to G$ 
and linear coframes $\cofra:T_x\aff(G,H) \to \msp_\aff\simeq\gsp$. 

The derivation of the representations \eqref{soldering.e.w} is straightforward 
from equations \eqref{solder} and \eqref{conx}. 
Specifically, for vectors $X,Y\in T_x\aff(G,H)\simeq\aff(G,H)$, 
note the soldering relation \eqref{soldering} gives a mapping of $Y$ into an element in $\msp_\aff$ given by $\Ad(\psi_\cofra\inv)Y$, 
which thus implies 
$\cofra\hook Y = (0,\Ad(\psi_\cofra\inv)Y)$. 
This mapping directly establishes the representation for the linear coframe. 
Then, note 
$\cofra\hook \nabla_X Y 
= \partial_X(0,\Ad(\psi_\cofra\inv) Y) +\ad(\conx_X)(0,\Ad(\psi_\cofra\inv) Y)
= (0,\partial_X(\Ad(\psi_\cofra\inv) Y) +[\conx_X,\Ad(\psi_\cofra\inv) Y])$
by equation \eqref{conx},
while 
$\cofra\hook \nabla_X Y = (0,\Ad(\psi_\cofra\inv)\nabla_X Y)
= (0,\partial_X(\Ad(\psi_\cofra\inv)Y) + [\psi_\cofra\inv\partial_X\psi_\cofra,\Ad(\psi_\cofra\inv)Y])$
by the previous soldering relation, 
with the use of $\partial_X\Ad(\psi_\cofra\inv) = -\ad(\psi_\cofra\inv\partial_X\psi_\cofra)$. 
Equating these expressions for $\cofra\hook \nabla_X Y$ yields 
the representation for the linear connection.

\subsection{Parallel framing of non-stretching curve flows} 

Now consider an arbitrary smooth non-stretching curve flow
$\curvflow=\curvflow(x,t)$ in an affine Hermitian symmetric space $\aff(G,H)$,
where $x$ parameterizes the curve and $t$ is time.
The non-stretching property means that the arclength $|\curvflow_x|dx$
is locally preserved at every point $x$ on the curve under evolution,
so that $D_t|\curvflow_x|=0$,
with $|\curvflow_x|=\sqrt{g(\curvflow_x,\curvflow_x)}$.
Hereafter, $x$ will be taken to be the arclength, 
and its domain will be assumed to be $C=\Rnum$ or $S^1$. 

A moving frame for $\curvflow(x,t)$ consists of specifying a set of $\dim(\gsp)$ orthonormal vectors spanning $\aff(G,H)$ at each point in the two-dimensional surface $\curvflow(x,t)$.
From the viewpoint of Klein geometry \cite{Sha}, 
this is equivalent to specifying an orthonormal basis for $\gsp\simeq\msp_\aff$
which gets mapped into corresponding frame vectors for the tangent space 
$T_\curvflow\aff(G,H)$ of the surface $\curvflow(x,t)$ 
through the soldering relation \eqref{solder}
provided by the linear coframe $\cofra$.
In particular,
if $X$ is a frame vector in $T_\curvflow\aff(G,H)$,
then $\cofra\hook X$ is its corresponding basis element in $\gsp$.
The frame is said to be adapted to the curve
if the tangent vector $\curvflow_x$ is one of the frame vectors,
while the remaining frame vectors belong to the normal space of $\curvflow_x$.

Given an orthonormal adapted frame at any point $x$ on the curve $\curvflow(x)$,
its infinitesimal transport along the curve is given by the Cartan connection matrix
which describes the Frenet equations of the frame.
Through the soldering relation \eqref{solder},
the Cartan connection matrix can be identified with the linear connection $\conx$
projected along the tangent vector $\curvflow_x$ of the curve:
\begin{equation}\label{transport.curve}
\nabla_x\cofra = -\ad(\conx\hook\curvflow_x)\cofra, 
\end{equation}
where $\nabla_x =g(\curvflow_x,\grad)$. 
Similarly,
the infinitesimal transport of the frame at each point $x$ under the flow $\curvflow(x,t)$
is given by the Cartan connection matrix 
arising from the evolution equations of the frame,
which can be identified with the linear connection $\conx$
projected along the flow vector $\curvflow_t$:
\begin{equation}\label{transport.flow}
\nabla_t\cofra = -\ad(\conx\hook\curvflow_t)\cofra,
\end{equation}
where $\nabla_t =g(\curvflow_t,\grad)$. 

Two orthonormal adapted frames will be related locally by
the action \eqref{gaugetransf} of frame gauge group $G$
such that $\cofra\hook\curvflow_x$ remains fixed,
namely $\Ad(\g(x)\inv)(\cofra\hook\curvflow_x) =\cofra\hook\curvflow_x$. 
When $\conx\hook\curvflow_x$ is held fixed in addition to $\cofra\hook\curvflow_x$, 
then the action of the gauge group relating the two frames consists only of
a rigid transformation \eqref{gaugetransf} in which $\g$ is a constant group element
belonging to the stabilizer subgroup of $\cofra\hook\curvflow_x$ in $\gsp$.
This subgroup in $G$ is called the equivalence group of an adapted frame. 

As a result, a non-stretching curve flow together with an adapted orthonormal frame
is determined up to equivalence by specifying 
\begin{equation}
\begin{aligned}
&
e_x :=\cofra\hook\curvflow_x \in\msp_\aff, 
\quad
e_t :=\cofra\hook\curvflow_t \in\msp_\aff, 
\\
&
\omega_x :=\conx\hook\curvflow_x \in\hsp_\aff, 
\quad
\omega_t :=\conx\hook\curvflow_t \in\hsp_\aff ,
\end{aligned}
\end{equation}
for the projections of $\cofra$ and $\conx$ along the curve and along the flow, respectively.
In particular, $e_x$ can be chosen as a constant,
while $e_t$, $\omega_x$, $\omega_t$ will be, in general, functions of $(x,t)$. 
These functions will satisfy the Maurer-Cartan equations \eqref{aff-tors}--\eqref{aff-curv} on the two-dimensional surface swept out by $\curvflow(x,t)$:
\begin{align}
& \Dx e_t -D_te_x+[\omega_x,e_t]-[\omega_t,e_x]=0,
\label{tors-eqn}
\\
& \Dx\omega_t-D_t\omega_x+[\omega_x,\omega_t] =0,
\label{curv-eqn}
\end{align}
where $\Dx,D_t$ represent total derivatives.
Note that $[e_x,e_t]=0$ because $\aff(G,H)$ is a flat space.

The non-stretching property of $\curvflow(x,t)$ imposes the condition that $e_x$ has unit norm,
$\brack{e_x,e_x}_\aff=1$.
Hence $e_x$ is a constant unit-norm element in $\msp_\aff$. 
Under a local gauge transformation on the frame, 
$e_x$ transforms into $\tilde e_x=\Ad(\g(x)\inv)e_x$. 
Hence, 
the possible inequivalent choices for $e_x$ are in one-to-one correspondence with the orbits of the group action $\Ad(G)$ on $\map_\aff$. 
If the group $G$ has rank greater than one, 
then there will be more than one choice of $e_x$ up to equivalence. 
(For further discussion, see \Ref{Anc2008}.)
Consequently, 
curves $\curvflow(x,t)$ in $\aff(G,H)$ at any fixed time $t$ 
can be divided into algebraic equivalence classes
determined by the inequivalent choices of $e_x$ in $\msp_\aff$. 
Since $e_x$ is constant, 
a curve flow $\curvflow(x,t)$ will stay in the same equivalence class 
throughout its evolution. 

\begin{rem}
The equivalence class of curves that will be relevant for obtaining 
a geometrical realization of the isospectral flows \eqref{hierarchy.isoflow} 
turn out to have $e_x$ belonging to the span of the element $A$ in $\hsp$. 
This is motivated by comparing the representation 
$e_x = (0,\Ad(\psi_\cofra\inv)\curvflow_x)$ for the tangent vector $\curvflow_x$
given by the soldering relation \eqref{soldering.e.w}
and the analogous representation \eqref{SP.tangent.flow.vectors} 
for the tangent vector $\curvflow_s$ of the Sym-Pohlmeyer curve
associated with an isospectral flow. 
\end{rem}

Thus, hereafter the curve flows $\curvflow(x,t)$ under consideration 
will be taken to have 
\begin{equation}\label{e.A}
e_x := \e = \enorm (0,A) \in \msp_\aff, 
\end{equation}
where 
\begin{equation}\label{norm.A}
\efac = \brack{(0,A),(0,A)}_\aff=-\Kill{A,A}=\dim(\msp)
\end{equation}
is the norm-squared of $A$. 

From the general theory of parallel frames in \Ref{Anc2008}, 
a $G$-parallel frame for these curve flows $\curvflow(x,t)$ in $\aff(G,H)$ 
will be an adapted frame having the property that 
its Cartan connection matrix $\omega_x$ 
belongs to the perp space of the centralizer of $e_x$ in $\hsp_\aff$.
The equivalence group of this frame consists of the subgroup $G_\pa\subset G$ that preserves $e_x$.
Its Lie algebra $\gsp_\pa\subset \gsp$ is the centralizer of $A$,
which is given by $\hsp$ from property \eqref{centrA},
and hence $G_\pa = H$. 

To work out what this definition implies for $\omega_x$, 
let $(\hsp_\aff)_\pa$ be the centralizer of $\e$ in $\hsp_\aff$,
given by
$[\e,(\hsp_\aff)_\pa]=0$,
and let $(\hsp_\aff)_\pe$ be the orthogonal complement defined by
$\brack{(\hsp_\aff)_\pe, (\hsp_\aff)_\pa}_\aff=0$.
Then 
\begin{equation}\label{affhsp-pape}
\hsp_\aff = (\hsp_\aff)_\pe\oplus (\hsp_\aff)_\pa
\end{equation}
is a direct sum of vector spaces.
The explicit representation for these subspaces can be obtained from the Lie bracket
$[\enorminv\e,\hsp_\aff]=[(0,A),(a,0)] = (0,[A,a])$
plus the property \eqref{centrA} that the centralizer of $A$ in $\gsp$ is $\hsp$. 
This determines 
\begin{equation}\label{aff.h.pe.pa}
(\hsp_\aff)_\pa=\{(a,0)|a\in\hsp\},
\quad
(\hsp_\aff)_\pe=\{(a,0)|a\in\msp\} . 
\end{equation}
Since $\hsp_\aff\simeq \msp_\aff\simeq \gsp$ as vector spaces, 
it will be convenient to also define
\begin{equation}\label{aff.m.pe.pa}
(\msp_\aff)_\pa=\{(0,a)|a\in\hsp\},
\quad
(\msp_\aff)_\pe=\{(0,a)|a\in\msp\},
\end{equation}
whereby 
\begin{equation}\label{aff.m.decomp}
 \msp_\aff = (\msp_\aff)_\pe\oplus (\msp_\aff)_\pa . 
\end{equation}
As a consequence of invariance of the Killing form with respect to $\ad(\e)$, 
the Lie bracket structure \eqref{aff.symmsp} has the following decomposition:
\begin{align}
&
[(\hsp_\aff)_\pa,(\hsp_\aff)_\pa] \subseteq (\hsp_\aff)_\pa,
\quad
[(\hsp_\aff)_\pa,(\msp_\aff)_\pa] \subseteq (\msp_\aff)_\pa,
\label{lie.rel.pa}
\\
& 
[(\hsp_\aff)_\pa,(\hsp_\aff)_\pe] \subseteq (\hsp_\aff)_\pe,
\quad
[(\hsp_\aff)_\pe,(\msp_\aff)_\pa] \subseteq (\msp_\aff)_\pe,
\quad
[(\hsp_\aff)_\pa,(\msp_\aff)_\pe] \subseteq (\msp_\aff)_\pe,
\label{lie.rel.pape}
\\
& 
[(\hsp_\aff)_\pe,(\hsp_\aff)_\pe] \subseteq (\hsp_\aff)_\pa,
\quad
[(\hsp_\aff)_\pe,(\msp_\aff)_\pe] \subseteq (\msp_\aff)_\pa . 
\label{lie.rel.pe}
\end{align}  

Through these decompositions \eqref{affhsp-pape}--\eqref{aff.m.decomp} and \eqref{lie.rel.pa}--\eqref{lie.rel.pe},
the Maurer-Cartan equations \eqref{tors-eqn}--\eqref{curv-eqn} become 
\begin{align}
& \Dx (e_t)_\pa+[\omega_x,(e_t)_\pe ]=0,
\label{tors-eq.mpar}
\\
& \Dx(e_t)_\pe+[\omega_x,(e_t)_\pa] -[(\omega_t)_\pe,\e]=0,
\label{tors-eq.mperp}
\\
& \Dx(\omega_t)_\pa+[\omega_x,(\omega_t)_\pe]=0,
\label{curv-eq.hpar}
\\
& \Dx(\omega_t)_\pe -D_t \omega_x+[\omega_x,(\omega_t)_\pa] =0, 
\label{curv-eq.hperp}
\end{align}
in which the variables are represented by 
\begin{gather}
\begin{aligned}
(e_t)_\pa & = (0,h_\pa) \in(\msp_\aff)_\pa, 
\quad
(e_t)_\pe & = (0,h_\pe) \in(\msp_\aff)_\pe,
\end{aligned}
\label{cofra.flow.vars}
\\
\begin{aligned}
(\omega_t)_\pa & = (w_\pa,0) \in(\hsp_\aff)_\pa, 
\quad
(\omega_t)_\pe & = (w_\pe,0) \in(\hsp_\aff)_\pe,
\end{aligned}
\label{conx.flow.vars}
\\
\omega_x = (\q,0) \in(\hsp_\aff)_\pe,
\label{conx.tangent.var}
\end{gather}
where $h_\pa,w_\pa\in\hsp$ and $h_\pe,w_\pe,\q\in\msp$ are functions of $(x,t)$.
The Lie brackets in these equations \eqref{tors-eq.mpar}--\eqref{curv-eq.hperp}
are straightforward to compute using the representations \eqref{aff.rep} and \eqref{aff-Liebrackets}. 
This yields
\begin{align}
& \Dx h_\pa +[\q,h_\pe] =0,
\label{flow.mpar}
\\
& \Dx h_\pe +[\q,h_\pa]+ \enorm  [A,w_\pe]=0,
\label{flow.mperp}
\\
& \Dx w_\pa +[\q,w_\pe]=0,
\label{flow.hpar}
\\
& \Dx w_\pe - \q_t +[\q,w_\pa] =0 , 
\label{flow.hperp}
\end{align}
where all of these reduced Lie brackets sit in $\gsp$.

The preceding developments,
combined with Proposition~\ref{prop:soldering}, 
can be summarized as follows.

\begin{thm}\label{thm:Gparallel.frame}
In any affine Hermitian symmetric space $\aff(G,H)=(G\rtimes\gsp)/G$, 
with $\cofra$ and $\conx$ being a linear coframe and a linear connection \eqref{soldering.e.w}, 
there is a distinguished class of non-stretching curve flows $\curvflow(x,t)$ for which 
\begin{equation}\label{curve}
\cofra\hook\curvflow_x = \enorm (0,A) \in (\msp_\aff)_\pa\simeq\hsp, 
\end{equation}
where $A\in\hsp$ is the complex-structure element \eqref{centrA}--\eqref{adA}. 
These curve flows have a $G$-parallel framing given by 
\begin{equation}\label{Gparallel.conx}
\conx\hook\curvflow_x = (\psi_\cofra\inv \Dx\psi_\cofra,0) =(\q,0)\in(\msp_\aff)_\pe \simeq\msp,
\end{equation}
and
\begin{equation}\label{Gparallel.flow}
\cofra\hook\curvflow_t= (0,h_\pa + h_\pe) \in\msp_\aff\simeq\gsp, 
\quad
\conx\hook\curvflow_t = (\psi_\cofra\inv D_t\psi_\cofra,0)=(w_\pa+w_\pe,0)\in\msp_\aff\simeq\gsp , 
\end{equation}
where $\psi_\cofra(\curvflow)\in G$. 
The equivalence group of the framing is $G_\pa=H \subset G$,
where $H$ is the stabilizer group of $A\in\hsp$,
acting in the adjoint representation. 
For such a framing, 
the frame structure equations take the form \eqref{flow.mpar}--\eqref{flow.hperp},
which are invariant under $G_\pa$. 
\end{thm}

\begin{rem}
The variables \eqref{cofra.flow.vars}--\eqref{conx.tangent.var}
in the frame structure equations have the following geometrical meaning:
$\q$ encodes the Cartan connection matrix of the $G$-parallel frame; 
$h_\pe$ and $h_\pa$ respectively encode the motion of the curve 
projected into the subspaces \eqref{aff.m.pe.pa} of $\msp_\aff$; 
$w_\pe+w_\pa$ encodes the Cartan connection matrix of the evolution of the $G$-parallel frame
induced by the motion of the curve.
\end{rem}

A geometrical interpretation of the algebraic property \eqref{curve} distinguishing the curve flows $\curvflow(x,t)$ 
will be given in \secref{subsec:geom.curve}. 

Furthermore, as will be shown later,
there is an explicit geometrical correspondence between
the $G$-parallel frame structure equations \eqref{flow.mpar}--\eqref{flow.hperp} 
and the isospectral flow equations \eqref{ZCh0}--\eqref{ZCh1}, \eqref{ZCm0}--\eqref{ZCm1}. 
This correspondence is closely related to the Sym-Pohlmeyer construction.

\subsection{Bi-Hamiltonian structure and hierarchies of isospectral flows}

From the general results \cite{Anc2008}
known for non-stretching curve flows in Riemannian symmetric spaces,
the $G$-parallel frame structure equations \eqref{flow.mpar}--\eqref{flow.hperp} 
turn out to encode a bi-Hamiltonian structure for the induced flows
on the Cartan connection-matrix variable $\q$. 

Introduce the variable 
\begin{equation}\label{hperp}
h^\pe :=\ad(\e) h_\pe = \enorm [A,h_\pe] \in\msp
\end{equation}
and use equations \eqref{flow.mpar} and \eqref{flow.hpar} to eliminate the variables
\begin{equation}\label{hpar}
h_\pa = \enorminv \Dxinv[\q,Jh^\pe],
\quad
w_\pa = -\Dxinv[\q,w_\pe],  
\end{equation}
where the property \eqref{adA} that $\ad(A)=J$ is a complex structure on $\msp$
has been used. 
Then the remaining equations \eqref{flow.mperp} and \eqref{flow.hperp} become
\begin{align}
w_\pe & = \efac \Jop(h^\pe) , 
\label{wperp.eq}
\\
\q_t & = \Hop(w_\pe) , 
\label{q.eq}
\end{align}
where $\Hop$ is the operator
\begin{equation}\label{Hop}
\Hop = \Dx -\ad(\q)\Dxinv\ad(\q)
\end{equation}  
and $\Jop$ is the conjugated operator
\begin{equation}\label{Jop}
\Jop=J\Hop J\inv = \Dx -J\ad(\q)\Dxinv\ad(\q)J\inv .
\end{equation}  

\begin{thm}\label{thm:iso.flow.eqs}
The operator equations \eqref{wperp.eq}--\eqref{q.eq}
on the Cartan connection-matrix variables $\q$ and $w_\pe$ provide 
an equivalent formulation of the $G$-parallel frame structure equations
describing non-stretching curve flows $\curvflow(x,t)$ of the algebraic type \eqref{curve}
in any affine Hermitian symmetric space $\aff(G,H)=(G\rtimes\gsp)/G$.
In this formulation,
$\Hop$ and $J$ are a compatible pair of Hamiltonian operators with respect to the $\msp$-valued flow variable $\q$,
and $\Rop=\Hop J\inv$ is a hereditary recursion operator.
A curve flow $\curvflow(x,t)$ is determined by specifying the variable $h^\pe(x,t)$,
which yields the flow equation 
\begin{equation}\label{Gparallel.floweqn}
\tfrac{1}{\efac} \q_t = \Hop(w_\pe) = -\Rop^2(h^\pe)
\end{equation}
for $\q(x,t)$.
The curve flow equation as well as the Hamiltonian operators (and the recursion operator)
are invariant under the equivalence group $G_\pa=H\subset G$ of the $G$-parallel frame. 
\end{thm}

This result is a direct application of the general theorem in \Ref{Anc2008}
on a parallel-frame formulation of non-stretching curve flows
in Riemannian symmetric spaces.
The proof of that theorem involves showing that the pair of operators arising
in the flow equation in a parallel frame satisfy the conditions to be a bi-Hamiltonian pair,
in particular,
$\Hop$ gives rise to a bilinear bracket that is skew and obeys the Jacobi identity,
while $\Jop$ gives rise to a bilinear form that is skew and closed.
These conditions are proved to hold by using an extension of the standard theory of bi-Hamiltonian structures to the setting of Lie algebra-valued variables.
The general theory of recursion operators \cite{Olv-book} then shows that
the hereditary property of $\Rop$ follows from its factorized form 
in terms of $\Hop$ and $\Jop$. 

Thus, a $G$-parallel frame encodes a natural bi-Hamiltonian structure $\Hop$, $J$,
along with a hereditary recursion operator $\Rop$. 
By Magri's theorem,
a hierarchy of bi-Hamiltonian flows can be produced by taking $\X = h^\pe\hook\p_{\q}$
to be the generator of a symmetry of the Hamiltonian operators $\Hop$, $J$. 

There are two natural symmetries,
consisting of phase rotations and translations,
generated by
\begin{align}
\X_{\text{phas.}} & = -J\q\hook\p_{\q},
\label{phas}\\
\X_{\text{trans.}} & = -\q_x\hook\p_{\q} ,
\label{trans}
\end{align}
which respectively produce a hierarchy of flows on $\q$
starting from $h^\pe = -J\q$ and $h^\pe = -\q_x$.

The first hierarchy is given by
\begin{equation}\label{nls-hierarchy}
\tfrac{1}{\efac} \q_t = \Rop^{2n}(J\q), 
\quad
n=0,1,2,\ldots . 
\end{equation}
Its root flow, for $n=0$, is $\tfrac{1}{\efac}\q_t = J\q$.
The next flow, for $n=1$, is obtained by starting with
$\Rop(J\q)= \Hop(\q) = \q_x$,
and so
$\Rop^2(J\q) = \Rop(\q_x) = \Hop(J\inv \q_x) = -J\q_{xx} +\ad(\q)\Dxinv(\ad(\q)J\q_x)$,
where $\ad(\q)J\q_x= \tfrac{1}{2}(\ad(\q)J\q)_x$ was derived earlier. 
Hence, this yields
\begin{equation}\label{Gparallel.nls}
\tfrac{1}{\efac}\q_t=\Rop^2(J\q) = -J\q_{xx} +\tfrac{1}{2}\ad(\q)^2J\q ,
\end{equation}
which is, up to scaling, the same as the isospectral NLS system \eqref{FK.nls}. 
Note that this flow \eqref{Gparallel.nls} has
\begin{equation}\label{nls.flow.vars}
h^\pe = -J\q,
\quad
h_\pa = 0,
\quad
\tfrac{1}{\efac} w_\pe=-J\q_x,
\quad
\tfrac{1}{\efac} w_\pa = \tfrac{1}{2} \ad(\q)J\q . 
\end{equation}

Similarly,
the second hierarchy is given by 
\begin{equation}\label{mkdv-hierarchy}
\tfrac{1}{\efac}\q_t = \Rop^{2n}(\q_x), 
\quad
n=0,1,2,\ldots , 
\end{equation}
which has the root flow, for $n=0$, is $\tfrac{1}{\efac}\q_t = \q_x$.
The next flow is given by
$\Rop^2(\q_x)= \Rop(-J\q_{xx} +\tfrac{1}{2}\ad(\q)^2J\q)$,
which requires several steps to simplify.
The first term
\begin{equation}\label{mkdvterm1}
-\Rop(J\q_{xx}) = -\Hop(\q_{xx}) = -\q_{xxx} +\ad(\q)^2\q_x
\end{equation}  
follows from the use of $\Dxinv(\ad(\q)\q_{xx})= [\q,\q_x]$.
The second term is 
$\tfrac{1}{2}\Rop(\ad(\q)^2J\q) = \tfrac{1}{2}\Hop(J\inv\ad(\q)^2J\q) = -\tfrac{1}{2} J(\ad(\q)^2J\q)_x +\tfrac{1}{2}\ad(\q)\Dxinv(\ad(\q)J\ad(\q)^2J\q)$.
This expression can be simplified via the identity 
\begin{equation}\label{ident}
J\ad(\q)^2J\q = [J,\ad(\q)]\ad(\q)J\q + (\ad(\q)J)^2\q =-\ad(J\q)^2\q ,
\end{equation}
which holds by $[J,\ad(\q)]=\ad(J\q)$ and $(\ad(\q)J)^2\q =\ad(q)J[\q,J\q]=0$
due to $[\q,J\q]\in\hsp$.
First, note that the local term in $\tfrac{1}{2}\Rop(\ad(\q)^2J\q)$ 
directly becomes
$-\tfrac{1}{2} J(\ad(\q)^2J\q)_x = \tfrac{1}{2}(\ad(J\q)^2\q)_x$.
Next, the nonlocal term in $\tfrac{1}{2}\Rop(\ad(\q)^2J\q)$
vanishes by the following steps applied to $\ad(\q)J\ad(\q)^2J\q$.
On one hand, from the identity \eqref{ident}, 
consider 
$\ad(\q)J\ad(\q)^2J\q = -\ad(\q)\ad(J\q)^2\q
= \ad(\q)[J\q,\ad(\q)J\q] = [\ad(\q)J\q,\ad(\q)J\q] + [J\q,\ad(\q)^2J\q]= \ad(J\q)\ad(\q)^2J\q$.
On the other hand, by the same steps used in proving the identity \eqref{ident},
note 
$\ad(\q)J\ad(\q)^2J\q = [\ad(\q),J]\ad(\q)^2J\q + J\ad(\q)^3J\q = -\ad(J\q)\ad(\q)^2J\q$. 
Thus, $\ad(\q)J\ad(\q)^2J\q = -\ad(\q)J\ad(\q)^2J\q$,
which implies $\ad(\q)\Dxinv(\ad(\q)J\ad(\q)^2J\q)= 0$. 
Consequently, these steps give
\begin{equation}\label{mkdvterm2}
\tfrac{1}{2}\Rop(\ad(\q)^2J\q) = \tfrac{1}{2}(\ad(J\q)^2\q)_x . 
\end{equation}
Combining the terms \eqref{mkdvterm2} and \eqref{mkdvterm1}
thereby yields
\begin{equation}\label{Gparallel.mkdv}
\tfrac{1}{\efac}\q_t =\Rop^2(\q_x) 
= -\q_{xxx} +\ad(\q)^2\q_x +\tfrac{1}{2}(\ad(J\q)^2\q)_x ,
\end{equation}
which is, up to scaling, the same as the isospectral mKdV system \eqref{FK.mkdv}. 
This flow \eqref{Gparallel.mkdv} has
\begin{equation}\label{mkdv.flow.vars}
h^\pe = -\q_x,
\quad
\enorm h_\pa = -\tfrac{1}{2}\ad(\q)J\q, 
\quad
\tfrac{1}{\efac} w_\pe=-\q_{xx} - \tfrac{1}{2}\ad(\q)^2J\q, 
\quad
\tfrac{1}{\efac} w_\pa = -\ad(\q)\q_x . 
\end{equation}

In these two hierarchies \eqref{nls-hierarchy} and \eqref{mkdv-hierarchy}, 
the respective root flows are related by the recursion operator $\Rop$,
since $\Rop(J\q) = \q_x$.
Hence, the hierarchies can be merged to form a single intertwined hierarchy of flows
given by 
\begin{equation}\label{aff.hierarchy}
\tfrac{1}{\efac}  \q_t = \Rop^n(J\q) := h^\pe_{(n)},
\quad
n=0,1,2,\ldots . 
\end{equation}
In particular, 
the $0$ flow is $\tfrac{1}{\efac}\q_t = J\q =h^\pe_{(0)}$ 
and the $+1$ flow is $\tfrac{1}{\efac}\q_t = \q_x =h^\pe_{(1)}$,
while the $+2$ and $+3$ flows are precisely the NLS flow \eqref{Gparallel.nls}
and the mKdV flow \eqref{Gparallel.mkdv}.

\begin{rem}
All of the higher flows in the hierarchy \eqref{aff.hierarchy}
can be viewed as higher symmetries
$\X^{(n)} = h^\pe_{(n)}\hook\p_\q$, $n\geq 4$,
of both the NLS flow and the mKdV flow.
\end{rem}

The merged hierarchy \eqref{aff.hierarchy} has the following Hamiltonian structure.

\begin{thm}\label{thm:iso.biHamil.hiearchy}
Each flow in the hierarchy \eqref{aff.hierarchy} for $n\geq 1$
has a bi-Hamiltonian formulation
\begin{equation}\label{aff.hierarchy.biHam}
\tfrac{1}{\efac}\q_t = \Hop(\delta H^{(n)}/\delta \q) = J(\delta H^{(n+1)}/\delta \q),
\quad
n=1,2,\ldots 
\end{equation}
where the Hamiltonians are given by the functionals
\begin{equation}\label{aff.Ham.n}
H^{(n)} = \frac{1}{n}\int_C \Dxinv\brack{\Rop^n(J\q),\q}\; dx, 
\quad
n=1,2,\ldots 
\end{equation}
on the domain $C=\Rnum$ or $S^1$.
For $n\geq 2$,
each flow has an additional Hamiltonian structure,
\begin{equation}\label{aff.hierarchy.triHam}
\tfrac{1}{\efac}\q_t = \Eop(\delta H^{(n-1)}/\delta \q), 
\quad
n=2,3,\ldots 
\end{equation}
where $\Eop= \Rop^2 J$ is a third Hamiltonian operator
which is compatible with the Hamiltonian operators $\Hop=\Rop J$ and $J$.
\end{thm}

As a consequence, 
the NLS flow \eqref{Gparallel.nls} and the mKdV flow \eqref{Gparallel.mkdv}
have the tri-Hamiltonian structures
\begin{equation}\label{nls.tri.ham}
\tfrac{1}{\efac}\q_t = -J\q_{xx} +\tfrac{1}{2}[\q,[\q,J\q]]
= \Eop(\delta H^{(1)}/\delta \q) = \Hop(\delta H^{(2)}/\delta \q) = J (\delta H^{(3)}/\delta \q)
\end{equation}
and
\begin{equation}\label{mkdv.tri.ham}
\tfrac{1}{\efac}\q_t = -\q_{xxx} +\ad(\q)^2\q_x +\tfrac{1}{2}(\ad(J\q)^2\q)_x
= \Eop(\delta H^{(2)}/\delta \q) = \Hop(\delta H^{(3)}/\delta \q) = J (\delta H^{(4)}/\delta \q) , 
\end{equation}
which are given in terms of the Hamiltonians
\begin{align}
H^{(1)} & = \int_C \tfrac{1}{2}\brack{\q,\q}_\msp\;dx ,
\\
H^{(2)} &= \int_C \tfrac{1}{2}\brack{J\q,\q_x}_\msp\;dx,
\\
H^{(3)} &= \int_C \tfrac{1}{2}\brack{\q_x,\q_x}_\msp -\tfrac{1}{8}\brack{[\q,J\q],[\q,J\q]}_\msp \;dx,
\\
H^{(4)} &= \int_C \tfrac{1}{2}\brack{J\q_x,\q_{xx}}_\msp -\tfrac{3}{8}\brack{[\q,J\q],[\q,\q_x]}_\msp \;dx . 
\end{align}

Expression \eqref{aff.Ham.n} for the Hamiltonians comes from
a general scaling formula in \Ref{Anc2003} (see also \Ref{Anc2008}),
with $x\to e^\epsilon x$, $\q\to e^{-\epsilon} \q$ ($\epsilon\in\Rnum$)
being the scaling group.
In the case $n=0$, this formula leads to a trivial identity
$(0)H^{(0)} = \int_C \Dxinv(0)\; dx =0$. 

Correspondingly, 
the $n=0$ flow has only a single Hamiltonian form, 
$\tfrac{1}{\efac}\q_t = J\q =  J(\delta H^{(1)}/\delta \q)$,
and the $n=1$ flow has only a bi-Hamiltonian form, 
$\tfrac{1}{\efac}\q_t = \q_x =  \Hop(\delta H^{(1)}/\delta \q) =  J(\delta H^{(2)}/\delta \q)$,
rather than a tri-Hamiltonian form. 
Nevertheless, 
it is possible to view $\Hop(0)= J\q= -\ad(\q)A$
if $\Dxinv$ which appears in $\Hop$ 
is defined to contain a constant of integration given by 
\begin{equation}\label{Dxinv}
\Dxinv(0)=A . 
\end{equation}

\begin{prop}\label{prop:alt.Ham.struc}
In the hierarchy \eqref{aff.hierarchy} of flows, 
the bi-Hamiltonian structure \eqref{aff.hierarchy.biHam} 
can be extended to $n=0$ with a trivial Hamiltonian $H^{(0)}=0$
if the Hamiltonian operator $\Hop$ 
is generalized by defining $\Dxinv$ to have a non-trivial cokernel \eqref{Dxinv}. 
This yields $\Hop(0)= J\q=h^\pe_{(0)}$, which is the $0$ flow. 
In a similar way, 
the tri-Hamiltonian structure \eqref{aff.hierarchy.biHam} 
can be extended to $n=1$ with a trivial Hamiltonian $H^{(0)}=0$, 
whereby $\Eop(0)= \Rop\Hop(0) = \Rop(J\q)=\q_x =h^\pe_{(1)}$,
which is the $+1$ flow. 
\end{prop}

As a final remark,
there is a direct link between the Hamiltonians \eqref{aff.Ham.n}
and the variable $h_\pa$ given in terms of $h^\pe$ by equation \eqref{hpar}.
Specifically, consider
$\Dx\brack{A,h_\pa}_\gsp = \enorminv\brack{A,[\q,\ad(A)h^\pe]}_\gsp = \enorminv\brack{J\q,Jh^\pe}_\msp = \enorminv\brack{\q,h^\pe}_\msp$.
This shows that
$\Dxinv\brack{h^\pe_{(n)},\q}_\msp= \brack{A,h_\pa^{(n)}}_\gsp$,
where
\begin{equation}
h_\pa^{(n)}:= \enorminv \Dxinv[\q,\ad(A)h^\pe_{(n)}]
\end{equation}
based on equation \eqref{hpar}.
Hence the following relationship holds. 

\begin{rem}
The Hamiltonians in the hierarchy \eqref{aff.hierarchy} 
can be expressed as
\begin{equation}
H^{(n)} = \frac{1}{n}\int_C \brack{A,h_\pa^{(n)}}_\gsp\; dx,
\quad
n=1,2,\ldots, 
\end{equation}
where
\begin{equation}
h_\pa^{(n-2)} =h_\pa,
\quad
n\geq 2
\end{equation}  
holds for the $+n$ flow. 
\end{rem}

The $+n$ flow in the hierarchy \eqref{aff.hierarchy} for $n\geq2$ 
determines a non-stretching curve flow
\begin{equation}\label{curveflow}
\curvflow_t = \Pvec_\pe^{(n-1)},
\quad
n=2,3,\ldots
\end{equation}
arising from the soldering relation
\begin{equation}\label{solder.curveflow}
\cofra\hook \Pvec_\pe^{(n-1)} = h_\pe + h_\pa,
\quad
n\geq 2, 
\end{equation}
with $h_\pe$ and $h_\pa$ being given in terms of $h^\pe_{(n-2)}$
through equations \eqref{hperp}, \eqref{hpar}, and \eqref{Gparallel.floweqn}.
Note that the presence of $\Rop^2$ in equation \eqref{Gparallel.floweqn}
accounts for why $h_\pe$ is related to $h^\pe_{(n-2)}$ rather than to $h^\pe_{(n-1)}$ or $h^\pe_{(n)}$. 

It will be useful geometrically to split the motion of the curve into 
components $(\curvflow_t)_\pe = \Pvec_\pe^{(n-1)}$ and $(\curvflow_t)_\pa = \Pvec_\pa^{(n-1)}$ 
at every point along the curve. 
Then the $\pe$ component is given by 
\begin{equation}\label{aff.hierarchy.hperp}
\cofra\hook \Pvec_\pe^{(n-1)} = h_\pe = \enorminv J\inv h^\pe = \enorminv J\inv h^\pe_{(n-2)} = \enorminv {\Rop^{*}}^{n-2}(\q),
\quad
n\geq 2,
\end{equation}
where
\begin{equation}\label{adjRop}
\Rop^*= J\inv\Hop = J\inv(\Dx -\ad(\q)\Dxinv\ad(\q))
= (\Dx -\ad(J\q)\Dxinv\ad(J\q))J\inv
\end{equation}
is the adjoint of the recursion operator
\begin{equation}\label{Rop}
\Rop=(\Dx -\ad(\q)\Dxinv \ad(\q))J\inv
=J\inv(\Dx -\ad(J\q)\Dxinv \ad(J\q)) . 
\end{equation}
The tangential component is determined from the normal component by 
\begin{equation}\label{aff.hierarchy.hpar}
\cofra\hook \Pvec_\pa^{(n-1)} = h_\pa = -\Dxinv(\ad(\q) h_\pe) = -\enorminv\Dxinv(\ad(\q){\Rop^{*}}^{n-2}(\q)),
\quad
n\geq 2 . 
\end{equation}
These relations \eqref{aff.hierarchy.hpar} and \eqref{aff.hierarchy.hperp}
allow $P^{(n-1)}$ to be expressed in terms of $\curvflow_x$ and its $x$-covariant derivatives,
as will be shown explicitly in \secref{sec:VFE}.

\subsection{Relationship to the Sym-Pohlmeyer construction}

The Sym-Pohlmeyer construction of non-stretching curve flows
uses the scaled arclength parameter $s= \enorminv x$
as shown by the relation \eqref{s}. 
Thus, the tangent vector of $\curvflow(s,t)$ with respect to $s$ is represented by 
$\cofra\hook\curvflow_s = (0,A)$
in a $G$-parallel frame. 
Consequently,
the covariant derivative of a vector $Y$ in $\aff(G,H)\simeq\gsp$
with respect to $\curvflow_s$ has the representation 
\begin{equation}\label{Gparallel.conx.curve}
\cofra\hook\nabla_s Y = (0,\p_s y+\enorminv\ad(\q)y)
\end{equation}
with $y=\cofra\hook Y$. 
These soldering relations correspond to the identifications provided by
$\Ad(\phi_0)$ in the Sym-Pohlmeyer construction,
with $\phi_0=\phi|_{\lambda=0}$
where $\phi(s,t,\lambda)$ is the $G$-valued spectral function. 
Specifically, 
if $y=\phi_0 Y\phi_0\inv$ then
$\nabla_s Y =(\phi_0\inv y\phi_0)_s = \phi_0\inv(\p_s y+ [y,\Q])\phi_0
= \phi_0\inv(\p_s y-\ad(\Q)y)\phi_0$
using $\phi_{0\,s} = q \phi_0$ from the linear isospectral system \eqref{spectraleq} and \eqref{U.akns},
whereby 
\begin{equation}\label{SP.conx.curve}
\phi_0(\nabla_s Y)\phi_0\inv =\p_s y-\ad(\Q)y . 
\end{equation}
Comparison of equations \eqref{SP.conx.curve} and \eqref{Gparallel.conx.curve}
establishes the correspondence
\begin{equation}\label{q.Q.rel}
Q=-\enorminv q . 
\end{equation}

Likewise,
the covariant derivative of a vector $Y$ with respect to $\curvflow_t$ is given by
\begin{equation}\label{Gparallel.conx.flow}
\cofra\hook\nabla_t Y = (0,\p_t y+\ad(w_\pe+w_\pa)y), 
\end{equation}
while in the Sym-Pohlmeyer construction,
\begin{equation}\label{SP.conx.flow}
\phi_0(\nabla_t Y)\phi_0\inv =\p_t y-\ad(V^{(0)})y . 
\end{equation}
This yields the correspondence 
\begin{equation}\label{conx.flow.rel}
V^{(0)}= -(w_\pe + w_\pa) . 
\end{equation}

Finally,
a comparison of the flow vector $\curvflow_t$ itself in these two formulations
$\cofra\hook\curvflow_t = (0,h_\pa+h_\pe)$
and $\phi_0\curvflow_t \phi_0\inv =V^{(1)}$
gives 
\begin{equation}\label{conx.flow.rel2}
V^{(1)}=h_\pa+h_\pe . 
\end{equation}

Hence, the following correspondence holds. 

\begin{prop}\label{prop:SP.Gparallel.correspond}
The $G$-parallel frame structure equations \eqref{flow.mpar}--\eqref{flow.hperp} 
and the isospectral flow equations \eqref{ZCh0}--\eqref{ZCh1}, \eqref{ZCm0}--\eqref{ZCm1} 
are related by the transformation
\begin{equation}\label{SP.Gparallel.ident}
x = \enorminv s,
\quad
q = -\enorm Q,
\quad
w_\pe = -V^{(0)}_{\msp},
\quad
w_\pa = -V^{(0)}_{\hsp},
\quad
h_\pe = V^{(1)}_{\msp},
\quad
h_\pa = V^{(1)}_{\hsp} .
\end{equation}
\end{prop}

The main advantage of the $G$-parallel frame formulation
over the Sym-Pohlmeyer construction is that it encodes
an explicit bi-Hamiltonian structure as well as a recursion operator. 
These structures are hidden in the Lax pair approach
underlying the Sym-Pohlmeyer construction. 

The transformation \eqref{SP.Gparallel.ident} now leads to the following
important correspondence result. 

\begin{thm}\label{thm:SP.Gparallel.hierarchies}
The merged hierarchy \eqref{aff.hierarchy} of tri-Hamiltonian flows
on the Hasimoto variable $\q$ 
coincides with the hierarchy of isospectral flows \eqref{hierarchy.isoflow}
on the isospectral potential $\Q$
under the transformation \eqref{SP.Gparallel.ident}. 
In particular,
the $+2$ and $+3$ flows \eqref{Gparallel.nls} and \eqref{Gparallel.mkdv}
respectively correspond to the $H$-invariant NLS and mKdV equations 
\eqref{FK.nls} and \eqref{FK.mkdv}.
Moreover,
the non-stretching curve flow \eqref{curveflow} determined by the $+n$ flow
for $n\geq 2$ 
coincides with the Sym-Pohlmeyer curve flow
associated with the corresponding isospectral flow.
\end{thm}

In this correspondence,
the operators \eqref{Hop} and \eqref{Jop} which provide the bi-Hamiltonian structure
are related to the isospectral recursion operator \eqref{iso.Rop} by
\begin{equation}
\Hop\Jop = -\Rop^2
\end{equation}
up to an overall scaling factor.

\subsection{Covariants of non-stretching curves}

Each of the isospectral flow equations \eqref{aff.hierarchy}
is invariant under rigid (constant) transformations $\Ad(G_\pa)$
on the dynamical variable $\q(x,t)$,
where $G_\pa=H\subset G$ is the equivalence group of the $G$-parallel frame.
The explicit form of these transformations is given by  
\begin{equation}\label{equivgroup.action}
\q \to \Ad(\g)\q = \g\q\g\inv,
\quad
\g\in H\subset G. 
\end{equation}

Since the Cartan connection matrix of the $G$-parallel frame is represented by $\q$,
this variable encodes differential covariants of the curve $\curvflow(x)$,
namely, 
$\q$ is invariantly defined by $\curvflow(x)$ up to the rigid group action \eqref{equivgroup.action}.
In particular,
once any $G$-parallel frame is fixed through the soldering relation \eqref{solder}
by specifying an orthonormal basis for $\gsp\simeq\msp_\aff$,
the components of $\q$ in this basis are scalar covariants, 
and likewise the components of any $x$-derivatives of $\q$ are scalar differential covariants.
These covariants carry geometrical information
related to curvature invariants of the curve $\curvflow(x)$.
Specifically, any scalar constructed only in terms of $\q$ and $x$-derivatives of $\q$
by use of the Lie bracket and the inner product on $\gsp$
will be a curvature (differential) invariant. 

This will be illustrated more fully in \secref{sec:VFE}
when the geometrical curve flows arising from the isospectral flows \eqref{aff.hierarchy}
are obtained.

\section{Generalized vortex filament equations and Heisenberg spin models}\label{sec:VFE}

We will now derive the geometric curve flows
defined through Theorem~\ref{thm:iso.flow.eqs}
for the hierarchy of isospectral flows \eqref{aff.hierarchy}
given in Theorem~\ref{thm:iso.flow.hierarchy}.
From Theorem~\ref{thm:SP.Gparallel.hierarchies}, 
these curve flows coincide with Sym-Pohlmeyer curves
after a scaling of the arclength parameter. 
In particular,
the resulting curve flows that arise from the NLS and mKdV isospectral flows
will be shown to be close analogs of the equations in $\Rnum^3$
for the motion of a vortex filament and its axial generalization,
which will be formulated using
the tangent vector, the normal vector, and a generalized bi-normal vector
along the curves.
We will then show how these generalized vortex filament equations
give rise to associated Heisenberg spin models
for the tangent vector of the underlying curve flows.
Finally, we will present a geometrical recursion operator for
the generalized vortex filament equations as well as the associated Heisenberg spin models.

The derivation will use only some of the structure of 
an affine Hermitian symmetric space $\aff(G,H)$.
Specifically,
as a manifold, $\aff(G,H)$ is isomorphic to a Euclidean space of dimension $\dim(\gsp)$,
namely $\aff(G,H)\simeq \Rnum^{\dim(\gsp)}$;
and $\aff(G,H)\simeq \gsp$ carries the Lie bracket structure of
a Hermitian symmetric Lie algebra.

The role of a $G$-parallel framing will be seen to be essential
in providing a simple geometrical way to relate the isospectral flow equations
to geometrical vectors arising from curves in $\aff(G,H)$.
Looked at another way,
the Cartan connection-matrix variable $\q$ provides a generalized Hasimoto variable,
such that the relationship between
the isospectral flows on $\q$ and the geometric curve flows 
is a generalized Hasimoto transformation.

\subsection{Geometric preliminaries}\label{subsec:geom.curve}

Given a curve $\curvflow(x)$ in $\aff(G,H)$, 
the tangent vector is $\curvflow_x$, 
and hence $\nabla_x = g(\curvflow_x,\grad)$ 
is the directional derivative along the curve.
For $G$-parallel framed curves of algebraic type \eqref{curve}, 
the directional derivative $\nabla_x$ of vector fields 
corresponds to the Lie-algebra element in $\gsp$ 
given by the soldering relation \eqref{conx}:
\begin{equation}\label{solder.Dx}
\cofra\hook\nabla_x Y =  \Dx(\cofra\hook Y) + [(\q,0),\cofra\hook Y]
=  \Dx(0,y) + [(\q,0),(0,y)] = (0,\partial_x y +\ad(\q)y), 
\end{equation}
where $y=\cofra\hook Y$ is the Lie-algebra element corresponding to the vector $Y$. 
In particular, 
\begin{align}
\cofra\hook\nabla_x\curvflow_x & = \enorm (0,[\q,A]) =-\enorm (0,J\q),
\label{solder.Jq}
\\
\cofra\hook\nabla_x^2\curvflow_x & = -\enorm (0,J\q_x + [\q,J\q]),
\label{solder.DxJqplusqJq}
\end{align}
and so on for higher derivatives. 
Namely, $\nabla_x$ corresponds to $\partial_x + \ad(\q)$. 

Similarly, the soldering of the flow vector $\curvflow_t$ is obtained from
the representation \eqref{cofra.flow.vars} and the expressions \eqref{hperp}--\eqref{hpar}:
\begin{equation}\label{flow.rep}
\cofra\hook\curvflow_t = (0,h_\pa+h_\pe) = -\enorminv (0,Jh^\pe-\Dxinv[\q,Jh^\pe]) . 
\end{equation}

Any vector in $\aff(G,H)$ has an action on the vector space $\aff(G,H)$
through the Lie bracket on $\gsp$.
This action will be denoted by $\ad_\gsp(\cdot)$. 
For vectors $Y,Z\in\aff(G,H)$,
let $\cofra\hook Y=(0,y)$ and $\cofra\hook Z=(0,z)$
in the representation \eqref{aff.rep}.
Then $\ad_\gsp(\cdot)$ is given by 
\begin{equation}\label{ad.action}
\cofra\hook\ad_\gsp(Z)Y := (0,[z,y]) 
\end{equation}
and consequently 
\begin{equation}\label{adsq.action}
\cofra\hook\ad_\gsp(Z)^2Y := (0,[z,[z,y]]) . 
\end{equation}

In particular, $\T$ has the action
\begin{equation}\label{adT.action}
\cofra\hook\ad_\gsp(\curvflow_x)Y := \enorm (0,[A,y]) = \enorm (0,Jy) ,
\quad
\cofra\hook\ad_\gsp(\curvflow_x)^2Y := \tfrac{1}{\efac} ([0,J^2y]) = - \tfrac{1}{\efac} (0,(y)_\msp), 
\end{equation}
where $(y)_\msp$ denotes the projection of $y$ into $\msp$. 
If $Y$ belongs to the subspace in $\aff(G,H)$ corresponding to $\msp$ in $\gsp$,
namely $y\in\msp$, 
then there is an inverse for the action of $\T$:
\begin{equation}\label{adinvT.action}
\cofra\hook\ad_\gsp(\curvflow_x)\inv Y = \enorminv (0,J\inv y) = - \enorminv (0,Jy)
\end{equation}
since $J^2=-\id_\msp$. 

The action \eqref{adT.action} of $\T$ provides a geometrical way to characterize 
the algebraic property \eqref{curve} distinguishing the curve flows $\curvflow(x,t)$
in Theorem~\ref{thm:Gparallel.frame}. 

\begin{rem}
Non-stretching curve flows $\curvflow(x,t)$ of algebraic type \eqref{curve} 
have the geometrical characterization that the action of $\ad_\gsp(\curvflow_x)^2$
on $T_\curvflow\aff(G,H)\simeq\aff(G,H) \simeq \msp\oplus\hsp$
has eigenvalues $-1$ and $0$ on the respective subspaces in $T_\curvflow\aff(G,H)$
corresponding to $\msp$ and $\hsp$. 
\end{rem}

\subsection{Tangent and normal vectors, and curvatures}

The $x$-derivative of the tangent vector 
\begin{equation}\label{T}
\T=\curvflow_x
\end{equation}  
along the curve $\curvflow(x)$ 
defines the principal normal vector $\nabla_x\T=\nabla_x\curvflow_x$. 
The normal vector of the curve $\curvflow(x)$ is then given by
\begin{equation}\label{N}
\N = \tfrac{1}{\kappa} \nabla_x\T, 
\end{equation}  
where
\begin{equation}
\kappa = g(\N,\nabla_x\T)
\end{equation}    
defines the principal curvature scalar of the curve. 
This scalar is related to the norm of $\q$ by 
$\kappa^2 = g(\nabla_x\T,\nabla_x\T)
= \brack{\cofra\hook\nabla_x\curvflow_x,\cofra\hook\nabla_x\curvflow_x}_\aff
=\tfrac{1}{\efac} \brack{(0,J\q),(0,J\q)}_\aff
=-\tfrac{1}{\efac} \Kill{J\q,J\q}
= -\tfrac{1}{\efac} \Kill{\q,\q}$.
Hence,
\begin{equation}
|\kappa| = \enorm \sqrt{-\Kill{\q,\q}}
\end{equation}
which is a direct generalization of the relationship between
the Hasimoto variable and the principal curvature
for non-stretching curves in $\Rnum^3$. 

A vector that is orthogonal to both $\T$ and $\N$ can be constructed algebraically
by using $\ad_\gsp(\T)\inv$.
Note 
\begin{equation}
g(\ad_\gsp(\T)\inv \N,\T) = 0,
\quad
g(\ad_\gsp(\T)\inv \N,\N) = 0,
\quad
g(\ad_\gsp(\T)\inv \N,\ad_\gsp(\T)\inv \N) =1
\end{equation}
due to $\ad$-invariance of the inner product on $\aff(G,H)$.
Therefore,
\begin{equation}
\B := \ad_\gsp(\T)\inv \N
\end{equation}  
can be viewed as defining a bi-normal vector.
The vectors $\T,\N,\B$ correspond to the Lie-algebra elements
\begin{align}
\cofra\hook\T & = \enorm (0,A),
\label{solder.T}
\\
\cofra\hook\N & = -\tfrac{1}{\kappa\enorminv} (0,J\q),
\label{solder.N}
\\
\cofra\hook\B & = \tfrac{1}{\kappa} (0,J^2\q) = -\tfrac{1}{\kappa} (0,\q) . 
\label{solder.B}
\end{align}
This triple of vectors is a direct generalization of
the vectors in a Frenet frame for non-stretching curves in $\Rnum^3$.
However, in general $\T,\N,\B$ do not span all of $\gsp\simeq\Rnum^{\dim(\gsp)}$. 

For comparison with the Sym-Pohlmeyer construction, 
we can use Proposition~\ref{prop:soldering} to obtain 
an explicit identification between vectors $Y$ in $T_x\aff(G,H)\simeq\aff(G,H)$
and Lie-algebra elements $\cofra\hook Y$ in $\msp_\aff\simeq\gsp$.
Specifically, 
let $\cofra\hook Y=(0,y)$, where $y\in\gsp$. 
Then we have 
$(0,y) = (0,\Ad(\psi_\cofra\inv)Y)$
with $\psi_\cofra(x)\in G$,
whereby 
\begin{equation}\label{aff.vector}
\begin{aligned}
& Y \in\aff(G,H) \\& 
\simeq \Ad(\psi_\cofra)y \in\gsp . 
\end{aligned}
\end{equation}
Moreover, 
derivatives $\nabla_x$ of vector fields $Y$ in $\aff(G,H)$
can be identified with derivatives $\Dx$ of functions $\Ad(\psi_\cofra)y$ in $\gsp$: 
\begin{equation}\label{aff.grad.vector}
\begin{aligned}
& \nabla_x Y \in\aff(G,H) \\& 
\simeq \Dx(\Ad(\psi_\cofra)y) \in\gsp . 
\end{aligned}
\end{equation}
This is established by noting, on one hand, 
$\cofra\hook\nabla_x Y =  (0,\Ad(\psi_\cofra\inv)\nabla_x Y)$
from equation \eqref{aff.vector},
while on the other hand, 
$\cofra\hook\nabla_x Y =  (0,\partial_x y +\ad(\psi_\cofra\inv \partial_x\psi_\cofra)y)
= (0,\Ad(\psi_\cofra\inv)\Dx(\Ad(\psi_\cofra)y))$
through equation \eqref{solder.Dx} combined with equation \eqref{Gparallel.conx}. 

Thus, the vectors $\T,\N,\B$ can be represented as 
\begin{equation}\label{T.N.B}
\T \simeq \enorm \Ad(\psi_\cofra) A, 
\quad
\N \simeq -\tfrac{1}{\kappa\enorminv}\Ad(\psi_\cofra)J\q, 
\quad
\B \simeq -\tfrac{1}{\kappa} \Ad(\psi_\cofra)\q . 
\end{equation}

We will now derive the geometrical curve flows defined by
the NLS and mKdV isospectral flows.

\subsection{NLS curve flow and generalized vortex filament equation}

The NLS isospectral flow \eqref{Gparallel.nls} is generated by $h^\pe = -J\q$,
with $h_\pa = 0$,
as given by equation \eqref{nls.flow.vars}. 
Thus, from expression \eqref{flow.rep}, 
\begin{equation}
\cofra\hook\curvflow_t =-\enorminv (0,\q) . 
\end{equation}
Now apply the relation \eqref{adinvT.action} with $Y=\curvflow_{xx}$:
\begin{equation}\label{adinvT.N}
\cofra\hook\ad_\gsp(\curvflow_x)\inv\nabla_x\curvflow_x = -(0,\q) . 
\end{equation}
This expression directly yields
\begin{equation}
\cofra\hook\curvflow_t = \enorminv \cofra\hook\ad_\gsp(\curvflow_x)\inv\nabla_x\curvflow_x
\end{equation}
which gives the equation of motion
\begin{equation}\label{nls.aff.curve}
\enorm \curvflow_t = \ad_\gsp(\curvflow_x)\inv\nabla_x\curvflow_x . 
\end{equation}
After scaling $t$,
this equation of motion can be expressed as
\begin{equation}\label{nls.aff.VFE}
\curvflow_t = \kappa\ad_\gsp(\T)\inv \N = \kappa\B
\end{equation}
in terms of the tangent vector, normal vector, and the principal curvature scalar.

Hence the NLS isospectral flow corresponds to a bi-normal equation of motion \eqref{nls.aff.VFE}
for a curve $\curvflow$ in $\aff(G,H)$.
This geometric flow equation is analogous to the vortex filament equation in $\Rnum^3$ \cite{Has}.
In the form \eqref{nls.aff.curve},
it coincides with the Sym-Pohlmeyer curve flow \eqref{SP.nls.curveflow}
after the scaling \eqref{s} of $x$.

\subsection{mKdV curve flow and generalized vortex filament axial equation}

The mKdV isospectral flow \eqref{Gparallel.mkdv} is generated by $h^\pe = -\q_x$.
This flow has $h_\pa = -\enorminv \tfrac{1}{2} [\q,J\q]$,
from equation \eqref{mkdv.flow.vars}. 
Thus 
\begin{equation}
\cofra\hook\curvflow_t =\enorminv (0,J\q_x -\tfrac{1}{2}[\q,J\q]) . 
\end{equation}
To relate this expression to $x$-derivatives of $\curvflow_x$,
apply the relation \eqref{adsq.action} with $Y=\nabla_x^2\curvflow_x$ and $Z=\curvflow_x$:
$\cofra\hook\ad_\gsp(\curvflow_x)^2\nabla_x^2\curvflow_x =
\tfrac{1}{\enorminv^3} (0,J\q_x)$.
Then this expression can be combined with expression \eqref{solder.DxJqplusqJq}
to get
\begin{equation}
\cofra\hook(1+\efac\ad_\gsp(\curvflow_x)^2)\nabla_x^2\curvflow_x =
 -\enorm (0,[\q,J\q]) . 
\end{equation}
Hence,
\begin{equation}
\cofra\hook\curvflow_t = \tfrac{1}{2}\efac \cofra\hook (1+3\efac\ad_\gsp(\curvflow_x)^2)\nabla_x^2\curvflow_x
\end{equation}
which gives the equation of motion
\begin{equation}
\tfrac{1}{\efac} \curvflow_t = \tfrac{1}{2}(1+3\efac\ad_\gsp(\curvflow_x)^2)\nabla_x^2\curvflow_x . 
\end{equation}
After scaling $t$,
this equation of motion can be expressed as
\begin{equation}\label{mkdv.aff.VFE.alt}
\curvflow_t = \tfrac{1}{2}(1+3\efac\ad_\gsp(\T)^2)\nabla_x(\kappa\N)
\end{equation}
in terms of the tangent vector, normal vector, and the principal curvature scalar.

There is an alternative expression for the equation of motion \eqref{mkdv.aff.VFE.alt}. 
Consider the relation \eqref{adsq.action} with $Y=\curvflow_x$ and $Z=\nabla_x\curvflow_x$:
\begin{equation}
\cofra\hook \ad_\gsp(\nabla_x\curvflow_x)^2\curvflow_x
=\tfrac{1}{\sqrt{\efac}^3} (0,[J\q,[J\q,A]])
= -\tfrac{1}{\sqrt{\efac}^3} (0,[\q,J\q]) . 
\end{equation}
Then this expression combined with expression \eqref{solder.DxJqplusqJq}
yields
\begin{equation}
\cofra\hook \curvflow_t
= \efac\cofra\hook( -\nabla_x^2\curvflow_x +\tfrac{3}{2}\efac\ad_\gsp(\nabla_x\curvflow_x)^2\curvflow_x ) . 
\end{equation}
After scaling $t$,
this equation of motion can be expressed as
\begin{equation}\label{mkdv.aff.VFE}
\curvflow_t = -\nabla_x(\kappa\N) +\tfrac{3}{2}\efac\kappa^2 \ad_\gsp(\N)^2\T,
\quad
\kappa^2=g(\nabla_x\T,\nabla_x\T) 
\end{equation}
in terms of the tangent vector, normal vector, and the principal curvature scalar.

Thus the mKdV isospectral flow corresponds to the geometrical equation of motion \eqref{mkdv.aff.VFE}
for a curve $\curvflow$ in $\aff(G,H)$.
This flow equation is analogous to the equation in $\Rnum^3$
describing axial flow of vortex filament \cite{FukMiy}.
When it is expressed in the form
\begin{equation}\label{SP.mkdv}
\curvflow_t = -\nabla_x^2\curvflow_x +\tfrac{3}{2}\efac\ad_\gsp(\nabla_x\curvflow_x)^2\curvflow_x,
\end{equation}
this coincides with the Sym-Pohlmeyer curve flow \eqref{SP.mkdv.curveflow} 
after the scaling \eqref{s} of $x$.

\subsection{Spin vector models in $\gsp$}

Any geometric non-stretching curve flow $\curvflow_t=\Pvec$, 
with the vector $\Pvec$ given in terms of $\curvflow_x$ and its $x$-covariant derivatives, 
induces an associated flow on the tangent vector $\T=\curvflow_x$:
\begin{equation}\label{spinmodel}
\T_t = \nabla_x \Pvec,
\quad
|\T| =1, 
\end{equation}
where, for notational ease,
\begin{equation}\label{notation}
Y_t:=\nabla_t Y,
\quad
Y_x:=\nabla_x Y,
\quad
Y\cdot X:=g(Y,X),
\quad
|Y|^2 :=g(Y,Y) . 
\end{equation}
This induced flow equation \eqref{spinmodel}
can be naturally viewed as defining a spin vector model,
since $\T$ has a unit norm that is preserved in the flow. 

The spin vector model arising from the NLS curve flow \eqref{nls.aff.VFE}
is given by 
\begin{equation}\label{nls.aff.spinmodel}
\T_t = -\ad_\gsp(\T)\T_{xx} ,
\quad
|\T|=1 . 
\end{equation}
This equation describes a generalized Heisenberg spin model.

The mKdV curve flow \eqref{mkdv.aff.VFE} yields the spin vector model
\begin{equation}\label{mkdv.aff.spinmodel}
\T_t = -\T_{xxx} +3\efac\ad_\gsp(\T_{xx})\ad_\gsp(\T_x)\T,
\quad
|\T|=1,
\end{equation}
after the identity $\ad_\gsp(\T)[\T_x,\T_{xx}]=0$ has been used.

\subsection{Recursion structure for geometric curve flow equations and spin vector models in $\gsp$}

Theorems~\ref{thm:iso.flow.eqs} and~\ref{thm:SP.Gparallel.hierarchies}
yield a hierarchy of geometric non-stretching curve flows,
starting from generalized vortex filament equation \eqref{nls.aff.VFE}. 
This hierarchy inherits a geometric recursion operator 
from the isospectral recursion operator \eqref{Rop}. 

Recall from equations \eqref{curveflow}, \eqref{aff.hierarchy.hperp}, \eqref{aff.hierarchy.hpar}
that each curve flow $\curvflow_t = \Pvec^{(n-1)}$, $n=2,3,\ldots$, in the hierarchy
is determined by the variable $h_\pe = \cofra\hook \Pvec_\pe^{(n-1)}$
giving the motion of the curve projected into the subspace 
in $\aff(G,H)$ corresponding to $\msp\subset\gsp$. 
A recursion operator on this variable arises from its relationship \eqref{hperp}
to the variable $h^\pe$ that determines the isospectral flow corresponding to the curve flow.
Specifically, consider a variation $\delta h^\pe = \Rop(h^\pe)$.
The variation of the relation \eqref{hperp} yields 
$\delta h_\pe = -J\delta h^\pe = -J\Rop(h^\pe) = -J\inv\Hop J(h^\pe) = \Rop^*(h_\pe)$,
where $\Rop^*=J\inv\Hop$ is the adjoint of the isospectral recursion operator $\Rop$.
This relationship shows that $\Rop^*$ acts as a recursion operator on $h_\pe$.
Then the soldering equation 
\begin{equation}
\Rop^*(h_\pe) = \cofra\hook \Rop_\curvflow(\Pvec_\pe)
\end{equation}
determines a corresponding recursion operator on the vector $\Pvec_\pe$.

The explicit geometric form of this recursion operator $\Rop_\curvflow$
can be derived from expression \eqref{adjRop} for $\Rop^*$ as follows. 
Consider
$\Rop^*(h_\pe) = (\Dx -\ad(J\q)\Dxinv \ad(J\q))J\inv h_\pe$. 
First, the local term in $\Rop^*(h_\pe)$ can be expressed as
$\Dx(J\inv h_\pe) = J\inv \Dx h_\pe= J\inv (\Dx h_\pe +\ad(\q)h_\pe)$, 
since $\ad(\q)h_\pe$ belongs to $\hsp\subset\gsp$ which is annihilated by $J$.
Hence,
\begin{equation}\label{Rop.curve.locterm}
\Dx(J\inv h_\pe) = \cofra\hook (\ad_\gsp(\curvflow_x)\inv \nabla_x \Pvec_\pe)
\end{equation}
using the soldering relation \eqref{adinvT.action}.
Next, the nonlocal term in $\Rop^*(h_\pe)$ can be expressed as
$-\ad(J\q)\Dxinv (\ad(J\q)J\inv h_\pe) = -\ad(J\q)h_\pa$
through the relation \eqref{flow.mpar} between $h_\pa$ and $h_\pe$,
combined with the identity $\ad(J\q)J\inv h_\pe= -\ad(\q)h_\pe$
due to the properties of $J$. 
This yields 
$\Dx h_\pa = \ad(J\q)J\inv h_\pe =\cofra\hook (-\ad_\gsp(\curvflow_{xx})\ad_\gsp(\curvflow_x)\inv \Pvec_\pe)$
by the soldering relations \eqref{solder.Jq} and \eqref{adinvT.action},
while
$\Dx h_\pa = (\Dx h_\pa +\ad(\q)h_\pa)_\hsp =\cofra\hook \nabla^\hsp_x \Pvec_\pa$,
with $h_\pa = \cofra\hook \Pvec_\pa$.
Here
\begin{equation}\label{curve.h.der}
\nabla^\hsp_x = (1 +\ad_\gsp(\curvflow_x)^2)\nabla_x
\end{equation}
denotes the projection of $\nabla_x$ into 
the subspace in $\aff(G,H)$ corresponding to $\hsp\subset\gsp$. 
Consequently,
$\nabla^\hsp_x \Pvec_\pa = -\ad_\gsp(\curvflow_{xx})\ad_\gsp(\curvflow_x)\inv \Pvec_\pe$
which yields
\begin{equation}\label{Ppar}
\Pvec_\pa = -(\nabla^\hsp_x)\inv (\ad_\gsp(\curvflow_{xx})\ad_\gsp(\curvflow_x)\inv \Pvec_\pe) . 
\end{equation}  
Hence,
\begin{equation}\label{Rop.curve.nonlocterm}
-\ad(J\q)\Dxinv \ad(J\q)J\inv h_\pe = \cofra\hook (\ad_\gsp(\curvflow_{xx})\Pvec_\pa)
= \cofra\hook ( -\ad_\gsp(\curvflow_{xx})(\nabla^\hsp_x)\inv (\ad_\gsp(\curvflow_{xx})\ad_\gsp(\curvflow_x)\inv \Pvec_\pe) ) . 
\end{equation}  
Combining the nonlocal term \eqref{Rop.curve.nonlocterm} and the local term \eqref{Rop.curve.locterm}
now yields
\begin{equation}
\Rop^*(h_\pe) = \cofra\hook ( \ad_\gsp(\curvflow_x)\inv \nabla_x \Pvec_\pe - \ad_\gsp(\curvflow_{xx})(\nabla^\hsp_x)\inv (\ad_\gsp(\curvflow_{xx})\ad_\gsp(\curvflow_x)\inv \Pvec_\pe) ), 
\end{equation}  
and thus
\begin{equation}\label{Rop.curve}
\Rop_\curvflow = \ad_\gsp(\curvflow_x)\inv \nabla_x - \ad_\gsp(\curvflow_{xx})(\nabla^\hsp_x)\inv \ad_\gsp(\curvflow_{xx})\ad_\gsp(\curvflow_x)\inv
\end{equation}  
gives the corresponding geometric recursion operator.


A similar derivation produces a recursion operator
for the hierarchy of spin vector models
arising from the hierarchy of geometric non-stretching curve flows.

Each spin vector model $\T_t = \nabla_x \Pvec^{(n-1)}$, $n=2,3,\ldots$, in the hierarchy
is determined by the variable
$\cofra\hook \nabla_x \Pvec^{(n-1)} = (\Dx + \ad(\q))(h_\pe +h_\pa)$.
By means of the frame structure equations \eqref{flow.mpar}--\eqref{flow.mperp} and \eqref{hperp}, 
this variable can be expressed as
\begin{equation}
\cofra\hook \nabla_x \Pvec^{(n-1)} = \Dx h_\pe + \ad(\q)h_\pa
= -\enorm J w_\perp = -\enorminv \Rop(h^\pe)
\end{equation}
on which $\Rop$ itself acts as a recursion operator.
The geometric form of $\Rop$ can be obtained through the soldering equation 
\begin{equation}\label{solder.R.q.T}
\Rop(h^\pe) = \cofra\hook \Rop_\T(\Pvec^\pe), 
\end{equation}
where $\Pvec^\pe$ is a vector defined by $\cofra\hook \Pvec^\pe = h^\pe$
in analogy to $\Pvec_\pe$.
Now consider
$\Rop(h^\pe) =J\inv(\Dx -\ad(J\q)\Dxinv \ad(J\q))h^\pe$. 
The local term in $\Rop(h^\pe)$ can be expressed as
$J\inv \Dx h^\pe = J\inv (\Dx h^\pe +\ad(\q)h^\pe)$, 
and hence,
\begin{equation}\label{Rop.spinmodel.locterm}
J\inv \Dx h^\pe = \cofra\hook (\ad_\gsp(\T)\inv \nabla_x \Pvec^\pe)
\end{equation}
similarly to the derivation of equation \eqref{Rop.curve.locterm}. 
Likewise, the nonlocal term in $\Rop(h^\pe)$ can be expressed as
$-J\inv\ad(J\q)\Dxinv (\ad(J\q)h^\pe) = J\inv\ad(J\q)h_\pa$,
which yields 
$\Dx h_\pa = -\ad(J\q)h^\pe =\cofra\hook (\ad_\gsp(\curvflow_{xx})\Pvec^\pe)$,
while
$\Dx h_\pa = (\Dx h_\pa +\ad(\q)h_\pa)_\hsp =\cofra\hook \nabla^\hsp_x \Pvec_\pa$
as before. 
Hence,
\begin{equation}\label{Rop.spinmodel.nonlocterm}
-J\inv\ad(J\q)\Dxinv \ad(J\q) h^\pe = \cofra\hook (-\ad_\gsp(\T)\inv\ad_\gsp(\T_x)(\nabla^\hsp_x)\inv (\ad_\gsp(\T_x) \Pvec^\pe)) . 
\end{equation}  
Combining the two terms \eqref{Rop.spinmodel.nonlocterm} and \eqref{Rop.spinmodel.locterm}
now yields
\begin{equation}
\Rop(h^\pe) = \cofra\hook \ad_\gsp(\T)\inv (\nabla_x \Pvec^\pe -\ad_\gsp(\T_x)(\nabla^\hsp_x)\inv (\ad_\gsp(\T_x) \Pvec^\pe) ) , 
\end{equation}  
and thus
\begin{equation}\label{Rop.spinmodel}
\Rop_\T = \ad_\gsp(\T)\inv( \nabla_x -\ad_\gsp(\T_x)(\nabla^\hsp_x)\inv \ad_\gsp(\T_x) )
\end{equation}  
gives the corresponding geometric recursion operator.


Now observe that an $x$-translation generated by 
$\X_{\text{trans.}} =-\curvflow_x\hook\p_\curvflow=-\T_x\hook\p_\T$
is a symmetry of the recursion operators $\Rop_\curvflow$ and $\Rop_\T$.
Applying these respective operators to the corresponding vectors $\curvflow_x$ and $\T_x$
yields
\begin{align}
\Rop_\curvflow(\curvflow_x) & = \ad_\gsp(\curvflow_x)\inv \curvflow_{xx},
\\
\Rop_\T(\T_x)  &= \ad_\gsp(\T)\inv \T_{xx},
\end{align}
which are precisely the righthand sides of the generalized vortex filament equation \eqref{nls.aff.VFE}
and the associated Heisenberg spin model \eqref{nls.aff.spinmodel}. 
This leads to the following result. 

\begin{thm}\label{thm:main.hierarchy}
The hierarchy of isospectral flows \eqref{aff.hierarchy}, with $n\geq 1$, 
on the Hasimoto variable $\q$ 
corresponds to the hierarchy of geometric non-stretching curve flows
\begin{equation}\label{hierarchy.curveflows.perp}
(\curvflow_t)_\pe = \Pvec_\pe^{(k)} = \Rop_\curvflow^{k}(\curvflow_x),
\quad
|\curvflow_x|=1,
\quad
k=0,1,2,\ldots
\end{equation}
and the associated hierarchy of spin vector models
\begin{equation}\label{hierarchy.spinmodels}
\T_t = \Pvec^\pe_{(k)} = \Rop_\T^{k}(\T_x),
\quad
|\T|=1,
\quad
k=0,1,2,\ldots, 
\end{equation}
with $k=n-1\geq 0$, 
starting from $(\curvflow_t)_\pe =\Pvec_\pe^{(0)}=\curvflow_x$
and $\T_t=\Pvec^\pe_{(0)}=\T_x$,
where $\Rop_\curvflow$ and $\Rop_\T$ are
hereditary recursion operators \eqref{Rop.curve}
and \eqref{Rop.spinmodel}. 
In particular, the $k=1$ curve flow and associated spin vector model 
are respectively given by 
the generalized vortex filament equation \eqref{nls.aff.VFE} 
and associated Heisenberg spin model \eqref{nls.aff.spinmodel},
which correspond to the $n=2$ isospectral flow 
given by the $H$-invariant NLS equation \eqref{Gparallel.nls} 
on the Hasimoto variable $\q$. 
\end{thm}

The hereditary property of the geometric recursion operators $\Rop_\curvflow$ and $\Rop_\T$
can be established from the general theory of recursion operators \cite{Olv-book}
by noting that the isospectral recursion operator $\Rop$ is hereditary
due to its factorized form \eqref{Rop} 
in terms of the pair of compatible Hamiltonian operators $J$ and $\Hop$.

\begin{rem}
The hierarchy \eqref{hierarchy.curveflows.perp}
describes the motion of the curves $\curvflow(x,t)$ 
projected into the subspace in $\aff(G,H)$ corresponding to $\msp\subset\gsp$. 
The motion in the full space $\aff(G,H)$ is given by
\begin{equation}\label{hierarchy.curveflows}
\curvflow_t = \Pvec_\pe^{(k)} +\Pvec_\pa^{(k)},
\quad
\Pvec_\pa^{(k)} = -(\nabla^\hsp_x)\inv (\ad_\gsp(\curvflow_{xx})\ad_\gsp(\curvflow_x)\inv \Pvec_\pe^{(k)}),
\quad
k=0,1,2,\ldots , 
\end{equation}
which is obtained from equation \eqref{Ppar} relating $\Pvec_\pa$ to $\Pvec_\pe$.
This hierarchy \eqref{hierarchy.curveflows} of non-stretching curve flows
is related to the hierarchy \eqref{hierarchy.spinmodels} of spin vector models
by
\begin{equation}
\Pvec^\pe_{(k)} = \nabla_x(\Pvec_\pe^{(k)} +\Pvec_\pa^{(k)}) . 
\end{equation}
\end{rem}

\begin{rem}
Since both $\Rop_\curvflow$ and $\Rop_\T$ are hereditary,
all of the higher curve flows and spin vector models 
in the hierarchies \eqref{hierarchy.curveflows} and \eqref{hierarchy.spinmodels}
can be viewed as higher symmetries
$\X^{(k)} = \Pvec^{(k)}\hook\p_\curvflow = \Pvec^\pe_{(k)}\hook\p_\T$, $k\geq 2$,
of the generalized vortex filament equation \eqref{nls.aff.VFE}
and the associated Heisenberg spin model \eqref{nls.aff.spinmodel}. 
\end{rem}

\section{Tri-Hamiltonian structures}\label{sec:biHam}

We will now show that the geometric curve flows and the associated spin vector models
obtained in Theorem~\ref{thm:main.hierarchy}
have a natural tri-Hamiltonian structure.
This structure arises from factorization of the geometric recursion operators that 
generate the hierarchies,
and it can be understood as a geometrical version of the tri-Hamiltonian form \eqref{aff.hierarchy.biHam}
of the corresponding isospectral flows.

To begin, recall that the isospectral recursion operator \eqref{Rop}
has the factorized form $\Rop=\Hop J\inv$
with $\Hop$ and $J$ being compatible Hamiltonian operators with respect to the Hasimoto variable $\q$.
In general,
a linear operator $\Dop$ is a Hamiltonian operator if it is skew and has vanishing Schouten bracket.
A concrete formulation of these two properties in the present context of Lie-algebra valued variables
is stated in \Ref{Anc2008},
generalizing the standard formulation for real variables given in \Ref{Olv-book}
using the calculus of multi-vectors.

It will turn out to be simplest to start first by considering the spin vector models \eqref{hierarchy.spinmodels}
and afterwards return to the geometric non-stretching curve flows \eqref{hierarchy.curveflows},
because of the simple form of the soldering relation \eqref{solder.R.q.T} 
between the isospectral recursion operator $\Rop$ and the spin vector recursion operator $\Rop_\T$.

\subsection{Tri-Hamiltonian form of generalized spin vector models}

In the context of spin vector models \eqref{hierarchy.spinmodels}, 
the counterpart of the operator $J$ is given by 
\begin{equation}\label{Jop.spinmodel}
J_\T = \ad_\gsp(\T) . 
\end{equation}
This operator has the algebraic characterization that 
$\enorminv J_\T$ is like a square-root of the projection operator $\Pop_\msp$
onto the subspace in $\aff(G,H)$ corresponding to $\msp\subset\gsp$:
\begin{equation}\label{projection.msp}
\efac J_\T^2 = -\Pop_\msp,
\quad
\Pop_\msp^2 = \Pop_\msp, 
\quad
\Pop_\msp \msp =\msp,
\quad
\Pop_\msp \hsp =0 . 
\end{equation}

It is straightforward to prove that 
$J_\T$ is a Hamiltonian operator with respect to the spin vector variable $\T$
by using a geometrical version of the method used in \Ref{Anc2008}
to show that $\Hop$ is a Hamiltonian operator with respect to the Hasimoto variable $\q$.
The general theory of Hamiltonian operators \cite{Olv-book}
then shows that $\Rop_\T J_\T:=\Hop_\T$ will be a compatible Hamiltonian operator,
where $\Rop_\T$ is the hereditary recursion operator \eqref{Rop.spinmodel}.
This yields
\begin{equation}\label{Hop.spinmodel}
\Hop_\T=-\ad_\gsp(\T)( \nabla_x -\ad_\gsp(\T_x)(\nabla^\hsp_x)\inv \ad_\gsp(\T_x) )\ad_\gsp(\T) . 
\end{equation}  

For writing down the Hamiltonian structure for the hierarchy of spin vector models 
in Theorem~\ref{thm:main.hierarchy}, 
it will be useful to introduce some additional spin vector operators. 
Firstly, 
$\grad^\hsp$ and $\grad^\msp$ 
will be the respective projections of the contravariant derivative $\grad$ 
as defined by 
$X\cdot\grad^\hsp = (1+\ad_\gsp(\T)^2)(X\cdot\grad)$
and 
$X\cdot\grad^\msp = -\ad_\gsp(\T)^2(X\cdot\grad)$, 
for all vectors $X\in\aff(G,H)$. 
Secondly, 
$\ad^*_\gsp:\gsp\to \wedge^2\gsp$ 
will be the skew tensor associated with the adjoint operator $\ad_\gsp:\gsp\to \End(\gsp)$
as defined by 
$\brack{\ad^*_\gsp(X),Y\wedge Z} = \brack{Y,\ad_\gsp(X)Z}$,
for all vectors $X,Y,Z\in\aff(G,H)$. 

Now, from Theorem~\ref{thm:iso.biHamil.hiearchy}, 
the following main result is obtained. 

\begin{thm}\label{thm:main.hierarchy.spinmodel}
The hierarchy of spin vector models \eqref{hierarchy.spinmodels} 
arising from the isospectral flow hierarchy \eqref{aff.hierarchy}
has a tri-Hamiltonian structure. 
For $k\geq 2$, this structure is given by 
\begin{equation}\label{hierarchy.spinmodel.biHam.k>1}
\T_t =\Jop_\T(\delta H^{(k)}/\delta \T) =\Hop_\T(\delta H^{(k-1)}/\delta \T)
=\Eop_\T(\delta H^{(k-2)}/\delta \T), 
\quad
k=2,3,\ldots,
\end{equation}
where $\Eop_\T=\Rop_\T\Hop_\T$ 
is a third Hamiltonian operator compatible with $J_\T$ and $\Hop_\T$, 
and where the Hamiltonians are given by the functionals
\begin{equation}\label{spinmodel.Ham.k}
H^{(k)}  = \frac{1}{k}\int_C \Dxinv(\Rop_\T^k(\T_x)\cdot(\ad_\gsp(\T)\T_x))\; dx, 
\quad
k=1,2,\ldots 
\end{equation}
and
\begin{equation}\label{spinmodel.Ham.k=0}
H^{(0)}  = \int_C \vec\xi(\T)\cdot\T_x \;dx
\end{equation}
on the domain $C=\Rnum$ or $S^1$. 
Here $\vec\xi(\T)$ is any vector function in $\aff(G,H)$ satisfying 
\begin{equation}\label{spinmodel.geom.vec}
\grad^\msp\wedge \vec\xi = \ad^*_\gsp(\T) . 
\end{equation}
For $k=1$, the tri-Hamiltonian structure consists of 
\begin{equation}\label{hierarchy.spinmodel.biHam.k=1}
\T_t =\Jop_\T(\delta H^{(1)}/\delta \T) =\Hop_\T(\delta H^{(0)}/\delta \T)
=\Eop_\T(\delta H_0/\delta \T)
\end{equation}
with $H_0=\const$ being a trivial Hamiltonian, 
and where $\nabla^\hsp_x$ is generalized to have a non-trivial cokernel 
\begin{equation}\label{hDxinv}
(\nabla^\hsp_x)\inv(0)= -\ad_\gsp(\T) . 
\end{equation}
\end{thm}

\begin{proof}
The first two Hamiltonian structures \eqref{hierarchy.spinmodel.biHam.k>1} for $k\geq2$
come from the Hamiltonian form of the corresponding isospectral flows 
in Theorem~\ref{thm:iso.biHamil.hiearchy}
by using a derivation similar to that for the geometric recursion operators 
in Theorem~\ref{thm:main.hierarchy}. 

The third Hamiltonian structure \eqref{hierarchy.spinmodel.biHam.k>1} for $k\geq2$
arises from the property 
\begin{equation}\label{spinmodel.varH0}
\delta H^{(0)}/\delta \T = \ad_\gsp(\T)\inv\T_x
\end{equation}
for the Hamiltonian \eqref{spinmodel.Ham.k=0}. 
This is shown by considering an arbitrary variation of $\vec\xi(\T)\cdot\T_x$.
First, the variation of $\vec\xi(\T)$ is given by 
$\delta\vec\xi(\T) = \nabla_{\delta\T}\vec\xi(\T) = (\delta\T\cdot\grad)\vec\xi(\T)$, 
and hence
$T_x\cdot\delta\vec\xi(\T) = g(\delta\T\otimes \T_x,\grad\vec\xi(\T))$.
Next, the variation of $\T_x$ is given by 
$\delta\T_x = \nabla_x\delta\T$,
and thus $\vec\xi(\T)\cdot\delta\T_x = -\delta\T\cdot \nabla_x\vec\xi(\T)$
modulo a total $x$-derivative,
where $\nabla_x\vec\xi(\T) = (\T_x\cdot\grad)\vec\xi(\T)$. 
Hence 
$\vec\xi(\T)\cdot\delta\T_x = -g(\T_x\otimes \delta\T,\grad\vec\xi(\T))$.
Combining these variations then yields 
\begin{equation}
\delta(\vec\xi(\T)\cdot\T_x)
= g(\delta\T\wedge\T_x,\grad\vec\xi(\T))
= g(\delta\T\otimes\T_x,\grad\wedge\vec\xi(\T)) , 
\end{equation}
and consequently, by using equation \eqref{spinmodel.geom.vec}, 
\begin{equation}
\delta(\vec\xi(\T)\cdot\T_x)
= g(\delta\T\otimes\T_x,\ad^*_\gsp(\T))
= -g(\delta\T,\ad_\gsp(\T)\T_x), 
\end{equation}
which establishes the result \eqref{spinmodel.varH0}. 

For $k=1$, 
the third Hamiltonian structure \eqref{hierarchy.spinmodel.biHam.k=1}
follows from 
\begin{equation}
\Hop_\T(0) = \ad_\gsp(\T)\ad_\gsp(\T_x)(\nabla^\hsp_x)(0) = -\ad_\gsp(\T)^2\T_x = \T_x,
\end{equation}
whereby 
$\Eop_\T(0) = \Rop_\T(\T_x) = \ad_\gsp(\T)\inv \T_{xx}$
is equal to 
$\Hop_\T(\ad_\gsp(\T)\inv\T_x) = -\ad_\gsp(\T)\nabla_x(\T_x) = \ad_\gsp(\T)\inv \T_{xx}$. 

Note that the generalization \eqref{hDxinv} of $\nabla^\hsp_x$ is well-defined 
because $\T$ belongs to the kernel of this operator. 
Specifically, 
as shown by $\cofra\hook\T_x = -\tfrac{1}{\enorminv} (0,J\q)$
from equations \eqref{N} and \eqref{solder.N}, 
the vector $\nabla_x\T=\T_x$ belongs to the subspace in $\aff(G,H)\simeq\gsp$ 
corresponding to $\msp\subset\gsp$, 
and thus $\nabla^\hsp_x\T = \Pop_\hsp\T_x=0$. 
Here $\Pop_\hsp=\id-\Pop_\msp$ is the projection operator onto the orthogonal subspace
in $\aff(G,H)\simeq\gsp$ corresponding to $\hsp\subset\gsp$,
where $\Pop_\msp$ is the projection operator \eqref{projection.msp}. 
\end{proof}

As a result, 
the general Heisenberg spin model \eqref{nls.aff.spinmodel}
has the tri-Hamiltonian form 
\begin{equation}\label{nls.spinmodel.Hamil}
\T_t =\ad_\gsp(\T)\inv\T_{xx}
= \Jop_\T(\delta H^{(1)}/\delta \T)
= \Hop_\T(\delta H^{(0)}/\delta \T)
=\Eop_\T(0) . 
\end{equation}
The second and third Hamiltonian structures presented here appear to be new. 

In contrast, 
the higher-order spin vector model \eqref{mkdv.aff.spinmodel} 
has the tri-Hamiltonian form 
\begin{equation}\label{mkdv.spinmodel.Hamil}
\T_t = -\T_{xxx} +3\efac\ad_\gsp(\T_{xx})\ad_\gsp(\T_x)\T
= \Jop_\T(\delta H^{(2)}/\delta \T) 
=\Hop_\T(\delta H^{(1)}/\delta \T) 
= \Eop_\T(\delta H^{(0)}/\delta \T) , 
\end{equation}
which holds without needing the generalized cokernel \eqref{hDxinv}. 

A useful remark is that 
the Hamiltonians \eqref{spinmodel.Ham.k} in these spin vector models for $k\geq2$
correspond to the Hamiltonians \eqref{aff.Ham.n} in the isospectral hierarchy
with $\q$ and its $x$-derivatives 
expressed in terms of $\T$ and its $x$-derivatives
through the soldering relations \eqref{solder.T}, \eqref{adT.action}, and \eqref{solder.Dx}.
In particular,
the first few Hamiltonians are explicitly given by 
\begin{align}
H^{(1)} & = \int_C \tfrac{1}{2} |\T_x|^2\;dx ,
\label{spinmodel.Ham.k=1}\\
H^{(2)} &= \int_C \tfrac{1}{2}\T_{xx}\cdot (\ad_\gsp(\T)\T_x)\;dx,
\label{spinmodel.Ham.k=2}\\
H^{(3)} &= \int_C \tfrac{1}{2}|\T_{xx}|^2 -\tfrac{5}{8}|\ad_\gsp(\T_x)^2\T|^2 \;dx , 
\label{spinmodel.Ham.k=3}\\
H^{(4)} &= \int_C \tfrac{1}{2}(\ad_\gsp(\T)\T_{xx})\cdot\T_{xxx} + \tfrac{7}{8} (\ad_\gsp(\T_{x})^2\T)\cdot(\ad_\gsp(\T_x)\T_{xx}) \;dx . 
\label{spinmodel.Ham.k=4}
\end{align}
A direct derivation of all of the spin vector Hamiltonians $H^{(k)}$ for $k\geq1$ 
can be given by applying a general scaling formula in \Ref{Anc2003} (see also \Ref{Anc2008}), 
with $x\rightarrow e^\epsilon x$, $\T \rightarrow \T$ ($\epsilon\in\Rnum$)
being the scaling group. 
This formula is not applicable for $k=0$ because the spin vector Hamiltonian 
$H^{(0)}$ is scaling invariant. 
Instead, $H^{(0)}$ can be obtained by using a general homotopy formula 
in \Ref{AncBlu2002b,Anc-review} (see also \Ref{Anc2008}).
A more explicit expression for $H^{(0)}$ will be derived in \secref{sec:geommap}.

\subsection{Tri-Hamiltonian form of generalized vortex filament equations}

It is straightforward to show that 
neither $\ad_\gsp(\curvflow_x)$ nor $\nabla_x=\curvflow_x\cdot\grad$
are Hamiltonian operators for the geometric non-stretching curve flows \eqref{hierarchy.curveflows},
and so the previous steps used to obtain a pair of compatible Hamiltonian operators
for spin vector models will not work.

Instead, the spin vector Hamiltonian operators $J_\T$ and $\Hop_\T$ 
can be transformed into corresponding operators $J_\curvflow$ and $\Hop_\curvflow$
via the relation $\T = \nabla_x \curvflow$, 
which will provide a pair of compatible Hamiltonian operators for geometric non-stretching curve flows. 

\begin{lem}\label{lem:var.spinvector.curve.rels}
(i) 
For any Hamiltonian functional $H$,
\begin{equation}\label{varder.spinvector.curve}
\frac{\delta H}{\delta \T} = -\nabla_x^{-1}\frac{\delta H}{\delta \curvflow}
\end{equation}  
is a variational identity.
(ii) 
If $\Dop_\T$ is a Hamiltonian operator with respect to $\T$,
then 
\begin{equation}\label{Dop.mapping.spinvector.curve}
\Dop_\curvflow = -\nabla_x^{-1}\Dop_\T\nabla_x^{-1}
\end{equation}
is a Hamiltonian operator with respect to $\curvflow$.
\end{lem}

\begin{proof}
Consider an arbitrary  variation $\delta\T = \nabla_x\delta\curvflow$.
Then for any Hamiltonian $H$,
since 
$\delta H = (\delta H/\delta \curvflow)\cdot\delta\curvflow = (\delta H/\delta \T)\cdot\delta\T = (\delta H/\delta \T)\cdot \nabla_x \delta\curvflow$
holds modulo a total $x$-derivative,
this directly yields the variational identity \eqref{varder.spinvector.curve}
after integration by parts. 
Now suppose that $\Dop_\T$ is Hamiltonian operator with respect to $\T$,
and consider $\T_t = \Dop_\T(\delta H/\delta \T)$.
Use of the previous variational identity combined with $\T_t=\nabla_x \curvflow_t$
gives $\curvflow_t =\Dop_\curvflow(\delta H/\delta \curvflow)$,
where $\Dop_\curvflow$ is the operator \eqref{Dop.mapping.spinvector.curve}.
The proof that this operator is Hamiltonian
can be shown to reduce to the proof that $\Dop_\T$ is a Hamiltonian operator, 
analogously to a canonical transformation, via $\T = \nabla_x\curvflow$.
\end{proof}

As a result, 
the Hamiltonian operators $J_\T$ and $\Hop_\T$ give rise to the corresponding Hamiltonian operators 
\begin{gather}
J_\curvflow = -\nabla_x^{-1}\ad_\gsp(\curvflow_x)\nabla_x^{-1} , 
\label{Jop.curve}
\\
\Hop_\curvflow=\nabla_x^{-1}\ad_\gsp(\curvflow_x)( \nabla_x -\ad_\gsp(\curvflow_{xx})(\nabla^\hsp_x)\inv \ad_\gsp(\curvflow_{xx}) )\ad_\gsp(\curvflow_x)\nabla_x^{-1} . 
\label{Hop.curve}
\end{gather}  
These two Hamiltonian operators \eqref{Jop.curve} and \eqref{Hop.curve} are compatible,
since $J_\T$ and $\Hop_\T$ are compatible. 

Now the Hamiltonian structure for the hierarchy of non-stretching curve flows 
in Theorem~\ref{thm:main.hierarchy}
can be obtained from the Hamiltonian form of the corresponding isospectral flows 
in Theorem~\ref{thm:iso.biHamil.hiearchy}.
The proof is the same as that of Theorem~\ref{thm:main.hierarchy.spinmodel}. 

\begin{thm}\label{thm:main.hierarchy.curve}
The hierarchy of geometric non-stretching curve flows \eqref{hierarchy.curveflows.perp} in $\aff(G,H)$
arising from the isospectral flow hierarchy \eqref{aff.hierarchy}
has a tri-Hamiltonian structure. 
For $k\geq 2$, this structure is given by 
\begin{equation}\label{hierarchy.curve.biHam}
(\curvflow_t)_\pe 
=\Jop_\curvflow(\delta H^{(k)}/\delta \curvflow) 
=\Hop_\curvflow(\delta H^{(k-1)}/\delta \curvflow) 
=\Eop_\curvflow(\delta H^{(k-2)}/\delta \curvflow) 
= \Pvec^{(k)}_\pe,
\quad
k=2,3,\ldots , 
\end{equation}
where $\Eop_\curvflow=\Rop_\curvflow Hop_\curvflow$ 
is a third Hamiltonian operator compatible with $J_\curvflow$ and $\Hop_\curvflow$, 
and where the Hamiltonians are given by the functionals
\begin{equation}\label{curve.Ham.k}
H^{(k)} = \frac{1}{k}\int_C \Dxinv(\Rop_\curvflow^{k+1}(\curvflow_x)\cdot\curvflow_{xx})\; dx, 
\quad
k=1,2,\ldots 
\end{equation}
and
\begin{equation}\label{curve.Ham.k=0}
H^{(0)}  = \int_C \vec\xi(\curvflow_x)\cdot\curvflow_{xx}\;dx
\end{equation}
on the domain $C=\Rnum$ or $S^1$. 
Here $\vec\xi(\curvflow_x)$ is any vector function in $\aff(G,H)$ satisfying 
\begin{equation}\label{spinmodel.Hamvector.k=0}
\grad^\msp\wedge \vec\xi = \ad^*_\gsp(\curvflow_x) . 
\end{equation}
For $k=1$, the tri-Hamiltonian structure consists of 
\begin{equation}\label{hierarchy.curve.biHam.k=1}
(\curvflow_t)_\pe 
=\Jop_\curvflow(\delta H^{(1)}/\delta \curvflow) 
=\Hop_\curvflow(\delta H^{(0)}/\delta \curvflow) 
=-\Eop_\curvflow(\delta H_0/\delta \curvflow)
\end{equation}
with $H_0=\const$ being a trivial Hamiltonian, 
and where $\nabla^\hsp_x$ is generalized to have a non-trivial cokernel \eqref{hDxinv}
given by $(\nabla^\hsp_x)\inv(0)= \curvflow_x$. 
\end{thm}

Similarly to the case $k=1$ in Theorem~\ref{thm:main.hierarchy.spinmodel}, 
note that here 
$\Hop_\curvflow(0) =-\curvflow_{x}$
since 
$-\nabla_x^{-1}\ad_\gsp(\curvflow_x)\ad_\gsp(\curvflow_{xx})(\nabla^\hsp_x)\inv(0)
= \nabla_x^{-1}(\ad_\gsp(\curvflow_x)^2\curvflow_{xx})
= -\nabla_x^{-1}\curvflow_{xx}
= -\curvflow_{x}$. 
Hence, 
$-\Eop_\curvflow(0) = \Rop_\curvflow(\curvflow_x) = \ad_\gsp(\curvflow_x)\inv \curvflow_{xx}$
is equal to 
$\Hop_\curvflow(\ad_\gsp(\curvflow_x)\curvflow_{xxx}) 
= \nabla_x^{-1}\ad_\gsp(\curvflow_x)\nabla_x(-\curvflow_{xx})
= \ad_\gsp(\curvflow_x)\inv \curvflow_{xx}$.

Consequently, 
the tri-Hamiltonian form of the general vortex filament equation \eqref{nls.aff.curve} 
is given by 
\begin{equation}
\curvflow_t =\ad_\gsp(\curvflow_x)\inv\curvflow_{xx}
= \Jop_\curvflow(\delta H^{(1)}/\delta \curvflow)
= \Hop_\curvflow(\delta H^{(0)}/\delta \curvflow)
=-\Eop_\curvflow(0) . 
\end{equation}
The second and third Hamiltonian structures shown here 
have not previously appeared in the literature. 

The general vortex filament axial equation \eqref{SP.mkdv} 
has the tri-Hamiltonian form 
\begin{equation}
\curvflow_t = -\curvflow_{xxx} +\tfrac{3}{2}\efac\ad_\gsp(\curvflow_{xx})^2\curvflow_x
=\Eop_\curvflow(\delta H^{(0)}/\delta \curvflow) 
=\Hop_\curvflow(\delta H^{(1)}/\delta \curvflow) 
= \Jop_\curvflow(\delta H^{(2)}/\delta \curvflow) , 
\end{equation}
which holds without needing the generalized cokernel \eqref{hDxinv}. 

Each of the Hamiltonians \eqref{curve.Ham.k}-\eqref{curve.Ham.k=0}
in these curve flow equations 
can be seen to be exactly the same as the Hamiltonians \eqref{spinmodel.Ham.k}--\eqref{spinmodel.Ham.k=0} 
for the associated spin vector models. 
In particular,
the first few Hamiltonians for $k\geq 1$ are explicitly given by 
\begin{align}
H^{(1)} & = \int_C \tfrac{1}{2} |\curvflow_{xx}|^2\;dx ,
\label{curve.Ham.k=1}\\
H^{(2)} &= \int_C \tfrac{1}{2}\curvflow_{xxx}\cdot (\ad_\gsp(\curvflow_x)\curvflow_{xx})\;dx,
\label{curve.Ham.k=2}\\
H^{(3)} &= \int_C \tfrac{1}{2}|\curvflow_{xxx}|^2 -\tfrac{5}{8}|\ad_\gsp(\curvflow_{xx})^2\curvflow_x|^2 \;dx , 
\label{curve.Ham.k=3}\\
H^{(4)} &= \int_C \tfrac{1}{2}(\ad_\gsp(\curvflow_x)\curvflow_{xxx})\cdot\curvflow_{xxxx} + \tfrac{7}{8} (\ad_\gsp(\curvflow_{xx})^2\curvflow_x)\cdot(\ad_\gsp(\curvflow_{xx})\curvflow_{xxx}) \;dx . 
\label{curve.Ham.k=4}
\end{align}
A more explicit expression for $H^{(0)}$ will be derived in \secref{sec:geommap}.

\section{Multi-Hamiltonian Schr\"odinger maps in Hermitian symmetric spaces}\label{sec:geommap}

We will now show how the hierarchy of general spin vector models 
obtained in Theorem~\ref{thm:main.hierarchy.spinmodel}
for affine Hermitian symmetric spaces $\aff(G,H)$ 
gives rise to a corresponding hierarchy of geometrical evolution equations
for a map $\map(x,t)$ into the Hermitian symmetric space $M=G/H$
associated with the Hermitian symmetric Lie algebra $(\gsp,A)$
underlying $\aff(G,H)$. 
These evolution equations will be geometric versions of the isospectral flows 
in Theorems~\ref{thm:iso.flow.hierarchy} and~\ref{thm:iso.biHamil.hiearchy}. 
In particular, 
the respective geometrical versions of 
the $H$-invariant NLS and mKdV equations \eqref{FK.nls} and \eqref{FK.mkdv} 
will be the Schr\:odinger map equation into $M=G/H$ and its mKdV analog. 
We will discuss some of their properties. 
Specifically, we show that, in contrast to the geometrical non-stretching curve flows 
in Theorem~\ref{thm:main.hierarchy.curve}, 
these geometrical map equations are locally stretching 
yet have a time-independent total arclength. 
Finally, 
we derive the Hamiltonian structure of the hierarchy of geometrical map equations. 
These equations turn out to have a multi-Hamiltonian structure which has a simple form
involving only the Hermitian structure and the Riemannian connection on $M=G/H$. 
This generalizes and makes explicit some results in \Ref{TerUhl1999}
for Grassmannian Hermitian spaces.

\subsection{Geometrical relationship between Hermitian spaces and Lie algebras}

There is an explicit isometric embedding of a Hermitian symmetric space $M=G/H$ 
into the orbit of the adjoint action of $G$ through $A\in\gsp$ 
in the symmetric Lie algebra $\gsp=\msp\oplus\hsp$,
where $A$ is the imaginary-unit \eqref{centrA}--\eqref{adA} in $\hsp$. 
Since $G$ is compact, 
this embedding is given by \cite{Mas}
$\Ad(G)A=\exp(\ad(\gsp))A \simeq M$. 
There is gauge freedom in the embedding, 
because $\ad(\hsp)A=0$ implies $\Ad(G)A= \Ad(G/H)\Ad(H)A= \Ad(G/H)A$. 
Thus, the embedding can be expressed as 
\begin{equation}\label{AdG.adm.rel}
\Ad(G/H) A\simeq \exp(\ad(\msp))A \simeq M,
\end{equation}
which is closely connected to the representation \eqref{T.N.B}
for the tangent vector $\T$ of non-stretching curves $\curvflow(x)$ 
in the affine Hermitian space $\aff(G,H)$ constructed from $\gsp$.
Specifically, 
\begin{equation}\label{T.iso.M}
\{\enorminv \T \simeq \Ad(\psi_\cofra)A,\ \psi_\cofra\in G\} \simeq M
\end{equation}
with $\enorminv$ given by the norm \eqref{norm.A} of $A$. 
(Note that $\Ad(\psi_\cofra) A =A$ when $\psi_\cofra(x,t)\in H\subset G$.)

Hence, for any given Hermitian symmetric space $M$, 
the components of the tangent vector $\T$ provide global coordinates in $M$,
as will be illustrated in the example in \secref{sec:example}. 
As a consequence, 
$\enorminv\T(x,t)\in\aff(G,H)=(G\rtimes\gsp)/G$ 
can be viewed as a map $\map(x,t)$ into $M=G/H$. 
Recall \cite{Hel,Sha} that the tangent space of $M$ is isomorphic $\msp$, 
whereby $T_\map M \simeq \msp$. 
An embedding of $T_\map M$ analogous to equation \eqref{T.iso.M}
comes from the principal normal vector $\T_x=\tfrac{1}{\kappa}\N$ of non-stretching curves $\curvflow(x)$, 
namely, 
\begin{equation}\label{DxT.iso.M}
\{\enorminv \T_x \simeq \Dx(\Ad(\psi_\cofra)A) = -\Ad(\psi_\cofra)J(\psi_\cofra\inv \partial_x\psi_\cofra),\ \psi_\cofra(x)\in G\} \simeq T_\map M . 
\end{equation}

A comparison between the tangent space embedding \eqref{DxT.iso.M} 
and the representation \eqref{T.N.B} of the principal bi-normal vector 
$\ad_\gsp(\T)\T_x=\tfrac{1}{\kappa}\B$ 
shows that the Hasimoto variable is given by 
$\q = \psi_\cofra\inv \partial_x\psi_\cofra$. 
The property that $\q$ belongs to $\msp\subset\gsp$ 
directly corresponds to the condition on $\psi_\cofra(x)$ that 
$\psi_\cofra\inv \partial_x\psi_\cofra$ belongs to $\msp\subset\gsp$. 
This condition can always be achieved, 
starting from a general $\psi_\cofra(x)\in G$, 
by making a gauge transformation
\begin{equation}\label{AdG.gauge}
\psi_\cofra(x) \rightarrow \psi_\cofra(x) \g(x),
\quad
\g(x)\in H\subset G, 
\end{equation}
where $\g(x)$ is determined by 
\begin{equation}\label{AdG.gauge.cond}
(\psi_\cofra\inv \partial_x\psi_\cofra)_\hsp 
\rightarrow ((\psi_\cofra\g(x))\inv \partial_x(\psi_\cofra\g(x)))_\hsp
= (\g(x)\inv(\psi_\cofra\inv \partial_x\psi_\cofra)\g(x) + \g(x)\partial_x\g(x))_\hsp =0 . 
\end{equation} 
Such a gauge transformation is the analog of going from a general linear coframe 
in $\aff(G,H)$ to a $G$-parallel frame 
as described in Theorem~\ref{thm:Gparallel.frame}. 
Note that $\T$, $\N$, $\B$ are gauge invariant,
since $\Ad(\psi_\cofra)A \rightarrow \Ad(\psi_\cofra)\Ad(\g(x))A = \Ad(\psi_\cofra)A$. 

The Hermitian and Riemannian structure of $M=G/H$ 
given in terms of the embedding \eqref{AdG.adm.rel} will be discussed next. 

The Hermitian structure on $T_\map M$ is related to the Hamiltonian operator $J_\T=\ad_\gsp(\T)$. 
Specifically, 
if $Y\simeq \Ad(\psi_\cofra)y \in T_\map M$, with $y\in\msp$, 
then 
\begin{equation}
J_\T Y \simeq \enorm\Ad(\psi_\cofra)Jy \in T_\map M , 
\end{equation}
whereby $\enorminv J_\T$ coincides with action of $J=\ad(A)$ on $\msp$. 
Hereafter we will denote 
\begin{equation}\label{J.map}
J_\map := \enorminv J_\T, 
\end{equation}
which acts as 
\begin{equation}
J_\map^2 = -\id 
\end{equation}
on $T_\map M$. 

The Riemannian metric, $g_\msp$, on $M$ is given by 
restriction of the Euclidean metric $g$ on $\aff(G,H)\simeq\gsp$ 
to the subspace corresponding to $\msp\subset\gsp$. 
Namely, for $Y\simeq \Ad(\psi_\cofra)y,Z\simeq \Ad(\psi_\cofra)z\in T_\map M$, 
with $y,z\in \msp$, 
\begin{equation}
g_\msp(Y,Z) = \brack{\Ad(\psi_\cofra)y,\Ad(\psi_\cofra)z}_\msp = \brack{y,z}_\msp
\end{equation}
by $\Ad$-invariance of the inner product. 
This metric determines a unique torsion-free covariant derivative, $\covder$, on $M$,
which can be identified with the restriction of $\grad$ 
in $\aff(G,H)\simeq\gsp$ to the subspace corresponding to $\msp\subset\gsp$:
\begin{equation}\label{solder.grad.m}
\covder \simeq -J_\map^2\grad . 
\end{equation}
The curvature of $M$ can be found by \cite{Hel,Sha}
\begin{equation}
R(X,Y) Z:= \comm{\covder_X,\covder_Y}Z =-\ad_\gsp([X,Y]_\gsp)Z 
\end{equation}
when $X,Y,Z$ are vector fields in $T_\map M$ such that $X,Y$ are commuting, 
where $R(\cdot\,,\cdot\,)$ is the Riemann curvature tensor. 
It is worth giving an explicit derivation of the curvature formula. 

Consider, from the representation \eqref{aff.grad.vector}, 
\begin{equation}
\begin{aligned}
\covder_Y Z & = -J_\map^2 D_Y (\Ad(\psi_\cofra)z) \\
& = -\Ad(\psi_\cofra)J^2(\partial_Y z + [\psi_\cofra\inv\partial_Y\psi_\cofra,z]) \\
& = \Ad(\psi_\cofra)(\partial_Y z+ [(\psi_\cofra\inv\partial_Y\psi_\cofra)_\hsp,z])
\end{aligned}
\end{equation}
by the Lie bracket structure \eqref{symmsp.brackets} of $M$. 
Thus, 
\begin{equation}
\begin{aligned}
\covder_X(\covder_Y)Z & = -J_\map^2 D_X(- J_\map^2D_Y (\Ad(\psi_\cofra)z)) \\
& = \Ad(\psi_\cofra)\big( 
\partial_X\partial_Y z + [(\psi_\cofra\inv\partial_Y\psi_\cofra)_\hsp,\partial_X z]
+ [(\psi_\cofra\inv\partial_X\psi_\cofra)_\hsp,\partial_Y z]
\\&\qquad
+[(\psi_\cofra\inv\partial_X\partial_Y\psi_\cofra)_\hsp,z]
-[(\psi_\cofra\inv\partial_X\psi_\cofra \psi_\cofra\inv\partial_Y\psi_\cofra)_\hsp,z]
\\&\qquad
+[(\psi_\cofra\inv\partial_X\psi_\cofra)_\hsp,[(\psi_\cofra\inv\partial_Y\psi_\cofra)_\hsp,z]] 
\big), 
\end{aligned}
\end{equation}
and hence 
\begin{equation}
\comm{\covder_X,\covder_Y}Z 
= \Ad(\psi_\cofra)( 
[[(\psi_\cofra\inv\partial_X\psi_\cofra)_\hsp,(\psi_\cofra\inv\partial_Y\psi_\cofra)_\hsp],z]
-[[\psi_\cofra\inv\partial_X\psi_\cofra, \psi_\cofra\inv\partial_Y\psi_\cofra]_\hsp,z]
)
\end{equation}
after use of the condition $[\partial_X,\partial_Y]=0$ that $X,Y$ are commuting,
and the Jacobi identity. 
Then, use of the decomposition 
$\psi_\cofra\inv\partial\psi_\cofra = (\psi_\cofra\inv\partial\psi_\cofra)_\hsp + (\psi_\cofra\inv\partial\psi_\cofra)_\msp$
combined with the Lie bracket structure \eqref{symmsp.brackets} of $M$
yields 
\begin{equation}
\comm{\covder_X,\covder_Y}Z 
= -\Ad(\psi_\cofra)( 
[[(\psi_\cofra\inv\partial_X\psi_\cofra)_\msp,(\psi_\cofra\inv\partial_Y\psi_\cofra)_\msp],z] 
) . 
\end{equation}
Next, 
\begin{equation}
\begin{aligned}\ 
[X,Y]_\gsp 
& = [-\Ad(\psi_\cofra)J(\psi_\cofra\inv\partial_X\psi_\cofra), -\Ad(\psi_\cofra)J(\psi_\cofra\inv\partial_Y\psi_\cofra)]
\\
& = \Ad(\psi_\cofra)[J(\psi_\cofra\inv\partial_X\psi_\cofra), J(\psi_\cofra\inv\partial_Y\psi_\cofra)]
\\
& = \Ad(\psi_\cofra)[(\psi_\cofra\inv\partial_X\psi_\cofra)_\msp, (\psi_\cofra\inv\partial_Y\psi_\cofra)_\msp]
\end{aligned}
\end{equation}
by $\Ad$-invariance of the Lie bracket and the property $J\hsp=0$. 
Thus, 
\begin{equation}
\begin{aligned}
\ad_\gsp([X,Y]_\gsp)Z
& =\ad(\Ad(\psi_\cofra)[(\psi_\cofra\inv\partial_X\psi_\cofra)_\msp, (\psi_\cofra\inv\partial_Y\psi_\cofra)_\msp])z
\\
& = \Ad(\psi_\cofra)\ad([(\psi_\cofra\inv\partial_X\psi_\cofra)_\msp, (\psi_\cofra\inv\partial_Y\psi_\cofra)_\msp])z , 
\end{aligned}
\end{equation}
which completes the derivation.

\subsection{Geometric Schr\"odinger map equation}

Through the relations
\begin{equation}\label{map.spinvector.rels}
\map \simeq \enorminv \T, 
\quad
\map_t\simeq \enorminv \T_t , 
\quad
\map_x\simeq \enorminv \T_x ,
\end{equation}
and 
\begin{equation}\label{map.J.rel}
J_\map \covder_x \simeq -\enorminv J_\T\nabla_x , 
\end{equation}
the general Heisenberg spin model \eqref{nls.aff.spinmodel} 
for $\T(x,t)$ in $\aff(G,H)$
can be expressed as a geometrical map equation for $\map(x,t)$. 
This yields
\begin{equation}\label{schr.map.eqn}
\map_t = -J_\map \covder_x \map_x, 
\end{equation}
which is the well-known Schr\"odinger map equation 
formulated on a general Hermitian symmetric space $M=G/H$. 

The Schr\"odinger map equation \eqref{schr.map.eqn} can be viewed geometrically
as a curve flow on $M=G/H$,
where the local arclength is given by 
$d\ell = \sqrt{g_\msp(\map_x,\map_x)}\; dx$. 
The flow is locally stretching:
\begin{equation}
\partial_t\,d\ell =  -(\partial_x g_\msp(\covder_x\map_x,J_\map\map_x)/\sqrt{g_\msp(\map_x,\map_x)})\; d\ell 
\neq 0 . 
\end{equation}
This result follows from 
\begin{equation}
D_t g_\msp(\map_x,\map_x) 
= 2g_\msp(\covder_t\map_x,\map_x) 
= 2g_\msp(\covder_x(J_\map\covder\map_x),\map_x) 
= -2\Dx g_\msp(\covder_x\map_x,J_\map\map_x) \neq 0
\end{equation}
using integration by parts followed by $g_\msp(J_\map X,X)=0$. 
Nevertheless, the total arclength of $\map$ on a domain $C$ in $M$ is preserved,
since 
\begin{equation}\label{map.total.arclength}
\ell = \int_C \sqrt{g_\msp(\map_x,\map_x)}\; dx
\end{equation}
satisfies 
\begin{equation}
\frac{d\ell}{dt} = 0
\end{equation}
for $C=S^1$ with periodic boundary conditions on $\map$, 
or $C=\Rnum$ with asymptotic decay conditions on $\map_x$. 

There is an explicit transformation 
under which the Schr\"odinger map equation \eqref{schr.map.eqn} 
is equivalent to the $H$-invariant NLS equation \eqref{FK.nls} 
for the Hasimoto variable $\q$. 
Moreover, from the Hamiltonian formulation of the NLS equation on $\q$ 
and the general Heisenberg spin model on $\T$, 
the Schr\"odinger map equation acquires a rich Hamiltonian structure, 
which will be derived next.

\subsection{Tri-Hamiltonian operators}

The identifications \eqref{map.spinvector.rels}
allow the spin vector Hamiltonian operators $J_\T$ and $\Hop_\T$ to be identified 
directly with corresponding operators 
$J_{\enorminv\inv\map}$ and $\Hop_{\enorminv\inv\map}$
which will be, formally, a pair of compatible Hamiltonian operators with respect to $\map$. 
However, the form of $\Hop_\T$ involves $\ad_\gsp$ and $\nabla^\hsp_x$ 
which are defined in terms of the structure of the affine space $\aff(G,H)$
and are not part of the intrinsic structure of the Hermitian symmetric space $M=G/H$. 
This obstacle can be circumvented by introducing geometrical versions of 
$\ad_\gsp$ and $\nabla^\hsp_x$ 
based on the structure of the group manifolds $G$ and $H$. 

One way to proceed is by exploiting the correspondence between 
the general Heisenberg spin model \eqref{nls.spinmodel.Hamil} 
and the NLS isospectral flow \eqref{Gparallel.nls} 
on the Hasimoto variable $\q$
arising through the soldering relations \eqref{T.N.B} and \eqref{solder.grad.m}. 
These relations show that 
\begin{equation}\label{struct.rels}
J_\map = \ad_G(\map),
\quad
\covder_x = -J_\map^2\Dx,
\quad
[\covder_x,J_\map] = 0, 
\end{equation}
hold on $T_xM$, 
and that 
\begin{equation}\label{map.rels}
\map = \Ad(\psi_\cofra) A,
\quad
\map_x = -\Ad(\psi_\cofra)J\q, 
\quad
\map_t = -\Ad(\psi_\cofra)Jw,
\end{equation}
with 
\begin{equation}\label{q.w.mapping}
\q = \psi_\cofra\inv \partial_x\psi_\cofra \in\msp, 
\quad
w = \psi_\cofra\inv \partial_t\psi_\cofra \in\gsp=\msp\oplus\hsp . 
\end{equation}

As a consequence, 
the Hamiltonian operators $J$ and $\Hop$ and the recursion operator $\Rop=\Hop J\inv$
with respect to $\q$ can be transformed into corresponding geometrical operators. 
Let $\Gcovder$ and $\Hcovder$ denote the unique torsion-free covariant derivative 
determined by the metric on $G$ and $H$, respectively. 
Then $\covder$ and $\Hcovder$ are the restrictions of $\Gcovder$ to $T_xM$ and $T_xH$:
\begin{equation}
\covder = -J_\map^2\Gcovder,
\quad
\Hcovder = (1+J_\map^2)\Gcovder . 
\end{equation}
Let $\ad_G(\cdot\,)$ denote the adjoint action of $G$ given by the Lie bracket in $T_xG$,
which can be identified with $\ad_\gsp(\cdot\,)$ introduced previously. 

\begin{prop}\label{prop:Ad.mapping}
Under the mapping $\Ad(\psi_\cofra)$:
\begin{align}
\Ad(\psi_\cofra)J\Ad(\psi_\cofra^{-1}) &
= J_\map, 
\label{map.J}\\
\Ad(\psi_\cofra)\Hop\Ad(\psi_\cofra^{-1}) & 
= -J_\map(\covder_x - \Mop_\map)J_\map 
:=\Hop_\map, 
\label{map.Hop}\\
\Ad(\psi_\cofra)\Rop\Ad(\psi_\cofra^{-1}) & 
= -J_\map(\covder_x - \Mop_\map)
:=\Rop_\map, 
\label{map.Rop}
\end{align}
where
\begin{equation}
\Mop_\map := \ad_G(\map_x)(\Hcovder_x)\inv\ad_G(\map_x)
\end{equation}
is a geometrical operator acting in $T_xM$. 
\end{prop}

\begin{proof}
Let $Y=\Ad(\psi_\cofra)y$ be any vector in $T_xM$. 
We start with $J$ and observe 
\begin{equation*}
\Ad(\psi_\cofra)J\Ad(\psi_\cofra^{-1})Y 
= \Ad(\psi_\cofra)[A,y] = [\Ad(\psi_\cofra)A,\Ad(\psi_\cofra)y] = \ad(\map)Y = J_\map, 
\end{equation*}
which yields the mapping \eqref{map.J}. 

We next consider 
\begin{equation*}
\Ad(\psi_\cofra)\Hop(\Ad(\psi_\cofra^{-1})Y)
= \Ad(\psi_\cofra)(\Dx y -\ad(q)\Dxinv(\ad(q)y)) . 
\end{equation*}
For the first term,
we note that  
$\Dx y = -J^2(D_x y + ad(q)y)$ 
by use of $-J^2\msp=\msp$ and $J\hsp=0$
along with the symmetric Lie-algebra structure \eqref{symmsp.brackets} of $\gsp$. 
Hence, we have 
\begin{equation}\label{Hop.term1}
\Ad(\psi_\cofra)\Dx y = -J_\map^2D_x(\Ad(\psi_\cofra)y) = \covder Y . 
\end{equation}
The second term is 
\begin{equation*}
-\Ad(\psi_\cofra)\ad(q)\Dxinv(\ad(q)y) = \Ad(\psi_\cofra)J\ad(Jq)\Dxinv(\ad(Jq)Jy)
\end{equation*}
after we again use the previous properties. 
Then
\begin{equation*}
\begin{aligned}
\Hcovder_x(\Ad(\psi_\cofra)\Dxinv( \ad(Jq)Jy ))
& =\Ad(\psi_\cofra)((\Dx+\ad(q))\Dxinv( \ad(Jq)Jy ))_\hsp \\
& =\Ad(\psi_\cofra)\ad(Jq)Jy
= \ad_G(\map_x)J_\map Y
\end{aligned}
\end{equation*}
since $(\ad(q)\Dxinv( \ad(Jq)Jy ))_\hsp=0$ 
by the symmetric Lie-algebra structure \eqref{symmsp.brackets}. 
Hence we have 
\begin{equation}\label{Hop.term2}
\Ad(\psi_\cofra)J\ad(Jq)\Dxinv(\ad(Jq)Jy)
= J_\map\ad_G(\map_x)(\Hcovder_x)\inv(\ad_G(\map_x)J_\map Y)
= J_\map\Mop_\map(J_\map Y) . 
\end{equation}
Combining the two terms \eqref{Hop.term2} and \eqref{Hop.term1},
we obtain the mapping \eqref{map.Hop}. 

Finally, composition of $\Hop_\map$ and $J_\map^{-1}= -J_\map$ yields
the mapping \eqref{map.Rop}. 
\end{proof}

The geometrical operators \eqref{map.J}--\eqref{map.Rop} 
have the following main properties. 

\begin{thm}\label{thm:Hamilops.map}
In any Hermitian symmetry space $M=G/H$, 
\begin{equation}\label{Rop.map}
\Rop_\map = \Hop_\map J_\map^{-1}
\end{equation}
is a hereditary recursion operator,
and 
\begin{equation}\label{J.Hop.Eop.map}
J_\map, 
\quad
\Hop_\map = \Rop_\map J_\map, 
\quad
\Eop_\map = \Rop_\map\Hop_\map = \Rop_\map^2 J_\map
\end{equation}
are compatible Hamiltonian operators 
with respect to the geometrical map variable $\map\in M$. 
Each of these operators involves only the intrinsic structure of $M$. 
They respectively correspond to the recursion operator $\Rop$ 
and the Hamiltonian operators
$\Eop=\Rop\Hop=\Rop^2J$, $\Rop\Eop=\Rop^2\Hop=\Rop^3J$, $\Rop^2\Eop=\Rop^3\Hop=\Rop^4J$, 
with respect to the Hasimoto variable $\q\in\msp$. 
\end{thm}

Theorem~\ref{thm:Hamilops.map} gives a generalization, in an explicit form, 
of some results stated in \Ref{TerUhl1999} 
for geometrical maps into Grassmannian Hermitian spaces. 
In particular, 
there is a shift of $-2$ in transforming Hamiltonian structures with respect to $\q$ 
into Hamiltonian structures with respect to $\map$. 
To proceed with the proof, 
a useful variational result will be established first. 

\begin{lem}\label{lem:var.map.q.rels}
(i) 
For any Hamiltonian functional $H$,
\begin{equation}\label{varder.map.q}
J_\map(\delta H/\delta\map) = - \Ad(\psi_\cofra)\Hop(\delta H/\delta q)
\end{equation}
is a variational identity,
where $\Hop$ is the Hamiltonian operator \eqref{Hop}. 
(ii) 
If $\Dop_\q$ is a Hamiltonian operator with respect to $\q$,
then 
\begin{equation}\label{Dop.mapping.map.q}
\Dop_\map = J_\map \Ad(\psi_\cofra)\Hop\inv\Dop_\q\Hop\inv\Ad(\psi_\cofra^{-1}) J_\map
\end{equation}
is a Hamiltonian operator with respect to $\map$.
\end{lem}

\begin{proof}
To derive part (i), 
consider an arbitrary variation $\delta\psi_\cofra$. 
The first expression \eqref{q.w.mapping} yields 
$\delta\q 
= (\Dx+\ad(\q))(\psi_\cofra\inv\delta\psi_\cofra)$, 
which can be projected into $\msp$ and $\hsp$. 
Since $\delta\q$ belongs to $\msp$, 
the projections give 
$\delta\q = \Dx(\psi_\cofra\inv\delta\psi_\cofra)_\msp + \ad(\q)(\psi_\cofra\inv\delta\psi_\cofra)_\hsp$
and 
$\Dx(\psi_\cofra\inv\delta\psi_\cofra)_\hsp + \ad(\q)(\psi_\cofra\inv\delta\psi_\cofra)_\msp=0$
by use of the symmetric Lie-algebra structure \eqref{symmsp.brackets} of $\gsp$. 
Solving the second equation for $(\psi_\cofra\inv\delta\psi_\cofra)_\hsp$ 
and substituting it into the first equation, 
we obtain 
\begin{equation}\label{var.q}
\delta\q = \Hop(\psi_\cofra\inv \delta\psi_\cofra)_\msp . 
\end{equation}
Next, the first expression \eqref{map.rels} yields
$\delta\map = \delta(\psi_\cofra A\psi_\cofra\inv) = -\Ad(\psi_\cofra)J(\psi_\cofra\inv\delta \psi_\cofra)$. 
Hence, we have 
\begin{equation}\label{var.map}
\delta\map = -J_\map(\delta\psi_\cofra\psi_\cofra\inv)_\msp . 
\end{equation}
The variational relations \eqref{var.map} and \eqref{var.q} then yield
\begin{equation}\label{var.q.map.rel}
\delta\q = \Hop(\Ad(\psi_\cofra^{-1})J_\map\delta\map) . 
\end{equation}
Now consider the corresponding variations of any Hamiltonian functional $H$:
modulo a total $x$-derivative, 
\begin{equation}\label{var.H.q}
\delta H = \brack{\delta H/\delta\q,\delta\q}_\msp 
= \brack{\delta H/\delta\q,\Hop(\Ad(\psi_\cofra^{-1})J_\map\delta\map)}_\msp 
= -\brack{\Hop(\delta H/\delta\q),\Ad(\psi_\cofra^{-1})J_\map\delta\map}_\msp 
\end{equation}
by use of the property that $\Hop$ is skew-adjoint. 
Likewise, 
\begin{equation}\label{var.H.map}
\delta H = g_\msp(\delta H/\delta\map,\delta\map) 
= \brack{\Ad(\psi_\cofra^{-1})J_\map(\delta H/\delta\map),\Ad(\psi_\cofra^{-1})J_\map\delta\map}_\msp 
\end{equation}
by use of the $\Ad$-invariance of $g_\msp$. 
Equating these two expressions \eqref{var.H.q} and \eqref{var.H.map}, 
we get the variational derivative relation \eqref{varder.map.q}.

To verify part (ii), 
suppose that $\Dop_\q$ is Hamiltonian operator with respect to $\q$, 
and consider $\q_t = \Dop_\q(\delta H/\delta \q)$.
We can relate $\q_t$ to $\map_t$ by applying equation \eqref{var.q.map.rel}
to the variation given by $\delta=\partial_t$:
\begin{equation}\label{Dtq.Dtmap.rels}
\q_t = \Hop(\Ad(\psi_\cofra^{-1})J_\map\map_t) . 
\end{equation}
This yields
$\map_t =-J_\map\Ad(\psi_\cofra)\Hop\inv(\q_t)$
and hence 
\begin{equation}
\Dop_\map(\delta H/\delta\map) =-J_\map \Ad(\psi_\cofra)\Hop\inv\Dop_\q(\delta H/\delta \q) . 
\end{equation}
Then, through the relation \eqref{varder.map.q}, 
$\Dop_\map$ is given by the operator \eqref{Dop.mapping.map.q}. 
The proof that this operator is Hamiltonian
can be given by the same argument used in the proof of Lemma~\ref{lem:var.spinvector.curve.rels}. 
\end{proof}

Now the main step in the proof of Theorem~\ref{thm:Hamilops.map} 
consists of applying Lemma~\ref {lem:var.map.q.rels}
to the three compatible Hamiltonian operators 
$\Dop_\q=\Eop=\Hop J\inv \Hop$, 
$\Dop_\q=\Rop\Eop= (\Hop J\inv)^2\Hop$, 
$\Dop_\q=\Rop^2\Eop = (\Hop J\inv)^3\Hop$. 
This yields, respectively, 
$\Dop_\map = J_\map \Ad(\psi_\cofra) J\inv \Ad(\psi_\cofra^{-1}) J_\map 
= J_\map$, 
$\Dop_\map = J_\map \Ad(\psi_\cofra) J\inv\Hop J\inv \Ad(\psi_\cofra^{-1}) J_\map 
= \Hop_\map$, 
$\Dop_\map = J_\map \Ad(\psi_\cofra) J\inv\Hop J\inv\Hop J\inv \Ad(\psi_\cofra^{-1}) J_\map 
= \Eop_\map$. 
Consequently, these geometrical operators are 
compatible Hamiltonian operators with respect to $\map$. 
Moreover, their compatibility implies that $\Rop_\map$ is a hereditary recursion operator. 
This completes the proof.

\subsection{Multi-Hamiltonian form of the general Schr\"odinger map equation}

The Hamiltonian operators \eqref{J.Hop.Eop.map} in Theorem~\ref{thm:Hamilops.map}
are geometrical counterparts of the spin vector operators $J_\T$, $\Hop_\T$, and $\Eop_T$. 
As a consequence, 
a tri-Hamiltonian form can be derived for the Schr\"odinger map equation \eqref{schr.map.eqn}
similarly to the tri-Hamiltonian structure of the general Heisenberg model \eqref{nls.spinmodel.Hamil}. 

This geometrical tri-Hamiltonian form is given by 
\begin{equation}\label{schr.map.triHamil}
\map_t = -J_\map \covder_x \map_x
= J_\map(\delta H^{(1)}/\delta \map) 
= \Hop_\map(\delta H^{(0)}/\delta \map) 
= \Eop_\map(0) , 
\end{equation}
where
\begin{equation}\label{map.Ham.n=1}
H^{(1)} = \int_C \tfrac{1}{2} g_\msp(\map_x,\map_x)\;dx 
\end{equation}
and
\begin{equation}\label{map.Ham.n=0}
H^{(0)}  = \int_C g_\msp(\xi(\map),\map_x)\;dx
\end{equation}
are the Hamiltonian functionals
on the domain $C=\Rnum$ or $S^1$,
with $\xi(\map)$ being any vector function satisfying 
\begin{equation}\label{map.Hamvector.n=0}
\covder\wedge \xi = J^*_\map . 
\end{equation}

Here $J^*_\map$ denotes the skew tensor associated with $J_\map$
as defined by 
$g(J^*_\map,X\wedge Y) =g(X,J_\map Y)$,
for all vectors $X,Y\in T_xM$,
where $M=G/H$. 
Note that equation \eqref{map.Hamvector.n=0} has a solution 
because $\covder J^*_\map=0$ and so no integrability condition arises on $\xi$. 
(Specifically, $\covder\wedge\covder\wedge \xi = 0$ holds identically.)

Also, note 
$\Eop_\map(0)=\Rop_\map\Hop_\map(0)$
with $\Hop(0) = J_\map\ad_G(\map_x)\map = -J_\map\ad_G(\map)\map_x = -J_\map^2\map_x = \map_x$,
which holds by generalizing $\Hcovder_x$ to have a non-trivial cokernel 
\begin{equation}\label{mDxinv}
(\Hcovder_x)\inv(0)= \map . 
\end{equation}
This generalization is well-defined 
because $\Hcovder_x\map = \Ad(\psi_\cofra)(\ad(q)A)_\hsp =0$. 
Then $\Eop_\map(0)=\Rop_\map(\map_x) = -J_\map \covder_x \map_x$. 

The second and third Hamiltonian structures for the Schr\"odinger map equation 
in a general Hermitian symmetric space $M=G/H$ 
are new results. 

Finally, it is worth looking at the geometrical equivalence between 
the Schr\"odinger map equation \eqref{schr.map.eqn} 
and the isospectral NLS equation \eqref{Gparallel.nls}.
This equivalence arises directly from the relation \eqref{Dtq.Dtmap.rels}
derived previously between $\q_t$ and $\map_t$,
together with the main relations
\begin{equation}\label{main.rels}
\map_x = -\Ad(\psi_\cofra)J\q,
\quad
J_\map = \Ad(\psi_\cofra)J\Ad(\psi_\cofra^{-1}),
\quad
\covder_x = D_x\Ad(\psi_\cofra^{-1}) = -\Ad(\psi_\cofra)J^2D_x, 
\end{equation}
which follow from equations \eqref{struct.rels}, \eqref{map.rels}, \eqref{Hop.term1}. 

\begin{prop}\label{prop:schr.map.nls.q}
The Schr\"odinger map equation \eqref{schr.map.eqn} on $\map\in M$
in any Hermitian symmetric space $M=G/H$
is equivalent to the NLS isospectral flow \eqref{Gparallel.nls}
on the Hasimoto variable $\q\in\msp$ 
in the associated affine Hermitian space $\aff(G,H)$. 
Both this isospectral flow and the Schr\"odinger map equation 
geometrically correspond to the general Heisenberg spin vector model \eqref{nls.spinmodel.Hamil}
for $\T\in \aff(G,H)$ 
\end{prop}

Through this equivalence, 
the tri-Hamiltonian structure \eqref{nls.tri.ham} of the isospectral NLS flow \eqref{Gparallel.nls}
turns out by Theorem~\ref{thm:Hamilops.map} 
to yield the first Hamiltonian structure \eqref{schr.map.triHamil} of the Schr\"odinger map equation
plus two additional Hamiltonian structures:
\begin{equation}\label{schr.map.othHamil}
\map_t = -J_\map \covder_x \map_x
= \Rop_\map^{-1} J_\map(\delta H^{(2)}/\delta \map) 
= \Rop_\map^{-2} J_\map(\delta H^{(3)}/\delta \map) , 
\end{equation}
where
\begin{align}
H^{(2)} & = \int_C \tfrac{1}{2} g_\msp(\covder_x\map_x, J_\map\map_x)\;dx,
\label{map.Ham.n=2}\\
H^{(3)} & = \int_C \tfrac{1}{2} g_\msp(\covder_x\map_x, J_\map\map_x)\;dx
\label{mapl.Ham.n=3}
\end{align}
are Hamiltonian functionals
on the domain $C=\Rnum$ or $S^1$. 
Clearly, there is a hierarchy of similar Hamiltonian structures given by 
\begin{equation}\label{schr.map.multiHam}
\map_t = \Rop_\map^{-k} J_\map(\delta H^{(k+1)}/\delta \map) 
\end{equation}
with the Hamiltonian functionals 
\begin{equation}\label{map.Ham.k}
H^{(k)}  = \frac{1}{k}\int_C \Dxinv(g_\msp(\Rop_\map^k(\map_x), J_\map\map_x))\; dx, 
\quad
k=1,2,\ldots 
\end{equation}
Each of these Hamiltonians corresponds to the Hamiltonians \eqref{aff.Ham.n}
in the hierarchy of isospectral flows \eqref{aff.hierarchy.triHam} 
with $\q$ and its $x$-derivatives 
expressed in terms of $\map$ and its covariant $x$-derivatives
through the relations \eqref{main.rels}. 

The bottom of the hierarchy of Hamiltonian structures can be viewed as being given by 
the third Hamiltonian form \eqref{schr.map.triHamil} of the Schr\"odinger map equation
involving $\Eop_\map$ and $H=0$. 
It is interesting to note that this non-standard Hamiltonian form corresponds to 
a similar non-standard Hamiltonian structure for the isospectral NLS flow \eqref{Gparallel.nls}:
\begin{equation}\label{nls.alt.ham}
\tfrac{1}{\efac}\q_t = -J\q_{xx} +\tfrac{1}{2}[\q,[\q,J\q]] =\Rop(q_x) = \Rop\Eop (0), 
\end{equation}
where $\Eop(0)= \q_x$ as explained in Proposition~\ref{prop:alt.Ham.struc}.

\subsection{Hierarchy of geometrical map equations}

Each isospectral flow equation \eqref{aff.hierarchy.biHam}  
in the hierarchy in Theorem~\ref{thm:iso.biHamil.hiearchy} 
is equivalent to a geometrical map equation for $\map(x,t)$
given by the same steps used to establish the equivalence between 
the NLS isospectral flow \eqref{Gparallel.nls}
and the Schr\"odinger map equation \eqref{schr.map.eqn}
in Proposition~\ref{prop:schr.map.nls.q}. 
Likewise, this equivalence allows identifying 
each spin vector model \eqref{hierarchy.spinmodel.biHam.k>1}, \eqref{hierarchy.spinmodel.biHam.k=1} 
for $\T(x,t)$ in the hierarchy in Theorem~\ref{thm:main.hierarchy.spinmodel}
with a corresponding geometrical map equation. 
All of the resulting geometrical map equations 
will have an explicit multi-Hamiltonian structure
similar to what was just derived for the Schr\"odinger map equation \eqref{schr.map.multiHam}. 

These results are summarized in the following two Theorems. 

\begin{thm}\label{thm:main.hierarchy.map}
For any Hermitian symmetric space $M=G/H$,
there is a hierarchy of geometrical map equations 
\begin{equation}\label{hierarchy.maps}
\map_t = \Rop_\map^{n}(\map_x) , 
\quad
n=0,1,2,\ldots
\end{equation}
for $\map(x,t)\in M$,
where $\Rop_\map$ is the hereditary recursion operator \eqref{Rop.map}.
Each of these geometrical map equations \eqref{hierarchy.maps}
for $n\geq 1$ has a multi-Hamiltonian structure given by 
\begin{equation}\label{hierarchy.map.triHam}
\map_t = \Rop_\map^{-k} J_\map(\delta H^{(k+n)}/\delta \map),
\quad
k=1,2,\ldots, 
\end{equation}
where the Hamiltonians are given by the functionals \eqref{map.Ham.k}
on the domain $C=\Rnum$ or $S^1$.
\end{thm}

The hierarchy of geometrical map equations \eqref{hierarchy.maps} 
starts from $\map_t =\map_x$ ($n=0$), 
which is a travelling wave equation. 
Next in this hierarchy ($n=1$) is the Schr\"odinger map equation \eqref{schr.map.eqn}. 
After that comes ($n=2$) 
\begin{equation}\label{mkdv.map.eqn}
\map_t = -(\nabla_x^\msp)^2 \map_x - \tfrac{1}{2}\ad_G(J_\map\map_x)^2\map_x, 
\end{equation}
which is an mKdV analog of the Schr\"odinger map equation. 
More precisely, it is analogous to the mKdV equation in potential form
and corresponds to the higher-order spin vector model \eqref{mkdv.aff.spinmodel}. 

Like the Schr\"odinger map equation \eqref{schr.map.eqn},
each of these geometrical map equations describes a locally stretching curve flow on $M$
with the feature that the total arclength of the curve $\map$ 
on a domain $C$ in $M$ is a constant of motion. 

\begin{thm}\label{thm:main.hierarchy.relations}
In the hierarchy of geometrical map equations \eqref{hierarchy.maps}
in any Hermitian symmetric space $M=G/H$, 
the $+n$ evolution equation \eqref{hierarchy.maps} for the map $\map$ 
is equivalent both to the $+n+1$ flow equation \eqref{aff.hierarchy} 
in the hierarchy of isospectral flows for the Hasimoto variable $\q$
and to the $+k$ vector equation \eqref{hierarchy.spinmodels} 
in the hierarchy of spin models for $\T$ 
in the associated affine Hermitian space $\aff(G,H)=(G\rtimes\gsp)/G$. 
\end{thm}

In particular, 
the $+1$ flow equation $\q_t=\Rop(J\q)= \q_x$ 
corresponds to the $n=0$ map equation $\map_t =\map_x$; 
the $+2$ flow equation $\q_t=\Rop^2(J\q) = -J\q_{xx} +\tfrac{1}{2}\ad(\q)^2J\q$,
which is the $H$-invariant NLS equation, 
corresponds to the $n=1$ map equation $\map_t =\Rop_\map(\map_x) = -J_\map \covder_x\map_x$,
which is the Schr\"odinger map equation \eqref{schr.map.eqn};
and $+3$ flow equation 
$\q_t=\Rop^3(J\q) = -\q_{xxx} +\ad(\q)^2\q_x +\tfrac{1}{2}(\ad(J\q)^2\q)_x$,
which is the $H$-invariant mKdV equation, 
corresponds to the $n=2$ map equation 
$\map_t =\Rop_\map^2(\map_x) = -(\nabla_x^\msp)^2 \map_x - \tfrac{1}{2}\ad_\msp(J_\map\map_x)^2\map_x$, 
which is the mKdV analog of the Schr\"odinger map equation.

\subsection{Matrix formulation of geometrical map equations and spin vector models}

When examples of Hermitian symmetric spaces $M=G/H$ are considered, 
it is useful to have a matrix representation for 
the geometrical map variable $\map\in M$, 
the spin vector variable $\T\in\aff(G,H)$, 
and the Hasimoto variable $\q\in\msp$. 

We will work with the fundamental matrix representation in $\mk{gl}(\dim\gsp)$ 
for the Hermitian symmetric Lie algebra $\gsp=\msp\oplus\hsp$. 
Since $G$ is a compact Lie group, 
recall that it is generated by the exponential mapping $\exp(\gsp)\simeq G$. 

A fundamental matrix expression for 
the spin vector variable $\T$ and the geometrical map variable $\map$ 
is given by 
\begin{equation}\label{map.spin.matrix}
\enorminv\T=\map=\Ad(\exp(\msp(x,t)))A
\end{equation}
because their representations \eqref{T.iso.M} and \eqref{map.rels}
are invariant under gauge transformations \eqref{AdG.gauge}. 
Derivatives of $\T$ and $\map$ can be obtained 
from relations \eqref{aff.grad.vector} and \eqref{Hop.term1} 
by applying, respectively, 
$\nabla_x = \Dx$ and $\covder_x = -\ad(\Ad(\exp(\msp(x,t)))A)^2\Dx$
to the matrix expression \eqref{map.spin.matrix}. 
In particular, 
\begin{gather}
\enorminv\T_x =\Dx\Ad(\exp(\msp(x,t)))A,
\label{spin.matrix.D}\\
\map_x = -\ad(\Ad(\exp(\msp(x,t)))A)^2\Dx\Ad(\exp(\msp(x,t)))A,
\label{map.spin.matrix.Dx}
\end{gather}
and so on. 
Similarly, their time derivative is given by 
\begin{gather}
\enorminv\T_t =D_t\Ad(\exp(\msp(x,t)))A,
\\
\map_t = -\ad(\Ad(\exp(\msp(x,t)))A)^2 D_t\Ad(\exp(\msp(x,t)))A . 
\end{gather}
Moreover, the matrix expression \eqref{map.spin.matrix} yields
\begin{equation}\label{map.spin.adT.Jmap.matrix}
\enorminv\ad_\gsp(\T)=J_\map=\ad(\Ad(\exp(\msp(x,t)))A) . 
\end{equation}
As a result, 
the Schr\"odinger map equation \eqref{schr.map.eqn}
and the Heisenberg spin model \eqref{nls.aff.spinmodel} 
are equivalent to the matrix system 
\begin{equation}\label{schrmap.spinmodel.matrix}
D_t \Ad(\exp(\msp(x,t)))A + \ad(\Ad(\exp(\msp(x,t)))A) \Dx^2 \Ad(\exp(\msp(x,t)))A = 0 . 
\end{equation}

Finally, a fundamental matrix expression for Hasimoto variable $\q$ 
can be obtained by starting with $\psi_\cofra(x)=\exp(\msp(x,t))$ 
and applying a gauge transformation \eqref{AdG.gauge} 
to achieve the gauge condition \eqref{AdG.gauge.cond}:
\begin{equation}
(\g(x,t)\inv(\exp(-\msp(x,t))\partial_x\exp(\msp(x,t)))\g(x,t) + \g(x,t)\partial_x\g(x,t))_\hsp =0 , 
\end{equation}
which determines $\g(x,t)\in H\subset G$. 
Then, from the representation \eqref{q.w.mapping}, 
\begin{equation}\label{q.matrix}
\q= \g(x,t)\inv\exp(-\msp(x,t))(\partial_x\exp(\msp(x,t)))\g(x,t) . 
\end{equation}

\section{Example of $\Cnum P^N$}\label{sec:example}

We will now apply the general theory developed in
Theorems~\ref{thm:SP.Gparallel.hierarchies}, ~\ref{thm:main.hierarchy.spinmodel}, ~\ref{thm:main.hierarchy.curve}, and~\ref{thm:main.hierarchy.map} 
to the Hermitian symmetric Lie algebra
$\mk{su}(N+1)\simeq \Cnum^{N}\oplus \mk{u}(N)$,
which underlies the complex projective space $\Cnum P^N= SU(N+1)/S(U(N)\times U(1))=SU(N+1)/U(N)$. 

In this example,
first the NLS and mKdV isospectral flows on the Hasimoto variable
will be written down,
where the cubic nonlinear terms will be expressed in terms of the inner product 
and the Hermitian structure of the Lie algebra. 
Next the corresponding geometric non-stretching curve flows 
in the affine Hermitian symmetric spaces $\aff(SU(N+1),U(N))$ and $\aff(SU(N+2),S(U(N)\times U(2)))$ 
will be formulated.
After that,
the fundamental matrix representation of
the Schr\"odinger map equation in $\Cnum P^N$
and Heisenberg spin vector model in $\mk{su}(N+1)$,
along with their Hamiltonian structures,
are written out. 

Details of the needed algebraic structure of the Lie algebra
$\mk{su}(N+1)\simeq \Cnum^{N}\oplus \mk{u}(N)$ 
are summarized in an appendix.

\subsection{Integrable systems arising from a complex vector-valued Hasimoto variable}

The Hermitian symmetric Lie algebra 
$\mk{su}(N+1)\simeq \Cnum^{N}\oplus \mk{u}(N)$
has the fundamental matrix representation 
\begin{align}
& \bpm -\tr\AA&\a\\-\abar^\t&\AA \epm 
\in \gsp = \mk{su}(N+1),
\label{CP.g}
\\
& \bpm -\tr\AA& \mb{0}\\\mb{0}&\AA \epm 
\in\hsp=\mk{s}(\mk{u}(1)\oplus\mk{u}(N))\simeq \mk{u}(N), 
\label{CP.h}
\\
& \bpm 0&\a\\-\abar^\t & \mb{0} \epm 
\in\msp=\mk{su}(N+1)/\mk{u}(N) \simeq \Cnum^N, 
\label{CP.m}
\end{align}
where $\AA\in\mk{u}(N)$, $\a\in\Cnum^N$, $\tr(\AA)\in\i\Rnum$. 
Hence, 
the Hasimoto variable in $\msp\simeq \msp_\aff$ is given by the matrix 
\begin{equation}
\q= \bpm 0&\qq(x,t)\\-\qqbar^\t(x,t) & \mb{0} \epm,
\quad
\qq\in\Cnum^N, 
\end{equation}
which represents a complex vector in $\Cnum^n$. 
The geometrical setting for this variable consists of 
arclength-parameterized curves $\curvflow(x)$ 
formulated using a linear coframe $\cofra$ and a linear connection $\conx$
in the affine Hermitian space 
\begin{equation}
\aff(SU(N+1),U(N)) = ( SU(N+1)\rtimes \mk{su}(N+1) )/SU(N+1) ,
\end{equation}
where the tangent vector of a curve is represented by the matrix 
\begin{equation}
\cofra\hook\curvflow_x \simeq \e 
= \enorm\bpm \tfrac{N}{N+1} i & 0 \\ 0& -\tfrac{1}{N+1}i I_N \epm 
\in Z(\hsp) ,
\quad
\efac=\dim(\Cnum^{N}) = 2N , 
\end{equation}
which is a unit-norm element belonging to the center of $\hsp$. 
In this setting, 
the Hasimoto variable represents the principal bi-normal vector of the curve, 
$\cofra\hook \ad_\gsp(\curvflow_x)\nabla_x\curvflow_x \simeq \q$. 
It also is the tangential projection of the linear connection $\conx$ in $\aff(SU(N+1),U(N))$
along the tangential direction $\curvflow_x$ of the curve, 
$\q\simeq \conx\hook \curvflow_x$. 

A non-stretching flow $\curvflow_t$ of an arclength-parameterized curve 
is specified by the variable 
$h = h_\pe +h_\pa \simeq \cofra\hook\curvflow_t$ in $\gsp\simeq \gsp_\aff$,
where this variable has the matrix representation 
\begin{align}
h_\pe & = \bpm 0&\hpe(x,t)\\-\hpebar^\t(x,t) & \mb{0} \epm,
\quad
\hpe\in\Cnum^N,
\\
h_\pa & = \bpm -\tr\hhpa(x,t) & \mb{0}\\\mb{0}&\hhpa(x,t) \epm ,
\quad
\hhpa\in\mk{u}(N) . 
\end{align}
The projection of the linear connection along the flow direction $\curvflow_t$ 
is given by the variable $w = w_\pe +w_\pa \simeq\conx\hook\curvflow_t$
in $\gsp\simeq \gsp_\aff$, 
which is represented by the matrices
\begin{align}
w_\pe & = \bpm 0&\wpe(x,t)\\-\wpebar^\t(x,t) & \mb{0} \epm,
\quad
\wpe\in\Cnum^N,
\\
w_\pa & = \bpm -\tr\wwpa(x,t) & \mb{0}\\\mb{0}&\wwpa(x,t) \epm ,
\quad
\wwpa\in\mk{u}(N) . 
\end{align}

To write out the integrable isospectral flows on the Hasimoto variable $\q$, 
we will need the Lie bracket structure \eqref{symmsp.brackets}
and, in particular, the $\ad$-squared action of $\Cnum^N$ on $\Cnum^N$. 
This structure can be expressed entirely in terms of 
the vector dot product on $\Cnum^N$ 
and the Hermitian structure $J=\ad(A)$ 
of $SU(N+1)/U(N)$ 
which acts by multiplication by $i$ on $\Cnum^N$,
where $A=\enorminv\e$ as shown in equations \eqref{CP.A}--\eqref{CP.J}. 

The NLS and mKdV isospectral flows \eqref{Gparallel.nls} and \eqref{Gparallel.mkdv}
thereby have the respective forms
\begin{equation}\label{CP.nls}
\tfrac{1}{\efac} \qq_t=-i(\qq_{xx} +2|\qq|^2\qq)
\end{equation}
and 
\begin{equation}\label{CP.mkdv}
\tfrac{1}{\efac} \qq_t=-\qq_{xxx}-3|\qq|^2\qq_x-3(\qq_x\cdot\qqbar)\qq . 
\end{equation}
These complex-vector equations \eqref{CP.nls} and \eqref{CP.mkdv},
which are $U(N)$-invariant, 
agree with the NLS and mKdV integrable systems derived 
from geometric curve flows in the Lie group $SU(N+1)$ in \Ref{Anc2007},
after a phase rotation, a Galilean boost, and a scaling of $t$  are applied. 
Their bi-Hamiltonian structure \eqref{nls.tri.ham} and \eqref{mkdv.tri.ham} is also 
presented explicitly there. 
In particular, the bi-Hamiltonian operators are given by 
\begin{align}
\Hop & =\Dx+2i\qq \Dxinv \im\qqbar\cdot+\qq\cdot \Dxinv  \qq\wedgebar, 
\label{Hop.CP}
\\
\Jop & =\Dx +2\qq \Dxinv \re\qqbar\cdot +\qq\cdot \Dxinv \qq\odotbar, 
\label{Jop.CP}
\end{align}
which are $U(N)$-invariant, 
where 
\begin{equation}
\a\wedgebar\b = \abar^\t\b - \bbar^\t\a,
\quad
\a\odotbar\b = \abar^\t\b + \bbar^\t\a,
\quad
\a,\b\in\Cnum^N . 
\end{equation}

An equivalent form for each isospectral flow equation \eqref{CP.nls} and \eqref{CP.mkdv}
is obtained by using the relation \eqref{CP.dotprod.rel} 
between the vector dot product and the inner product $\brack{\cdot,\cdot}$ on $\Cnum^N$. 
This yields 
\begin{equation}\label{CP.nls.alt}
\tfrac{1}{\efac} \qq_t=-i(\qq_{xx} +2\brack{\qq,\qq}\qq)
\end{equation}
and 
\begin{equation}\label{CP.mkdv.alt}
\tfrac{1}{\efac} \qq_t=-\qq_{xxx}-3\brack{\qq,\qq}^2\qq_x-3(\brack{\qq,\qq_x} +i\brack{\i\qq,\qq_x})\qq , 
\end{equation}
which are expressed just in terms of the inner product. 
The bi-Hamiltonian operators \eqref{Hop.CP}--\eqref{Jop.CP}
can be expressed in an analogous form. 

In the context of Hermitian symmetric spaces, 
the isospectral NLS equation \eqref{CP.nls} first appeared in \Refs{For1984,LanPer2000},
and likewise the isospectral mKdV equation \eqref{CP.mkdv} first appeared in \Ref{AthFor1987}, but the explicit bi-Hamiltonian operators were not given.

\subsection{Schr\"odinger map equation in $\Cnum P^N$ and Heisenberg spin vector model in $\mk{su}(N+1)$}

The isospectral NLS equation \eqref{CP.nls} is geometrically equivalent to 
the Schr\"odinger map equation \eqref{schr.map.eqn} 
for a map $\map(x,t)$ into the Hermitian symmetric space 
\begin{equation}\label{CPsp}
\Cnum P^N \simeq SU(N+1)/S(U(1)\times U(N))=SU(N+1)/U(N) . 
\end{equation}
This space has an isometric embedding into the Lie algebra 
$\mk{su}(N+1)\simeq \Cnum^{N}\oplus \mk{u}(N)$,
as given by the matrix
\begin{equation}\label{map.CP}
\map = 
\bpm 
\tfrac{N}{N+1}i -\sin^2(\theta) i & -\sin(\theta)\cos(\theta) i\Theta \\
\sin(\theta)\cos(\theta) i\bar\Theta^\t & \tfrac{-1}{N+1} i I_N + \sin^2(\theta) i \bar\Theta^\t\Theta 
\epm
\in \mk{su}(N+1)
\end{equation}
with 
\begin{equation}
\theta\in[0,\pi],
\quad
\Theta\in\Cnum^N,
\quad
\brack{\Theta,\Theta}=1 . 
\end{equation}
Thus, $(\theta,\Theta)$ provide global coordinates for the space \eqref{CPsp}.

Through the matrix representation of the principal bi-normal vector, 
$\ad_\gsp(\curvflow_x)\nabla_x\curvflow_x = \efac J_\map\covder_x\map$, 
shown in equation \eqref{CP.q}, 
the Hasimoto variable can be expressed explicitly in terms of the coordinates $(\theta,\Theta)$:
\begin{equation}
\qq = 
e^{i(1+\frac{1}{N})\Gamma}\left( 
\theta_x \Theta +\sin(\theta) \Theta_x + \sin(\theta)(1-\cos(\theta)) \brack{i\Theta_x,\Theta} i\Theta 
\right)\mb{\Gamma} , 
\end{equation}
where $\Gamma$ and $\mb{\Gamma}$ are nonlocal expressions given by 
the system of differential equations \eqref{CP.gauge}.

The Schr\"odinger map equation \eqref{schr.map.eqn}
in matrix form \eqref{schrmap.spinmodel.matrix} is given by 
$\map_t + \ad(\map)\Dx^2\map = 0$.
This $\mk{su}(N+1)$-matrix equation can be split into $\msp$ and $\hsp$ components,
given by a complex vector equation and a unitary matrix equation:
\begin{align}
&
\tfrac{1}{2}\sin(2\theta)\i\Theta_t + \cos(2\theta)\theta_t \i\Theta 
-\big(\theta_x\Theta+ \tfrac{1}{2}\sin(2\theta)( i (1-\cos(2\theta)) \brack{i\Theta_x,\Theta}\Theta +\Theta_x) \big)_x =0 , 
\\
&\begin{aligned}
& \tfrac{1}{2}(1-\cos(2\theta)) i(\bar{\Theta}^\t\Theta)_t
+\sin(2\theta)\theta_t i\bar{\Theta}^\t\Theta
\\&\qquad
-\big(\tfrac{1}{2}(1-\cos(2\theta))
( \Theta\wedgebar\Theta_x +\tfrac{1}{2}(1-\cos(2\theta))\brack{i\Theta_x,\Theta} i\bar{\Theta}^\t\Theta ) \big)_x=0 . 
\end{aligned}
\end{align}
This coupled nonlinear system is equivalent to the evolution system 
\begin{align}
&\begin{aligned}
\theta_t & =
-\sec(2\theta)\big( 
\theta_x \brack{i\Theta_x,\Theta}
+(\tfrac{1}{2}\sin(2\theta)\cos(2\theta) \brack{i\Theta_x,\Theta})_x
\big) , 
\end{aligned}
\label{CP.vec.eqn}
\\
&\begin{aligned}
i\Theta_t & =
2\csc(2\theta)\big( 
\theta_x \brack{i\Theta_x,\Theta} i\Theta
-\tfrac{1}{2}\sin(2\theta)\cos(2\theta)\brack{i\Theta_x,\Theta} i\Theta_x
\\&\qquad
+(\tfrac{1}{2}\sin(2\theta)\brack{i\Theta_x,\Theta} i\Theta)_x
+(\theta_x\Theta +\tfrac{1}{2}\sin(2\theta)\Theta_x)_x
\big) , 
\end{aligned}
\label{CP.matr.eqn}
\end{align}
which inherits the integrability properties of the Schr\"odinger map equation \eqref{schr.map.eqn}.

The integrable system \eqref{CP.vec.eqn}--\eqref{CP.matr.eqn}
is equivalent to the Heisenberg spin model \eqref{nls.spinmodel.Hamil}
for $\T=\enorm\map \in\mk{su}(N+1)$. 
Another way of formulating this system is by expressing 
\begin{equation}
\enorminv \T = \bpm -i\tr\S & \pm i\sqrt{1-|\s|^2}\s\\ \mp i\sqrt{1-|\s|^2}\sbar^\t& i\S \epm,
\quad
\S = \tfrac{-1}{N+1} I_N  +\bar\s^\t \s, 
\end{equation}
where
\begin{equation}
\s=\sin(\theta)\Theta\in \Cnum^N,
\quad
|\s|\leq 1 . 
\end{equation}
In this representation, 
$\T$ gives a double cover of the complex ball $0\leq|\s|\leq 1$
describing the stereographic projection of the two half-spaces of $\Cnum P^N$ 
onto the complex space $\Cnum^N$. 
Then the evolution equation for $\T$ can be expressed as an evolution equation for $\s$ given by combining the equations in the system \eqref{CP.vec.eqn}--\eqref{CP.matr.eqn}. 

It is worth remarking on the case $N=1$. 
The isospectral flow equations \eqref{CP.nls} and \eqref{CP.mkdv} reduce to 
the scalar NLS equation $u_t=-i(u_{xx} +2|u|^2u)$
and the scalar Hirota (complex mKdV) equation $u_t=-u_{xxx}-6|u|^2u_x$
for $u=\qq\in\Cnum$. 
In this case, $\Cnum P^1 \simeq S^2$ is the Riemann sphere. 
The coordinates $(\theta,\Theta)$ geometrically describe 
a polar angle $\theta\in[0,\pi]$ and an azimuthal angle $\varphi\in[0,2\pi)$
represented by a unit complex number $\Theta=e^{i\varphi}$. 
The complex number $\s\in\Cnum$ lies in the disk $|\s|\leq 1$,
which represents the stereographic projections of the upper and lower hemispheres of the Riemann sphere into the complex plane attached to the north and south poles. 

The equivalence between the scalar NLS equation $u_t=-i(u_{xx} +2|u|^2u)$, 
the Schr\"odinger map equation into $S^2$,
and the Heisenberg spin model in $\mk{su}(2)\simeq\Rnum^3$,
as well as their integrability structures, 
is discussed in \Ref{AncMyr2010}.

\section{Concluding remarks}\label{sec:conclude}

This paper has presented a general theory 
extending the geometrical relationships among
the NLS equation, the vortex filament equation,
the Heisenberg spin model, and the Schr\"odinger map equation
to the setting of Hermitian symmetric spaces. 

The main underlying relation is provided by 
geometric non-stretching curve flows in
an associated affine Hermitian symmetric Lie algebra.
There is a hierarchy of curve flows having a bi-Hamiltonian structure
which takes a simple explicit form in terms of a Hasimoto variable
defined by the connection-matrix variable in a parallel moving frame,
as given by the general results on parallel frames and Hasimoto variables
obtained in \Ref{Anc2008}. 

To illustrate some of the main results,
we have considered one example:
complex projective space $\Cnum P^N= SU(N+1)/S(U(N)\times U(1))=SU(N+1)/U(N)$,
which corresponds to the Hermitian symmetric Lie algebra 
$\mk{su}(N+1)\simeq \Cnum^{N}\oplus \mk{u}(N)$. 
In a subsequent paper \cite{AsaAnc2019b},
we will look at all other classical Hermitian symmetric spaces:
$SU(N+2)/S(U(N)\times U(2))$,
$SO(N+2)/(SO(N)\times SO(2))$,
$SO(8)/U(4)$,
$Sp(2)/U(2)$. 

An interesting direction for generalizing this work
would be to study Miura transformations.
It is well-known that a Miura transformation exists between
the NLS equation and the vector NLS equation,
and that this relationship extends \cite{For1984}
to the isospectral NLS flows in general Hermitian symmetric spaces. 

Another interesting direction would be to study negative isospectral flows
and explore their relationship to group invariant Camassa-Holm type equations
in general Hermitian symmetric spaces.

\section*{Acknowledgements}
E.A. thanks S. Ali for pointing out useful comments. 

S.C.A. is supported by an NSERC grant.

\appendix
\section{Lie algebraic structure of complex projective space}
\label{sec:CP}

The complex projective space $\Cnum P^N$ is isomorphic to the Hermitian symmetric space 
$SU(N+1)/S(U(1)\times U(N))=SU(N+1)/U(N)$
which is of type $AIII(q=1)$ in the classification of Cartan \cite{Hel}. 
This space has the associated symmetric Lie algebra 
$\mk{su}(N+1) = \Cnum^N\oplus \mk{u}(N)$. 

The Lie brackets \eqref{symmsp.brackets} 
in the fundamental matrix representation \eqref{CP.g}--\eqref{CP.m}
are straightforwardly given by 
\begin{align}
[\hsp,\msp] & =[\lrep\AA\rrep,\lrep\a\rrep]
=\lrep -\tr(\AA)\a-\a\AA\rrep
\in\msp, 
\label{CP.lie.brack.h.m}\\
[\msp,\msp] &=[\lrep\a_1\rrep,\lrep\a_2\rrep]
=\lrep \a_2\wedgebar\a_1 \rrep
\in\hsp, 
\label{CP.lie.brack.m.m}\\
[\hsp,\hsp] &=[\lrep\AA_1\rrep,\lrep\AA_2\rrep]
=\lrep[\AA_1,\AA_2]\rrep
\in\hsp, 
\label{CP.lie.brack.h.h}
\end{align}
where
\begin{equation}
\lrep\AA\rrep := \bpm -\tr\AA& \mb{0}\\\mb{0}&\AA \epm 
\in\hsp \simeq \mk{u}(N), 
\quad
\lrep\a\rrep := \bpm 0&\a\\-\abar^\t & \mb{0} \epm ,
\in\msp \simeq \Cnum^N,
\end{equation}
and 
\begin{equation}
\lrep\AA,\a\rrep := \lrep\AA\rrep + \lrep\a\rrep 
= \bpm -\tr\AA& \a\\-\abar^\t&\AA \epm 
\in\gsp = \mk{su}(N+1), 
\end{equation}
with $\a\in\Cnum^N$, $\AA\in\mk{u}(N)$,
and where 
$\a_2\wedgebar\a_1=\abar_2^\t\a_1-\abar_1^\t\a_2 \in\mk{u}(N)$. 
For more details, see \Ref{AhmAncAsa2018}. 

From composition of the Lie brackets \eqref{CP.lie.brack.m.m} and \eqref{CP.lie.brack.h.m}, 
the following formula for the action of $\ad(\msp)^2$ on $\msp$ is obtained:
\begin{equation}\label{CP.adsq}
\ad(\lrep\a\rrep)^2\lrep\b\rrep
= \lrep (2\bbar\cdot\a -\abar\cdot\b)\a -(\abar\cdot\a)\b \rrep
= \lrep (\brack{\b,\a}+ 3i\brack{i\b,\a})\a -\brack{\a,\a}\b \rrep . 
\end{equation}
This result uses the relations
\begin{gather}
\bbar\cdot\a = \brack{\b,\a} +i \brack{i\b,\a} , 
\label{CP.dotprod.rel}
\\
\brack{\b,\a} =\re(\bbar\cdot\a),
\quad
\brack{i\b,\a} =\im(\bbar\cdot\a) , 
\end{gather}
where $\brack{\cdot,\cdot}$ denotes the inner product on $\Cnum^N$. 
The Killing form on $\msp$ is related to this inner product by 
\begin{equation}
\brack{\b,\a}_\msp = -\Kill{\b,\a} = -2(N+1)\tr(\lrep\a\rrep\lrep\b\rrep) 
= 4(N+1)\brack{\b,\a} . 
\end{equation}

The Hermitian structure of $SU(N+1)/S(U(1)\times U(N))=SU(N+1)/U(N)$ is 
given by the imaginary-unit element 
\begin{equation}\label{CP.A}
A = \lrep -\tfrac{1}{N+1}i I_N \rrep \in\mk{u}(N)=\hsp . 
\end{equation} 
This element acts as 
\begin{equation}\label{CP.J}
\ad(A)\gsp = \ad( \lrep -\tfrac{1}{N+1}i I_N \rrep )\lrep \AA,\a \rrep 
= \lrep \mb{0},i\a\rrep , 
\end{equation} 
whereby $J=\ad(A)$ is represented by multiplication by $i$ on $\msp$. 
The squared norm of $A$ is given by
\begin{equation} 
\efac= -\Kill{A,A} = -2(N+1)\tr(A^2) = 2N =\dim(\Cnum^{N}) . 
\end{equation} 

The embedding \eqref{AdG.adm.rel} of 
$\Cnum P^N \simeq SU(N+1)/S(U(1)\times U(N))=SU(N+1)/U(N)$
into the Lie algebra $\mk{su}(N+1) = \Cnum^N\oplus \mk{u}(N)$,
along with the corresponding $\mk{su}(N+1)$-matrix expression for $\q$, 
will now be derived. 
This is most easily carried out in two steps. 
First, a simple gauge choice $\psi_\cofra=\exp(\msp)$ will be used 
to evaluate $\map=\Ad(\psi_\cofra)A$, which is gauge invariant. 
Next, a gauge transformation \eqref{AdG.gauge} will be applied to impose 
$(\psi_\cofra^{-1}\partial_x\psi_\cofra)_\hsp =0$ on $\psi_\cofra$, 
yielding $\q=\psi_\cofra^{-1}\partial_x\psi_\cofra$. 

For the first step, 
consider the matrix 
$\tilde\psi=\exp(\msp) = \sum_{k=0}^{\infty} \tfrac{1}{k!}\lrep \a\rrep^k$, 
where it will be simplest to evaluate the even and odd power terms separately:
\begin{equation}
\lrep \a\rrep^k = 
\begin{cases}
i^{k-1} |\a|^{k-1}\lrep \a\rrep, & 
k=1,3,5,\ldots
\\
-i^{k} 2^{k-1} |\a|^{k-2}\lrep \abar^\t\a \rrep, & 
k=2,4,6,\ldots
\end{cases}
\end{equation}
with $|\a|^2 = \brack{\a,\a}$. 
Hence, 
\begin{equation}
\sum_{l\geq 1} \tfrac{1}{(2l-1)!}\lrep \a\rrep^{2l-1} 
= \sin(|\a|) \bpm 0 & \hat\a \\ -\bar{\hat\a}^\t & \mb{0} \epm, 
\quad
\sum_{l\geq 1} \tfrac{1}{(2l)!}\lrep \a\rrep^{2l} 
= (\cos(|\a|)-1) \bpm 1 & 0 \\ 0 & \bar{\hat\a}^\t\hat\a \epm, 
\end{equation}
with $\hat\a = \frac{1}{|a|}\a$, 
which yields
\begin{equation}\label{CP.expm}
\tilde\psi= \exp(\msp) = 
\bpm 0 & 0 \\ 0 & I_N- \bar{\hat\a}^\t\hat\a \epm
+ \sin(|\a|) \bpm 0 & \hat\a \\ -\bar{\hat\a}^\t & \mb{0} \epm 
+ \cos(|\a|) \bpm 1 & 0 \\ 0 & \bar{\hat\a}^\t\hat\a \epm
\in SU(N+1) . 
\end{equation}

\begin{prop}
The orbit of $\Ad(SU(N+1)) A$ in $\mk{su}(N+1)$
is given by the matrix
\begin{equation}\label{embed.CP}
\map = \Ad(\tilde\psi)A
= -\tfrac{1}{N+1} \lrep i I_N \rrep + \tfrac{1}{2} (1-\cos(2\theta)) \lrep i\bar\Theta^\t\Theta \rrep 
- \tfrac{1}{2} \sin(2\theta) \lrep i\Theta \rrep
\in SU(N+1)
\end{equation}
with 
\begin{equation}
\theta\in [0,\pi],
\quad
\Theta\in\Cnum^N,
\quad
\brack{\Theta,\Theta}=1,
\end{equation}
where 
$\lrep i I_N \rrep,\lrep i\bar\Theta^\t\Theta \rrep\in\hsp\simeq\mk{u}(N)$, 
$\lrep i\Theta \rrep\in\msp\simeq\Cnum^N$. 
This unitary matrix \eqref{embed.CP} represents 
the isometric embedding of $\Cnum P^N \simeq \frac{SU(N+1)}{S(U(1)\times U(N))}=\frac{SU(N+1)}{U(N)}$
into the Lie algebra $\mk{su}(N+1) = \Cnum^N\oplus \mk{u}(N)$. 
\end{prop}

The next step consists of evaluating $\tilde\psi\inv \partial_x\tilde\psi$
using the matrix \eqref{CP.expm}. 
This yields, by a direct computation, 
\begin{align}
(\tilde\psi\inv \partial_x\tilde\psi)_\msp
& = \theta_x\, \lrep \Theta \rrep 
+\sin(\theta) \lrep \Theta_x \rrep 
+ \sin(\theta)(1-\cos(\theta)) \phi\,\lrep i\Theta \rrep, 
\\
(\tilde\psi\inv \partial_x\tilde\psi)_\hsp
& = (1-\cos(\theta)) \lrep i\mb{\Phi} \rrep 
+ (1-\cos(\theta))^2 \phi\,\lrep i\bar\Theta^\t\Theta \rrep , 
\end{align}
with 
\begin{equation}
\phi=\brack{i\Theta_x,\Theta}\in\Rnum,
\quad
i\mb{\Phi}=\Theta\wedgebar\Theta_x \in\mk{u}(N),
\quad
\tr\mb{\Phi} = -2\phi . 
\end{equation}
Now consider a gauge transformation 
$\tilde\psi\rightarrow \tilde\psi\g(x) =\psi_\cofra$,
where 
\begin{equation}
\g = \bpm e^{i\Gamma} & 0 \\ 0 & e^{-\frac{1}{N}i\Gamma}\mb{\Gamma} \epm \in H = S(U(1)\times U(N)) \simeq U(N),
\quad
\Gamma\in\Rnum,
\quad
\mb{\Gamma}\in SU(N) . 
\end{equation}
This transformation yields
\begin{align}
&\begin{aligned}
(\psi_\cofra\inv \partial_x\psi_\cofra)_\msp 
& = 
\theta_x\, \lrep e^{i(1+\frac{1}{N})\Gamma}\Theta\mb{\Gamma} \rrep 
+\sin(\theta) \lrep e^{i(1+\frac{1}{N})\Gamma}\Theta_x\mb{\Gamma} \rrep 
\\&\qquad
+ \sin(\theta)(1-\cos(\theta)) \phi\,\lrep ie^{i(1+\frac{1}{N})\Gamma}\Theta\mb{\Gamma} \rrep , 
\end{aligned}
\\
&\begin{aligned}
(\psi_\cofra\inv \partial_x\psi_\cofra)_\hsp 
& = 
(1-\cos(\theta)) \lrep i\Ad(\mb{\Gamma}\inv)\mb{\Phi} \rrep 
+ (1-\cos(\theta))^2 \phi\,\lrep i\Ad(\mb{\Gamma}\inv)\bar\Theta^\t\Theta \rrep 
\\&\qquad
+ \lrep \tfrac{-1}{N}i\Gamma_x I_N + \mb{\Gamma}\inv\mb{\Gamma}_x \rrep . 
\end{aligned}
\end{align}
Then the gauge condition $(\psi_\cofra^{-1}\partial_x\psi_\cofra)_\hsp =0$ gives 
a matrix differential equation. 
Its trace and trace-free parts yield the system 
\begin{equation}\label{CP.gauge}
\begin{gathered}
0 = \Gamma_x +\sin^2(\theta) \phi , 
\\
\mb{0} = 
\mb{\Gamma}_x 
+ (1-\cos(\theta)) i( \mb{\Phi} + (1-\cos(\theta)) \phi \bar\Theta^\t\Theta )\mb{\Gamma} , 
\end{gathered}
\end{equation}
for $\Gamma(x)\in\Rnum$, $\mb{\Gamma}(x)\in SU(N)$. 
Clearly, this system of differential equations \eqref{CP.gauge} 
will have a unique solution (locally in $x$) 
with any given initial data $\Gamma(x_0)\in\Rnum$, $\mb{\Gamma}(x_0)\in SU(N)$
specified at an arbitrary $x=x_0$. 
As a result, the following expression is obtained for $\q$. 

\begin{prop}
The Hasimoto variable $\q\in \msp\simeq \Cnum^N$
for $\Cnum P^N \simeq \frac{SU(N+1)}{S(U(1)\times U(N))}=\frac{SU(N+1)}{U(N)}$ 
is given by the matrix 
\begin{equation}\label{CP.q}
\begin{aligned}
\q = \psi_\cofra\inv \partial_x\psi_\cofra 
& = \theta_x\, \lrep e^{i(1+\frac{1}{N})\Gamma}\Theta\mb{\Gamma} \rrep 
+\sin(\theta) \lrep e^{i(1+\frac{1}{N})\Gamma}\Theta_x\mb{\Gamma} \rrep 
\\&\qquad
+ \sin(\theta)(1-\cos(\theta)) \phi\,\lrep ie^{i(1+\frac{1}{N})\Gamma}\Theta\mb{\Gamma} \rrep , 
\end{aligned}
\end{equation}
which is determined uniquely in terms of $\theta(x)$ and $\Theta(x)$
through the differential equations \eqref{CP.gauge}, 
up to a rigid ($x$-independent) gauge transformation 
$\q\rightarrow \Ad(U(N))\q$. 
\end{prop}

\end{document}